\documentclass[%
reprint,
superscriptaddress,
nofootinbib,
aps,longbibliography,
pra,
amsmath,amssymb
]{revtex4-1}

\usepackage[toc,page]{appendix} %for appendix/SI numbering

\newif\ifarxiv % Comment out the next line for journal-style version
\arxivtrue

\usepackage{array} 
\usepackage{xparse}
\usepackage[english]{babel}
\usepackage[utf8]{inputenc}
\usepackage{setspace}
\usepackage{pgfplots}
\usepgfplotslibrary{fillbetween} 
\pgfplotsset{compat=newest}
\usepackage{graphicx}
\graphicspath{ {./figures/} }
\usepackage[caption=false]{subfig} %Allows some subfigures while suppressing the caption effects that clash with revtex
\usepackage{amsmath}
\usepackage{amsthm}
\usepackage{amssymb}
\usepackage{braket}
\usepackage{dsfont}
\usepackage{xcolor}
\usepackage{mathrsfs} %For \mathscr
\usepackage{mathtools} %For \coloneqq
\usepackage{tikz}
\usetikzlibrary{calc}
\usetikzlibrary{arrows}
\usepackage{tcolorbox}
\usepackage{enumitem}
\newlist{steps}{enumerate}{1}
\setlist[steps, 1]{label = Step \arabic*:}
\usepackage[colorinlistoftodos,prependcaption]{todonotes}
\usepackage{booktabs}
\usepackage{thmtools,thm-restate}
\usepackage{hyperref} %Should be loaded after most packages unless otherwise specified
\usepackage{xurl} %Allow URL linebreaks on arXiv
\hypersetup{
	colorlinks = true,
	citecolor= blue,
	urlcolor= blue,
	linkcolor = blue
}
\usepackage{algorithm}%To be loaded after hyperref
\usepackage[noend]{algpseudocode}
\newtheorem{theorem}{Theorem}

\newtheorem{thm}[theorem]{Theorem}

\newtheorem{rem}[theorem]{Remark}

\newtheorem{lem}[theorem]{Lemma}
\newtheorem{prop}[theorem]{Proposition}

\theoremstyle{definition} %All \newtheorems below this point will use this theoremstyle
\newtheorem{defn}[theorem]{Definition}

\newcounter{algsubstate}

\newenvironment{algsubstates}
{\setcounter{algsubstate}{0}%
	\renewcommand{\State}{%
		\refstepcounter{algsubstate}%
		\Statex {\normalsize\arabic{ALG@line}.\arabic{algsubstate}.}\kern5pt}%Fixed-width space
}
{}

\newcommand{\ketbra}[2]{\ket{#1}\!\bra{#2}}
\newcommand{\vect}[1]{\mathbf{#1}}

\newcommand\norm[1]{\lVert#1\rVert}

\newcommand{\eps}{\epsilon} %Shortened epsilon command
\newcommand{\freq}{\operatorname{freq}}
\newcommand{\term}[1]{\textup{\textbf{#1}}} %Formatting for new terms
\newcommand{\Renyi}{R\'{e}nyi}
\newcommand{\tr}[1]{\operatorname{Tr}\left[#1\right]}
\newcommand{\ptr}[2]{\operatorname{Tr}_{#1}\left[#2\right]}
\newcommand{\id}{\mathbb{I}}
\newcommand{\score}{S}
\newcommand{\scalar}{x}
\newcommand{\scoreupper}{S}
\let\inner\relax
\NewDocumentCommand\inner{mg}{%
	\ensuremath{\left\langle #1 | \IfNoValueTF{#2}{#1}{#2}\right\rangle}%
}
\newcommand{\binh}{h_\mathrm{bin}} %Binary entropy
\newcommand{\strat}{\Lambda}
\newcommand{\Pos}{\mathrm{Pos}}
\newcommand{\States}{S} % Maybe too many S's already?

\newcommand{\OAT}{\Omega_{\mathrm{AT}}}
\newcommand{\OEV}{\Omega_{\mathrm{EV}}}

\usepackage[T3,T1]{fontenc}
\DeclareSymbolFont{tipa}{T3}{cmr}{m}{n}
\DeclareMathAccent{\invbreve}{\mathalpha}{tipa}{16}
\allowdisplaybreaks

\begin{document}
	
	%%%%%% Title %%%%%%
	% Full titles can be a maximum of 200 characters, including spaces. 
	% Title Format: Use title case, capitalizing the first letter of each word, except for certain small words, such as articles and short prepositions
	\title{Analytic Rényi Entropy Bounds for Device-Independent Cryptography}
	
	%%%%%% Authors %%%%%%
	% Authors should be listed in order of contribution to the paper, by first name, then middle initial (if any), followed by last name.
	% Authors should be listed in the order in which they will appear in the published version if the manuscript is accepted. 
	% Use an asterisk (*) to identify the corresponding author, and be sure to include that person’s e-mail address. Use symbols (in this order: †, ‡, §, ||, ¶, #, ††, ‡‡, etc.) for author notes, such as present addresses, “These authors contributed equally to this work” notations, and similar information.
	% You can include group authors, but please include a list of the actual authors (the group members) in the Supplementary Materials.
	\author{Thomas A. Hahn}
	\affiliation{The Center for Quantum Science and Technology, Department of Physics of Complex Systems, Weizmann Institute of Science, Rehovot, Israel}
    \author{Aby Philip}
    \affiliation{Institute of Fundamental Technological Research, Polish Academy of Sciences, Pawi\'{n}skiego 5B, 02-106 Warsaw, Poland}
	\author{Ernest Y.-Z. Tan}
	\affiliation{
		% Institute for Quantum Computing and Department of Physics and Astronomy, 
		University of Waterloo, Waterloo, Ontario N2L 3G1, Canada}
	\author{Peter Brown}
	\affiliation{T\'{e}l\'{e}com Paris, LTCI, Institut Polytechnique de Paris, Inria, 19 Place Marguerite Perey, 91120 Palaiseau, France}
	\date{\today}

	\begin{abstract} 
        Device-independent (DI) cryptography represents the highest level of security, enabling cryptographic primitives to be executed safely on uncharacterized devices. Moreover, with successful proof-of-concept demonstrations in randomness expansion, randomness amplification, and quantum key distribution, the field is steadily advancing toward commercial viability.
        Critical to this continued progression is the development of tighter finite-size security proofs. In this work, we provide a simple method to obtain tighter finite-size security proofs for protocols based on the CHSH game, which is the nonlocality test used in all of the proof-of-concept experiments. We achieve this by analytically solving key-rate optimization problems based on Rényi entropies, providing a simple method to obtain tighter finite-size key rates. 

	\end{abstract}
	
	\maketitle

	%%%%%% Main Text %%%%%%

	\ifarxiv\section{Introduction}\else\noindent{\it Introduction.|}\fi
	Quantum theory exhibits certain correlations between distant agents that cannot be explained classically~\cite{Bell}. These so-called nonlocal correlations have far-reaching implications beyond fundamental physics. In particular, by observing certain nonlocal correlations, it is possible to make statements about the underlying quantum systems used to produce them. This leads to device-independent information processing, which guarantees the successful completion of information processing tasks without detailed characterization of the underlying hardware. Such protocols have already been demonstrated in practice, with recent implementations of DI randomness expansion~\cite{dirng_expr1,dirng_expr2}, amplification~\cite{dira_expr} and quantum key distribution (QKD)~\cite{diqkd_expr1,diqkd_expr2,diqkd_expr3}. 
	
	DIQKD protocols provide a method to generate shared secret key whilst relying on minimal assumptions. Here, one can leverage the fact that if two honest parties, Alice and Bob, observe a shared nonlocal distribution, then they can information-theoretically verify that their outputs are random from the perspective of an eavesdropper, Eve. There has been significant effort in designing protocols that are secure and efficient in practice~\cite{di_security_review1,di_security_review2}, and whilst proof-of-principle demonstrations of DIQKD have been achieved, the current achievable rates remain far from practical.
	
	Protocols based on the CHSH inequality~\cite{CHSH} have been studied extensively for DIQKD, owing to its simplicity both theoretically and experimentally. In particular, when considering security against collective attacks,~\cite{pironio2009deviceindependent} provides a tight analytical lower bound on the von Neumann entropy (and hence the asymptotic key-rate) in terms of the expected CHSH violation. Applying techniques like the Entropy Accumulation Theorem (EAT)~\cite{DFR20,DF19,MFSR22}, one can elevate these bounds to finite-size key rates against general adversaries. Recently, a new {\Renyi} EAT (REAT) was proven in~\cite{arx_AHT24}, which is based on the more general family of entropies known as {\Renyi} entropies, and has the potential to yield tighter finite-size key-rates resulting in more practical protocols. However, in order to reap these benefits, tight and efficient methods for bounding Rényi entropies in a device-independent manner must be developed.
	
	In this work, we derive a tight analytical relationship between the expected CHSH value and the amount of Rényi entropy in the output. This can be seen as a broad generalization of the expression derived in~\cite{pironio2009deviceindependent}, which we recover as a special case. Crucially, our results unlock the potential of~\cite{arx_AHT24} and we demonstrate significantly improved finite-size key-rates for protocols based on the CHSH inequality. We further discuss how our results can be modified to include noisy preprocessing~\cite{HSTRBS20} and generalized to the asymmetric CHSH inequalities~\cite{Woodhead_2021}, further boosting the key rates and applicability of the technique. Overall, our work pushes DIQKD towards a new level of practicality.
	
	\ifarxiv\section{Analytic Rényi entropy bounds}\else\noindent{\it Analytic key-rate bounds.|}\fi 
	We now present the main technical contribution of this work, which is a tight analytical relationship between the CHSH value and the accumulated Rényi entropy. To more precisely state the problem at hand, consider the following setup. We have two honest parties Alice and Bob, and an eavesdropper Eve. Alice and Bob each hold a device which can receive binary inputs $X,Y$ and produce binary outputs $A,B$ respectively. The behavior of the boxes may be modeled in the following way: a shared tripartite state $\rho_{Q_AQ_BE}$ is distributed to Alice, Bob and Eve; upon receiving the input $X=x$, Alice's box measures $Q_A$ with a POVM $\{M_{a}^x\}_a$ and outputs the measurement outcome $A=a$; upon receiving the input $Y=y$, Bob's box measures $Q_B$ with the POVM $\{N_{b}^y\}_b$ and outputs the measurement outcome $B=b$. Given inputs $X=x$, $Y=y$, we can compute a post-measurement state $\rho_{ABE}^{xy} = \sum_{ab} \ketbra{ab}{ab}_{AB} \otimes \rho_E^{abxy}$, where
	\begin{align}
		\rho_E^{abxy} = \ptr{Q_AQ_B}{\rho_{Q_AQ_BE} (M_a^x \otimes N_b^y \otimes \id_E)}\,.
	\end{align}
	We call a tuple $(Q_AQ_BE, \rho_{Q_AQ_BE}, \{M_a^x\}, \{N_b^y\})$ of Hilbert spaces, a shared state, and POVMs a \emph{quantum strategy}.

	From the perspective of the honest parties, Alice and Bob, their devices are black boxes. As such, the exact Hilbert spaces, shared state, and POVMs used are unknown to them. On the other hand, we allow the eavesdropper Eve to have full control over the implementation of the devices, i.e., she can choose the quantum strategy. Despite the black-box nature of their devices, Alice and Bob are able to learn information about the correlations produced by their devices. In this work, we shall  
    % assume that they can measure 
    focus on protocols based on estimating
    the expected CHSH value of their devices, which is defined as \begin{align}\label{eq:expected_CHSH_value}
		\score = \sum_{abxy} (-1)^{xy + a + b} \tr{\rho_{Q_AQ_BE} \left(M_a^x \otimes N_b^y \otimes \mathds{1}\right)}  \; .
	\end{align} 
	For any $\score > 2$, the correlations produced by the devices are nonlocal and can be used to produce device-independent randomness. Moreover, the maximal achievable value using quantum systems is $\score = 2 \sqrt{2}$~\cite{tsirelson_bound}. Given a quantum strategy ${\strat = (Q_AQ_BE, \rho_{Q_AQ_BE}, \{M_{a|x}\}, \{N_{b|y}\})}$ we denote its expected CHSH value, computed using~Eq.\ \eqref{eq:expected_CHSH_value}, by $\scoreupper_{\mathrm{CHSH}}(\strat)$.
	
	In this section we are interested in solving the following optimization problem. Given a conditional entropy $\mathbb{H}$ (see~\ifarxiv Appendix~\ref{app:definitions} \else the {\it Supplemental Material }\fi for the formal definitions) and an expected CHSH value $\score = [2, 2\sqrt{2}]$, find the minimal value of $\mathbb{H}(A|X=0, E)$ over all possible quantum strategies that have an expected CHSH value of $\score$. In other words, we are interested in computing an $\mathbb{H}$ rate function for the CHSH Bell-inequality, as made precise in the following definition. 
	
	\begin{defn}[$\mathbb{H}$ rate function for CHSH]
		Let $\mathbb{H}$ be a conditional entropy and let $\score \in [2, 2 \sqrt{2}]$. We say that a function $f_{\mathbb{H}}: [2, 2 \sqrt{2}] \to \mathbb{R}$ is a tight $\mathbb{H}$ rate function for the CHSH Bell inequality if 
		\begin{equation} \label{Eq: DefHTradeoffFunction}
			\begin{aligned}
				f_{\mathbb{H}}(\score)                 \coloneqq\inf_{\strat}& \quad \mathbb{H}(A | X=0, E) \\
				\mathrm{s.t.}& \quad \scoreupper_{\operatorname{CHSH}}(\strat) = \score \,,  
			\end{aligned}
		\end{equation} 
		where the infimum is over all quantum strategies $\strat$.
	\end{defn}

	We note that the case of the von Neumann entropy, $H$, was solved in~\cite{pironio2009deviceindependent} (see also~\cite{Arnon-Friedman:2018aa}), where it was shown that
	\begin{equation}\label{eq:f_H}
		f_{H}(\score) = 1 - h\left(\frac12 + \frac12 \sqrt{\frac{\score^2}{4} - 1}\right)\,,  
	\end{equation}
	where $h(x) \coloneqq -x \log_2 x- (1-x)\log_2(1-x)$ is the binary entropy. Similarly, in~\cite{MPA11} it was shown that for the min-entropy, $H_{\min}$, one has
	\begin{equation}\label{eq:f_Hmin}
		f_{H_{\min}}(\score) = 1 - \log\left(1 + \sqrt{2 - \frac{\score^2}{4}}\right)\,.
	\end{equation}
	
	Our main technical result is an exact analytical form of the rate functions for multiple major families of {\Renyi} conditional entropy. In Theorem~\ref{Theorem: SandwichBound} below, we focus on presenting our results for 
    two particular families of ``sandwiched Rényi entropies'' (see~\ifarxiv Appendix~\ref{app:definitions} \else the {\it Supplemental Material }\fi or~\cite{MDS+13,WWY14} for full details), denoted as $\widetilde{H}^{\uparrow}_{\alpha}$ and $\widetilde{H}^{\downarrow}_{\alpha}$. We focus on these for now as they are the most relevant entropies for our finite-size analysis; however, we highlight that significant generalizations of this theorem are also possible. For instance, as we show in~\ifarxiv Appendix~\ref{app:generalisations}\else the {\it Supplemental Material}\fi, we can extend it to the family of asymmetric CHSH inequalities~\cite{Woodhead_2021}, modify $f_{\widetilde{H}^{\downarrow}_\alpha}(\score)$ to include noisy preprocessing~\cite{HSTRBS20}, and derive analogous results for the Petz-Rényi entropies~\cite{Petz}. The generalizations that include noisy preprocessing and the asymmetric CHSH inequalities have the potential to boost the achievable key-rates for DIQKD even further. 
	\begin{thm}\label{Theorem: SandwichBound}
		Let $\alpha > 1$ and $\score \in [2, 2\sqrt{2}]$. Then we have 
	\begin{align}
        f_{\widetilde{H}^{\uparrow}_\alpha}(\score) &= 1 + \frac{2\alpha - 1}{1-\alpha}\log \phi_{\frac{\alpha}{2\alpha-1}}(S) \, ,\\
			f_{\widetilde{H}^{\downarrow}_\alpha}(\score) &= 1 + \frac{\alpha}{1-\alpha}\log \phi_{\frac{1}{\alpha}}(S) \, ,
		\end{align}
        where
        \begin{equation}
                    \phi_\mu(S) = \left(\frac{1-\sqrt{\tfrac{S^2}{4} -1}}{2} \right)^\mu \!+ \left(\frac{1+\sqrt{\tfrac{S^2}{4}-1}}{2} \right)^\mu \, .
        \end{equation}  
	\end{thm}
    We provide the proof in~\ifarxiv Appendix~\ref{app:generalisations}\else the {\it Supplemental Material }\fi, together with the generalizations.
	As expected, in the limit $\alpha \to 1$, both $f_{\widetilde{H}^{\uparrow}_\alpha}$ and ${f_{\widetilde{H}^{\downarrow}_\alpha}}$ converge to {Eq.~\eqref{eq:f_H}}. Moreover, by setting $\alpha = 2$, one recovers the fact that $f_{\widetilde{H}^{\downarrow}_2}$ equals the rate function for the min-entropy in {Eq.~\eqref{eq:f_Hmin}}, a result that was first observed in~\cite{murta2018realization}. 
    
    Interestingly, the optimal strategy for Eve that achieves the infimum in Eq.~\eqref{Eq: DefHTradeoffFunction} is the same for all $\alpha \geq  1$ and both entropy families. In particular, for an expected CHSH value, $\score$, the optimal strategy for Eve is to program Alice and Bob's devices to measure the observables 
    \begin{align}
        A_0 = \sigma_{z}, \ A_1 = \sigma_{x}, \ B_{0,1} = \frac{\sigma_{z} \pm  g_{\score}\sigma_{x}}{\sqrt{1+g_{\score}^2}} \, ,
    \end{align}
    % $A_0 = \sigma_{z}$, $A_1 = \sigma_{x}$, $B_0 = \frac{\sigma_{z} + g_{\score}\sigma_{x}}{\sqrt{1+g_{\score}^2}}$, $B_1 = \frac{\sigma_{z} - g_{\score}\sigma_{x}}{\sqrt{1+g_{\score}^2}}$, 
 where  $g_{\score}=\sqrt{\frac{\score^2}{4} -1}$, on the state	
        \begin{equation}    \label{Eq: AttackSaturatesBound}  \sqrt{P_+}\ket{\phi^+}_{Q_AQ_B}\ket{0}_{E}+\sqrt{P_-}\ket{\phi^-}_{Q_AQ_B}\ket{1}_{E} \, ,
		\end{equation}
		for  ${\ket{\phi^\pm} = \frac{1}{\sqrt{2}} \left(\ket{00} \pm\ket{11}\right)}$, and $P_\pm = \frac{1}{2}\left(1 \pm g_{\score}\right)$.

	\ifarxiv\section{Improved finite-size Key rates}\else\noindent{\it Improved finite-size DIQKD rates.|}\fi \label{Section: GREATBounds}
    \ifarxiv\subsection{Protocol Description}\else\noindent{\it Protocol Description.|}\fi \label{Section: Protocol} 
    In addition to the CHSH set-up from the previous section, for DIQKD one generally allows Bob an extra measurement input, $Y=2$, which he uses to generate secret key (his other settings are used to test the device). This additional measurement improves the key rates by reducing the cost of error correction, and due to non-signaling conditions, will not affect the validity of the rate functions described in the previous section, see e.g.~\cite{pironio2009deviceindependent}. This extended set-up is what will be considered here.

Each round of the DIQKD protocol is either a test or a key-generation round. During test rounds, which occur with some probability $\gamma \in [0,1]$, each party chooses their measurement inputs $X,Y$ uniformly at random from the set $\{0,1\}$ and  generates some measurement outputs $A,B$. During a subsequent public announcement step, Bob will announce a value $\bar{B}$ that is set equal to his output value $B$ whenever it is a test round, which allows Alice to produce an additional test-data value for each round, $\bar{C}$, as follows:
	\begin{align}\label{eq:test-function}
		\bar{C}=\begin{cases}
			0, & \text{if }A \oplus B \neq X \cdot Y  \\
			1, & \text{if }A \oplus B = X \cdot Y \; .
		\end{cases}  
\end{align}
For single-round quantum strategies, the distribution of $\bar{C}$ can directly be related to the corresponding CHSH value, $\score$. As such, the test data encodes the relevant information needed for applying Theorem \ref{Theorem: SandwichBound}. During generation rounds, Alice and Bob use the inputs ${X=0}$ and $Y=2$, respectively. No test data is produced and so Bob sets his public announcement $\bar{B}$ to some arbitrary value (say, $\bar{B}=0$), while Alice sets $\bar{C}$ to a special symbol $\perp$. 

After $n$ rounds, both parties conduct two further classical postprocessing steps. First, they verify that the distribution of the test data lies within some predetermined set of probability distributions, $S_\Omega$. By aborting the protocol whenever the observed distribution lies outside $S_\Omega$, this step essentially ensures that Alice's and Bob's measurements produce (at least) weakly random bits, see e.g.~\cite{Arnon-Friedman:2018aa,TSB+22}.
%initial shared state was entangled, see e.g.\ []. 
Afterwards, the protocol concludes with a post-processing step, which converts Alice's and Bob's raw measurement data, into (almost) ideal states according to a suitable DIQKD security definition~\cite{Arnon-Friedman:2018aa,TSB+22}.

For clarity, we present an overall summary of this procedure as Protocol~\ref{Table: Protocol} below. More information regarding the classical post-processing steps can be found in~\ifarxiv Appendix~\ref{app:finitesize}\else the {\it Supplemental Material}\fi. Note that this protocol structure can also accommodate DI randomness expansion with only small changes, mostly in these classical post-processing steps; see e.g.~\cite{hahn2024boundspetzrenyidivergencesapplications} for such a description.
\begin{algorithm}[H]
    \floatname{algorithm}{}
    \caption{Cryptographic Protocol} \label{Table: Protocol}
    \begin{algorithmic}[1] 
        \State For all rounds, $i \in \{1,\dots,n \}$:
        \begin{algsubstates}
            \State Alice and Bob generate a common random bit $T_i$, such that $\operatorname{P}\left(T_i=0\right) = 1-\gamma$ and $\operatorname{P}\left(T_i=1\right) = \gamma$.
            \State If $T_i=0$, both parties will choose generation inputs $\left(X_i,Y_i\right) = \left(0,2\right)$. If $T_i=1$, both parties will choose test inputs $\left(X_i,Y_i\right) \in \{0,1\}^{ 2}$ uniformly at random. They supply the inputs to the devices and obtain outputs $(A_i,B_i)$.
        \end{algsubstates}
        \State Public announcements: Both parties announce all the input values $(X_i , Y_i)$. Bob also announces some other registers $\bar{B}_i$ for $i \in \left[n\right]$ as follows. If $T_i=0$, Bob announces $\bar{B}_i = 0$ and Alice sets $ \bar{C}_i = \perp$. If $T_i=1$, Bob announces $\bar{B}_i = B_i$, then Alice computes $\bar{C}_i$ according to the specified function of $\left(A_i,B_i,X_i,Y_i\right)$. 
        \State Acceptance test: Alice checks if the observed frequency distribution $\freq_{\bar{c}_1^n}$ lies inside some predetermined set $S_\Omega$, and aborts the protocol (via a public announcement) if it does not.
        \State Classical postprocessing: Alice and Bob perform some additional classical operations such as error correction and privacy amplification (see~\ifarxiv Appendix~\ref{app:finitesize}\else the {\it Supplemental Material}\fi \  for details) to generate their final keys.
    \end{algorithmic}
    \label{alg:OriginalProtocol}
\end{algorithm}
\ifarxiv\subsection{{\Renyi} EAT and Key Rate}\else\noindent{\it {\Renyi} EAT and Key Rate.|}\fi     

\begin{figure*}
% \centering
\subfloat[\label{subfig:smalln}]{%
  \includegraphics[width=0.45\linewidth]{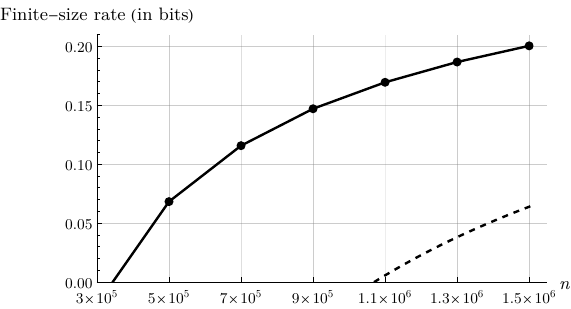}%
}
\subfloat[\label{subfig:largen}]{%
  \includegraphics[width=0.45\linewidth]{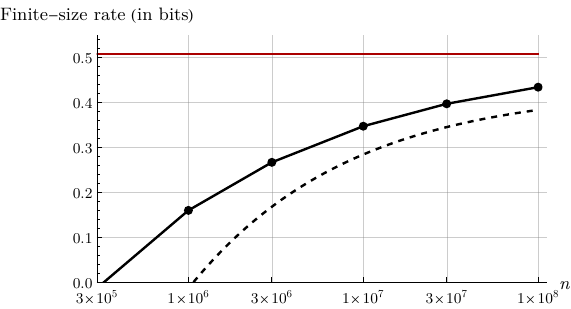}%
}
\caption{Plot of achievable finite-size key rates for the DIQKD experimental demonstration in~\cite{NDN+22}, as a function of number of rounds $n$ in two different ranges. The solid black curves show the results we obtain from our approach, the dashed black curves show the results previously computed in~\cite{NDN+22}, and the solid red line displays the asymptotic rate (omitted in Fig.~\ref{subfig:smalln} due to the lower key rates in that plot). Note that we have kept the protocol and parameter choices (other than the final key length) nearly identical to the one in~\cite{NDN+22}, except for minor changes and improvements we describe in~\ifarxiv Appendix~\ref{app:finitesize}\else the {\it Supplemental Material}\fi. We see there is a significant improvement in both the finite-size rate at the number of rounds used in that experiment ($n=1.5\times10^6$), and the minimum $n$ required to achieve nonzero finite-size rates. }
\label{fig:finitesize}
\end{figure*}

We provide a more detailed description of the finite-size security proof in~\ifarxiv Appendix~\ref{app:finitesize}\else the {\it Supplemental Material}\fi; here, we just outline the key steps of the proof. 
The global $n$-round state in the protocol after the public announcements is of the form $\rho_{{A}_{1}^{n} {B}_{1}^{n} \bar{{C}}_{1}^{n}{X}_{1}^{n}{Y}_{1}^{n}{T}_{1}^{n}\mathbf{E}}$, with $\mathbf{E}$ denoting quantum side-information Eve holds about the states in the devices (she also has access to the public announcements ${X}_{1}^{n}{Y}_{1}^{n}{T}_{1}^{n}$). 
Using the REAT, it was shown in~\cite[Lemmas~5.1~and~6.1]{arx_AHT24} that as long as the set $S_\Omega$ is convex, the total accumulated {\Renyi} entropy (conditioned on the acceptance test accepting, which we shall denote as the event $\OAT$) can be bounded by
    \begin{align}  \label{Eq: GREATBound}   \widetilde{H}_{\alpha}^{\uparrow}\left({A}_{1}^{n}\bar{{C}}_{1}^{n}|
{X}_{1}^{n}{Y}_{1}^{n}{T}_{1}^{n}\mathbf{E}\right)_{\rho_{|\OAT}} \geq n h_{\alpha} 
- \frac{\alpha}{\alpha-1} \log\frac{1}{\Pr[\OAT]} \, ,
    \end{align}
   where $h_{\alpha}$ is a quantity satisfying
   \begin{align}\label{eq:singlerndopt}
     h_{\alpha} \geq \inf_{\strat} \inf_{\vect{q} \in S_\mathrm{acc}}   \frac{1}{\alpha-1}D\left(
     \vect{q} \middle\Vert \vect{p}_\strat
     \right) + q(\perp)\widetilde{H}_{\alpha}^{\downarrow}\left( A|X=0,E\right) \, .
   \end{align}
Here, the optimization takes place over all single-round quantum strategies $\strat$, and probability distributions $\vect{q}$ (on a single-round test-data register $\bar{C}$) within the acceptance set $S_\Omega$. For each single-round strategy,  $\widetilde{H}_{\alpha}^{\downarrow}\left( A|X=0,E\right)$ refers to the corresponding {\Renyi} entropy of the state produced from that strategy, and the $D(\vect{q} \Vert \vect{p}_\strat)$ term denotes the Kullback-Leibler (KL) divergence (see e.g.~\cite{Cover}) between $\vect{q}$ and the distribution $\vect{p}_\strat$ produced by that quantum strategy.
   
Qualitatively, the above optimization has an intuitive informal interpretation, as follows. Observe that if the device behavior across the rounds were independent and identically distributed (IID), then in the asymptotic large-$n$ limit, the best bound we could hope for on the global entropy $\widetilde{H}_{\alpha}^{\uparrow}\left({A}_{1}^{n}\bar{{C}}_{1}^{n}|
{X}_{1}^{n}{Y}_{1}^{n}{T}_{1}^{n}\mathbf{E}\right)$ would be simply $n$ times of the minimal single-round entropy $\widetilde{H}_{\alpha}^{\downarrow}\left( A|X=0,E\right)$ over all strategies $\strat$ ``compatible with'' the accept condition $S_\Omega$ (more formally, such that the distribution $\vect{p}_\strat$ lies in $S_\Omega$).
The above optimization is similar in spirit to computing this minimal value, except that rather than the ``hard'' constraint of requiring $\vect{p}_\strat \in S_\Omega$, the KL divergence term serves to impose a ``soft'' version of this constraint, in that it acts as a ``penalty'' if $\vect{p}_\strat$ is far from $S_\Omega$ --- see~\cite[Sec.~5.2]{arx_AHT24} for more detailed exposition. We emphasize however that while this intuitive interpretation is informal, the bounds~\eqref{Eq: GREATBound}--\eqref{eq:singlerndopt} constitute a \emph{rigorous} lower bound on the global {\Renyi} entropy, against general (non-IID) attacks. 

In~\ifarxiv Appendix~\ref{app:finitesize}\else the {\it Supplemental Material}\fi, we explain how this bound can be used to compute finite-size key rates. Essentially, the LHS of Eq.\ \eqref{Eq: GREATBound} describes the {\Renyi} entropy that Alice generates over all rounds of the protocol (including the test data), conditioned on the side-information registers  ${X}_{1}^{n}{Y}_{1}^{n}{T}_{1}^{n}\mathbf{E}$. Other side-information that Eve obtains, such as $\bar{B}_{1}^{n}$, can be accounted for separately by using appropriate chain rules; see~\ifarxiv Appendix~\ref{app:finitesize}\else the {\it Supplemental Material }\fi~for details. The rate function derived in Theorem \ref{Theorem: SandwichBound} can then be used to provide tight bounds on the $\widetilde{H}_{\alpha}^{\downarrow}\left( A|X=0,E\right)$ term in $h_{\alpha}$, hence allowing us to compute finite-size key rates based on {\Renyi} entropies.

In Fig.~\ref{fig:finitesize}, we show the finite-size key rates obtained from this approach, as applied to the experimental parameters achieved in a DIQKD demonstration in~\cite{NDN+22}. We follow the parameters and implementation choices used in that work as closely as possible, apart from minor modifications we describe in~\ifarxiv Appendix~\ref{app:finitesize}\else the {\it Supplemental Material}\fi. (For Fig.~\ref{subfig:smalln} we also used exactly the same testing probability $\gamma$ as in that work, whereas for Fig.~\ref{subfig:largen} we optimized over the $\gamma$ value; we discuss the details of this choice in that appendix as well.) We see that at the value $n=1.5\times10^6$ used in that experiment, we improve the finite-size key rate by about a factor of $3$. Similarly, we also reduce the minimum $n$ required for nonzero finite-size key by nearly a factor of $3$. Such improvements are critical in the context of practical demonstrations of DIQKD, as they significantly reduce the experimental requirements for a desired length of final key.

As a final remark, we note that in~\cite{inprep_tradeoff}, a framework was developed to prove security for variable-length protocols, which do not simply make a binary accept/abort decision but rather adjust the length of the final key depending on the observed values. The key concept considered in their analysis is a ``weighted'' version of {\Renyi} entropy; refer to e.g.~\cite{fawzi2025additivity} for further details. Our bound in Theorem~\ref{Theorem: SandwichBound} can also be applied to bound these weighted {\Renyi} entropies, and would hence also be able to prove security for variable-length protocols, though we leave a detailed analysis for future work.
\ifarxiv\section{Conclusion}\else\medskip\noindent{\it Conclusion }\fi
This work represents an important step towards a practical implementation of DIQKD which represents the highest level of security, allowing for secret key generation using untrusted hardware. We leverage the R\'{e}nyi Entropy Accumulation Theorem~\cite{arx_AHT24} and demonstrate that it yields significantly tighter finite-size key rates for DIQKD protocols based on the CHSH inequality. To do this, we derive tight analytical bounds on R\'{e}nyi entropies in a device-independent manner, which in turn provide a tight relationship between the CHSH value and the amount of R\'{e}nyi entropy accumulated. Our results can be seen as a generalization of the expression derived in~\cite{pironio2009deviceindependent}, which we recover as a special case. In Figure~\ref{fig:finitesize}, we demonstrate the improvement in finite-size key rates for DIQKD protocols by comparing the rates from the DIQKD experimental demonstration in~\cite{NDN+22} to those that are achievable using our approach. In particular, we show that the work of~\cite{NDN+22} could triple their key-rates by using our technique, with no modifications to the experimental setup. 

Our work prompts several pertinent questions towards the end goal of practical DI cryptography. Firstly, whilst we were also able to derive a tight analytical bound for the ${\widetilde{H}_\alpha^\uparrow}$ entropy, that bound is currently not applicable to improving finite-size analysis, as the current tools (e.g., \cite{arx_AHT24}) use $\widetilde{H}_\alpha^{\downarrow}$ in their key rate expressions. Looking at the analogous results for device-dependent QKD~\cite{inprep_tradeoff,fawzi2025additivity,arq2025marginalconstrained}, the key rates are actually computed in terms of $\widetilde{H}_{\alpha}^\uparrow$, providing tighter bounds on the finite-size rates. If we could develop similar results in the DI setting, we would then be able to use the bound $f_{\widetilde{H}_\alpha^\uparrow}(\score)$ to obtain even higher finite-size key rates. In a second direction, the proofs of the analytical key rate formulae follow closely the work of~\cite{Woodhead_2021} and are fairly generic in nature, hence it is likely that they could be extended to other families of Bell inequalities to obtain further analytical key rate formulae, providing simpler and tighter security proofs for protocols beyond CHSH. 
It may also be of some interest to see whether our techniques apply to the broader family of {\Renyi} conditional entropies in~\cite{rubboli2024quantumconditionalentropies} that unifies the cases considered in this work.

\ifarxiv\acknowledgements\else\medskip\noindent{\it Acknowledgements }\fi
TH acknowledges support from the Peter and Patricia Gruber Award and by the Air Force Office of Scientific Research under award number FA9550-22-1-0391. 
AP acknowledges support from the National Science Centre Poland (Grant No. 2022/46/E/ST2/00115).
    EYZT conducted research at the Institute for Quantum Computing, at the University of Waterloo, which is supported by Innovation, Science, and Economic Development Canada; support was also provided by NSERC under the Discovery Grants Program, Grant No.~341495. PB acknowledges support from the European union’s Horizon
Europe research and innovation programme under the
project “Quantum Secure Networks Partnership” (QSNP,
grant agreement No. 101114043).

\bibliography{References.bib}

\newpage
\onecolumngrid
%\widetext

\ifarxiv\appendix\else\renewcommand{\thesection}{S\arabic{section}}\setcounter{equation}{0}\renewcommand{\theequation}{S\arabic{equation}}   \newpage\begin{center}\large\textbf{Supplemental Material for XXX}\end{center}\fi

\section{Notation and definitions} \label{app:definitions}

We begin by introducing the notation that we will be using. Quantum systems and their associated Hilbert spaces will often be denoted by capital letters, e.g. $A$. Given a space $A$, we denote the set of positive semidefinite operators acting on $A$ by $\Pos(A)$.
An operator $\rho \in \Pos(A)$ is called a \emph{quantum state} if we have $\tr{\rho} = 1$. The set of quantum states on $A$ is denoted by $\States_{=}(A)$. For two operators $\rho, \sigma \in \Pos(A)$ we write $\rho \ll \sigma$ if $\ker{\sigma} \subseteq \ker{\rho}$, where $\ker{\tau} \coloneqq  \{\ket{v} \,:\, \tau \ket{v} = 0\}$. Further, we say $\rho$ is orthogonal to $\sigma$, denoted by $\rho \perp \sigma$, if $\tr{\rho \sigma} = 0$. The function $\log$ denotes the logarithm base 2. We conclude this section with formal definitions of the conditional entropies considered in this work.  
	\begin{defn}
		Given any two positive semi-definite operators $\rho , \sigma\in \text{Pos}(A)$ with $\operatorname{Tr}\left[\rho\right] > 0$, and ${\alpha\in(1,\infty)}$, the \term{sandwiched {\Renyi} divergence} and \term{Petz-{\Renyi} divergence} between $\rho$, $\sigma$ are, respectively,  given by:
		\begin{align}
			\widetilde{D}_\alpha(\rho||\sigma) &\coloneqq \begin{cases}
				\frac{1}{\alpha-1}\log\frac{\operatorname{Tr} \norm{\sigma^{\frac{1-\alpha}{2\alpha}}\rho \sigma^{\frac{1-\alpha}{2\alpha}}}_\alpha^{\alpha}}{\operatorname{Tr}\left[\rho\right]} & \rho \ll \sigma \\
				+\infty & \text{otherwise} \; ,
			\end{cases}
            \intertext{and} 
            \bar{D}_\alpha(\rho||\sigma)
            &\coloneqq \begin{cases}
				\frac{1}{\alpha-1}\log\frac{\operatorname{Tr} \left[\rho^\alpha\sigma^{1-\alpha}\right]}{\operatorname{Tr}\left[\rho\right]} & \rho \ll \sigma \\
				+\infty & \text{otherwise} \; .
			\end{cases}  
		\end{align}
        These definitions are extended to $\alpha=1$ and $\alpha=\infty$ by taking the respective limits.
	\end{defn}	
	\begin{defn}\label{def:condent}
		For any bipartite, normalized state $\rho\in S_{=}(AB)$, and $\alpha\in  [1,\infty]$, we define the following \term{conditional {\Renyi} entropies}:
		\begin{align}
			&\widetilde{H}^{\downarrow}_\alpha(A|B)_\rho\coloneqq -\widetilde{D}_\alpha(\rho_{AB}||\mathds{1}_A\otimes\rho_B) \\
			& \widetilde{H}^{\uparrow}_\alpha(A|B)_\rho\coloneqq \sup_{\sigma \in S_{=}(B)} -\widetilde{D}_\alpha(\rho_{AB}||\mathds{1}_A\otimes\sigma_B) \label{Eq:HUpSandwichDef} \\
			&\Bar{H}^{\downarrow}_\alpha(A|B)_\rho\coloneqq -\bar{D}_\alpha(\rho_{AB}||\mathds{1}_A\otimes\rho_B)  \\
			&\Bar{H}^{\uparrow}_\alpha(A|B)_\rho\coloneqq \sup_{\sigma \in S_{=}(B)} -\bar{D}_\alpha(\rho_{AB}||\mathds{1}_A\otimes\sigma_B)
			\; .\label{Eq:HUpPetzDef}
		\end{align}
	\end{defn}
    It is these four conditional {\Renyi} entropy families that we focus on in this work. For the Petz-{\Renyi} entropies, i.e.\ the latter two expressions, we do not further consider $\alpha  > 2$, as data-processing inequalities do not generally hold in this range. Also note that $\widetilde{H}^{\uparrow}_\infty$ is often referred to as the min-entropy $H_{\min}$ (some works instead refer to $\widetilde{H}^{\downarrow}_\infty$ as the min-entropy, though we shall not use this convention in this work). 

	\section{Analytic Bounds and Proofs}
    \label{app:generalisations}
Our main result, which encompasses the asymmetric CHSH score\footnote{We will refer to the expected value of a Bell-inequality as a ``score'', despite it not being formulated as a nonlocal game, as the term ``value'' is ambiguous in certain places.}
\begin{align}
    \score_{\beta}= \sum_{abxy} (-1)^{xy + a + b} \beta^{1-x}\tr{\rho_{Q_AQ_BE} \left(M_a^{x} \otimes N_b^{y} \otimes \mathds{1}\right)} 
\end{align}
for all $\beta \in \mathbb{R}$~\cite{Woodhead_2021,Sekatski_2021}, is given by the following theorem, which summarizes the results we obtain in the rest of this section. The result for the CHSH inequality is recovered by setting $\beta=1$. We discuss how to include noisy preprocessing in Appendix~\ref{app:noisypreprocessing}.
     
    \begin{thm}\label{thm:allcasesnoNPP}
		Let $\lvert \beta \rvert \geq 1$, $\score_\beta \in \left[2\lvert \beta \rvert, 2\sqrt{1+\beta^2}\right]$, and $g_\score = \sqrt{\tfrac{\score^2_\beta }{4} - \beta^2}$. Then, we have 
		\begin{align}
		      f_{\widetilde{H}^{\downarrow}_\alpha}(\score_{\beta}) &= 1 + \frac{\alpha}{1-\alpha}\log \left[ \left(\frac{1-g_\score}{2}\right)^{\frac{1}{\alpha}} + \left(\frac{1+g_\score}{2} \right)^{\frac{1}{\alpha}}\right]\label{def:down_arrow_sandw} \\
            f_{\widetilde{H}^{\uparrow}_\alpha}(\score_{\beta}) &= 1 + \left( \frac{2\alpha -1}{1-\alpha} \right)\log \left[ \left(\frac{1-g_\score}{2}\right)^{\frac{\alpha}{2\alpha-1}} + \left(\frac{1+g_\score}{2} \right)^{\frac{\alpha}{2\alpha-1}}\right]\label{def:up_arrow_sandw}
		\end{align}
        for all $\alpha \in  (1,\infty)$. Similarly, 
        \begin{align}  
            f_{\Bar{H}^{\downarrow}_\alpha}(\score_{\beta}) &= 1 + \frac{1}{1-\alpha}\log \left[ \left(\frac{1-g_\score}{2}\right)^{2-\alpha} + \left(\frac{1+g_\score}{2} \right)^{2-\alpha}\right] \label{def:down_arrow_petz}\\
    	    f_{\Bar{H}^{\uparrow}_\alpha}(\score_{\beta}) &= 1 + \frac{\alpha}{1-\alpha}\log \left[ \left(\frac{1-g_\score}{2}\right)^{\frac{1}{\alpha}} + \left(\frac{1+g_\score}{2} \right)^{\frac{1}{\alpha}}\right]\label{def:up_arrow_petz} 
    	\end{align}
        for all $\alpha \in  (1,2)$. 
    \end{thm}

\begin{rem}
The above rate functions can be validly extended to $\alpha=\infty$ (for the first pair of formulas) and $\alpha=2$ (for the second pair of formulas) by taking the limits $\alpha\to\infty$ and $\alpha \nearrow 2$ respectively, as we prove and discuss further in Sec.~\ref{app:discont}. (Though for the case of $f_{\bar{H}^{\uparrow}_2}$, taking this limit is slightly unnecessary, in that one can directly substitute $\alpha=2$ to obtain a well-defined expression which also matches the limiting value.) 
Note that the resulting formulas for  $f_{\widetilde{H}^{\downarrow}_\infty}$ and $f_{\bar{H}^{\downarrow}_2}$ are discontinuous with respect to the CHSH score: more precisely, we have $f_{\widetilde{H}^{\downarrow}_\infty}(S) = f_{\bar{H}^{\downarrow}_2}(S) = 0$ for all $\score < 2\sqrt{2}$ and $f_{\widetilde{H}^{\downarrow}_\infty}(S) = f_{\bar{H}^{\downarrow}_2}(S) = 1$ for $\score = 2\sqrt{2}$. 
\end{rem}

The same bounds hold for $\alpha <1$, in suitable parameter ranges.
%\footnote{For the sandwiched {\Renyi} entropies, the rate function does not satisfy the necessary convexity conditions for all $\alpha \in \left[\frac{1}{2},1\right]$. This has an effect on the parameter range for which this bound is provably tight.} 
However, this regime is not generally useful for QKD, so we omit it for brevity. 
We also note that this approach can be extended to $\lvert \beta \rvert \leq 1$ by instead replacing the expression for $g_{\score}$ with that in Eq.~\eqref{Eq: betaleq1} and then taking a concave envelope at an appropriate point in the formula, essentially the same as what was done in~\cite{Woodhead_2021}. Note that taking this concave envelope is indeed necessary in this regime, as the formula resulting from only replacing the $g_{\score}$ expression does not generally satisfy the required convexity properties.

We note that the rate functions $f_{\widetilde{H}^{\downarrow}_\alpha}$ and $f_{\widetilde{H}^{\uparrow}_\alpha}$ are in fact related and we have $f_{\widetilde{H}^{\uparrow}_\alpha} = f_{\widetilde{H}^{\downarrow}_{2-1/\alpha}}$. From this, it immediately follows that our bounds satisfy $f_{\widetilde{H}^{\uparrow}_\infty}= f_{\widetilde{H}^{\downarrow}_2}$, which gives a reasoning for this special case that was observed in~\cite{murta2018realization}.

    \subsection{Qubit Reductions} \label{Sec: Qubit Reduction}

This section contains several known results, which are necessary for future calculations. When both honest parties are restricted to two-input two-output projective measurements, Jordan's lemma, see e.g.~\cite{pironio2009deviceindependent}, can be used to claim that it is sufficient to consider states and projective measurements of the form
    \begin{align}
        \rho_{IQ_AQ_B} &= \sum_{i} \operatorname{Pr}\left[ I=i\right]\ketbra{i}{i}_{I} \otimes \rho_{Q_{A} Q_{B}}^{i}  \\
        M_a^x &= \sum_{i} \ketbra{i}{i}_{I} \otimes M_a^{i,x}  \\
        N_b^y &= \sum_{i} \ketbra{i}{i}_{I} \otimes N_b^{i,y} \, ,
    \end{align}
    where $Q_{A}$, $Q_{B}$ are single-qubit Hilbert spaces on which the projective measurements $M_a^{i,x},N_b^{i,y}$ act. Moreover, Eve's side-information consists of the classical register $I$, as well as the purification of the state $\rho_{Q_{A} Q_{B}}^{i}$, for any value $I=i$. We denote this $\textit{pure}$ tripartite state by $\rho_{Q_{A} Q_{B} E}^{i}$, and the post-measurement state is given by
    \begin{align}
        \rho_{IABE}^{xy} = \sum_{i} \operatorname{Pr} \left[ I=i\right] \ketbra{i}{i}_{I} \otimes \sum_{ab} \left[ \ketbra{ab}{ab}_{AB} \otimes \rho_E^{iabxy} \right]\, .
    \end{align}
    Using this qubit reduction, the asymmetric CHSH score $\score_{\beta}$ will similarly be expressed as a convex mixture over all $i$, i.e.
    \begin{align}
        \score_{\beta} &= \sum_{i}\operatorname{Pr}\left[ I=i\right] \score^{i}_{\beta}  \\
		\score^{i}_{\beta} &= \sum_{abxy} (-1)^{xy + a + b}  \beta^{1-x}\tr{\rho_{Q_AQ_BE}^i \left(M_a^{i,x} \otimes N_b^{i,y} \otimes \mathds{1}\right)}  \; .
	\end{align} 
    After applying the key-generation measurement, Alice's and Eve's joint post-measurement state is given by
    \begin{align} 
        \rho_{IAE} = \sum_{i} \operatorname{Pr} \left[ I=i\right] \ketbra{i}{i}_{I} \otimes \sum_{a} \left[ \ketbra{a}{a}_{A} \otimes \rho_E^{ia} \right]
        \, ,
    \end{align}
where the measurement input $X=0$ is kept implicit and $\rho_E^{ia} = \ptr{Q_{A}}{\rho_{Q_{A}E}^{i} (M_a^{i,0} \otimes \id_E)} $. 

Our goal is now to lower bound $\mathbb{H}(A|I=i, E)_{\rho}$ for each value $I=i$; or, in other words, to lower bound the entropy for states generated by qubit strategies. We note that by using the methods from~\cite{Woodhead_2021}, one can easily show the following (see Lemma~\ref{Lem: QProperties} below for a general version of this property): each such state satisfies
\begin{align} \label{Eq: SingleISimplification}
    \mathbb{H}(A|I=i, E)_{\rho} \geq \mathbb{H}(A|I=i, E)_{\sigma} \, ,
\end{align}
where 
\begin{align}
    \sigma_{IAE} &= \sum_{i} \operatorname{Pr} \left[ I=i\right] \ketbra{i}{i}_{I} \otimes \sigma_{AE}^{i} \label{Eq: AllISigma}\\
    \sigma_{AE}^{i} &= \frac{1}{2} \ketbra{0}{0} \otimes \ketbra{\psi_{=}}{\psi_{=}} +  \frac{1}{2} \ketbra{1}{1} \otimes \ketbra{\psi_{\neq}}{\psi_{\neq}} \, , \label{Eq: SingleISigma}
\end{align}
for a pair of vectors $\{\ket{\psi_{=}},\ket{\psi_{\neq}} \}$ that can be written in the following form (for some basis vectors $\{\ket{0},\ket{1}\}$ of a two-dimensional subspace containing the span of $\{\ket{\psi_{=}},\ket{\psi_{\neq}} \}$):
\begin{align}
  \ket{\psi_{=}} &= \ket{0} \label{Eq: psieq}\\
    \ket{\psi_{\neq}} &= g_{\score}^i\ket{0} + \sqrt{1-g_{\score}^{i2}}\ket{1} \, , \label{Eq: psineq}
\end{align}
where $g_{\score}^i=\sqrt{\tfrac{\score^{i2}_{\beta} }{4} -\beta^2}$ for $\lvert \beta \rvert \geq 1$. For $\lvert \beta \rvert \leq 1$,  one instead uses $g_{\score}^i = E_{\beta}(\score_\beta^i)$, where
	\begin{align}
		E_{\beta}=\begin{cases}
			\sqrt{\tfrac{\score^{i2}_{\beta} }{4} -\beta^2}, & \text{if }\lvert \score^{i}_{\beta} \rvert  \geq 2\sqrt{1+\beta^2-\beta^4}\\
			\sqrt{1-\left( 1-\frac{1}{|\beta|}\sqrt{\left(1-\beta^2\right)\left(\frac{\score^{i2}_{\beta}}{4}-1\right)}\right)^2}, & \text{if }\lvert \score^{i}_{\beta} \rvert  \leq 2\sqrt{1+\beta^2-\beta^4} \; . \label{Eq: betaleq1}
		\end{cases}  
\end{align}
More generally, the above property of qubit strategies is an instance of the following lemma. Here, one should view $\mathbb{Q}$ as the function for which (depending on which case we are considering) either $\mathbb{H}^{\downarrow}\left(A|B\right) = \frac{1}{1-\alpha}\log \mathbb{Q}\left(A|B\right)$ or $\mathbb{H}^{\uparrow}\left(A|B\right) = \frac{\alpha}{1-\alpha}\log \mathbb{Q}\left(A|B\right)$ holds, using $\mathbb{H}$ to generically represent either Petz or sandwiched entropy.

\begin{lem} \label{Lem: QProperties}
    Suppose $\mathbb{Q}(A|B): \States_{=}(AB) \to \mathbb{R}$ is a function satisfying
    \begin{enumerate}
        \item (Local unitary invariance): For any unitary $V$ on $A$ we have $\mathbb{Q}(A|B)_{\rho_{AB}}= \mathbb{Q}(A|B)_{V \rho_{AB} V^\dagger}$. 
        \item (Classical linearity): For any state $\rho_{ABC} = \sum_c \Pr[C=c] \rho_{AB}^c \otimes \ketbra{c}{c}_C$ classical on $C$, we have  $\mathbb{Q}(A|BC)_{\rho_{ABC}} = \sum_{c} \Pr[C=c] \mathbb{Q}(A|B)_{\rho_{AB}^c}$.
        \item (Data processing): For any $\rho_{ABC} \in \States_{=}(ABC)$ we have $\mathbb{Q}(A|BC)_{\rho_{ABC}} \geq \mathbb{Q}(A|B)_{\rho_{AB}}$.
    \end{enumerate}

     Let $\ket{\psi} \in Q_AQ_BE$ with $Q_A$ and $Q_B$ being qubit systems, let $\{M_a\}_a$ be a rank-one projective measurement on $Q_A$ and let
        \begin{equation}
            \rho_{AE} = \sum_a \ketbra{a}{a}_A \otimes \rho_E^a
        \end{equation}
        be the post-measurement state such that $g_{\score} \geq 0$. Then 
        there exists a state 
        \begin{equation}
                    \sigma_{AE} = \frac{1}{2} \ketbra{0}{0} \otimes \ketbra{\psi_{=}}{\psi_{=}}  +  \frac{1}{2} \ketbra{1}{1} \otimes \ketbra{\psi_{\neq}}{\psi_{\neq}} \label{Eq: AltStateEmbedded}
        \end{equation}
 such that $|\inner{\psi_{=}}{\psi_{\neq}}| \geq g_{\score}$ and 
        \begin{equation}
            \mathbb{Q}(A|E)_{\rho_{AE}} \leq \mathbb{Q}(A|E)_{\sigma_{AE}}  \, , \label{Eq: QUpperBound}
        \end{equation}
where $g_{\score}$ is as defined above and, for any $\beta \in \mathbb{R}$, depends on the asymmetric CHSH score, $\score_\beta$, achieved by the above system.
\end{lem}
 \begin{proof}
            The proof is detailed for the special case of the von Neumann entropy in~\cite{Woodhead_2021}. We extend it to the more general setting considered here but we note that it remains almost exactly the same proof. 
Without loss of generality, let us assume that $\{M_a\}_a$ is a measurement in the computational basis.\footnote{Any two-output projective measurement on a qubit is related to it via a unitary, which we can implicitly apply a priori.}
Let 
\begin{align}
    \rho_{AE} &= \sum_a \ketbra{a}{a}_A \otimes \rho_E^a \\
    \rho^\prime_{AE} &= \sum_a \ketbra{a \oplus 1}{a\oplus 1}_A \otimes \rho_E^a   \, .
\end{align}
Due to local unitary invariance, we have that $\mathbb{Q}(A|E)_{\rho_{AE}} = \mathbb{Q}(A|E)_{\rho_{AE}^\prime}$. Moreover, using classical linearity, it holds that 
\begin{align}
    \mathbb{Q}(A|EF)_{\bar{\rho}_{AEF}} = \mathbb{Q}(A|E)_{\rho_{AE}} \, ,
\end{align}
where
\begin{align}
    \bar{\rho}_{AEF} &= \frac{1}{2} \rho_{AE}\otimes \ketbra{0}{0}_{F} +\frac{1}{2} \rho_{AE}^{\prime}\otimes \ketbra{1}{1}_{F} \\
    &= \frac{1}{2} \ketbra{0}{0}_{A} \otimes \left(\sum_{a}\rho_{E}^{a}\otimes \ketbra{a}{a}_{F} \right) + \frac{1}{2} \ketbra{1}{1}_{A} \otimes \left(\sum_{a}\rho_{E}^{a \oplus 1}\otimes \ketbra{a}{a}_{F} \right)   \, .
\end{align}
Furthermore, the initial state can be written as 
\begin{align}
    \ket{\psi}_{Q_AQ_BE} = \ket{0}_{Q_A}\otimes \ket{\psi_{0}}_{BE}+\ket{1}_{Q_A}\otimes \ket{\psi_{1}}_{BE} \, ,
\end{align}
and $\ket{\psi_{a}}_{BE}$ can be viewed as purifications of $\rho_{E}^{a}$. One potential extension of $\bar{\rho}_{AEF}$ is thus
\begin{align}
    \bar{\rho}_{ABEFF^\prime} = \frac{1}{2} \ketbra{0}{0}_{A} \otimes \ketbra{\psi_{=}}{\psi_{=}}_{BEFF^\prime} + \frac{1}{2} \ketbra{1}{1}_{A} \otimes \ketbra{\psi_{\neq}}{\psi_{\neq}}_{BEFF^\prime} \, ,
\end{align}
where 
\begin{align}
    \ket{\psi_{=}} &= \ket{\psi_{0}}_{BE} \otimes \ket{00}_{FF^\prime} + \ket{\psi_{1}^\prime}_{BE} \otimes \ket{11}_{FF^\prime} \\
    \ket{\psi_{\neq}} &= \ket{\psi_{1}^\prime}_{BE} \otimes \ket{00}_{FF^\prime} + \ket{\psi_{0}}_{BE} \otimes \ket{11}_{FF^\prime} \, ,
\end{align}
and $\ket{\psi_{1}^\prime}_{BE} = \left(\sigma_X \otimes \mathds{1} \right) \ket{\psi_{1}}_{BE}$, where $\sigma_X$ denotes the corresponding Pauli operator. Due to the data processing inequality, it holds that 
\begin{align}
    \mathbb{Q}(A|E)_{\rho_{AE}} \leq \mathbb{Q}(A|BEFF^\prime)_{\bar{\rho}_{ABEFF^\prime}} \, .
\end{align}
We have thus identified a state for which the desired inequality holds. The states $\{\ket{\psi_{=}},\ket{\psi_{\neq}}\}$ span (at most) a two-dimensional subspace, which can be embedded in the Hilbert space $E$, yielding the state in Eq.~\eqref{Eq: AltStateEmbedded}. Moreover, by~\cite[Eqs.~(73) and (95)]{Woodhead_2021} it must hold that $|\inner{\psi_{=}}{\psi_{\neq}}| \geq g_{\score}$.
        \end{proof}

\begin{rem}
    Since $\{\ket{\psi_{=}},\ket{\psi_{\neq}}\}$ span at most a two-dimensional subspace, they
can always be written as 
\begin{align}
     \ket{\psi_{=}} &= e^{i \phi}\ket{0} \\
    \ket{\psi_{\neq}} &= |\inner{\psi_{=}}{\psi_{\neq}}|\ket{0} + \sqrt{1-|\inner{\psi_{=}}{\psi_{\neq}}|^{2}}\ket{1} \, .
\end{align}
However, as can be seen in Eq.~\eqref{Eq: AltStateEmbedded}, the angle $\phi$ disappears. As such, one can without loss of generality set it to $\phi =0$. Moreover, in principle, the rate functions we derive in the following sections should depend on $|\inner{\psi_{=}}{\psi_{\neq}}|$ rather than $g_{\score}$. However, all our rate functions are monotonically increasing in the score (also in $|\inner{\psi_{=}}{\psi_{\neq}}|$), and one can thus further lower bound the entropy by replacing $|\inner{\psi_{=}}{\psi_{\neq}}|$ with $g_{\score}$. This is a property that is also implicitly used in~\cite{Woodhead_2021}, and it is this feature that allows us to simply consider states of the form given by Eq.~\eqref{Eq: SingleISigma}. In our case, monotonicity is proven in Appendix~\ref{app: ConcavityProperties}.
\end{rem}
\begin{rem}
The last two properties we wish to mention are as follows.  Due to the relation between $\mathbb{Q}(A|E)$, and $\mathbb{H}(A|E)$, Eq.~\eqref{Eq: QUpperBound} immediately implies Eq.~\eqref{Eq: SingleISimplification}. Also, we do not consider degenerate measurements, which only have one potential outcome, in Lemma~\ref{Lem: QProperties}. The reason for this is simple, and also discussed in~\cite{Woodhead_2021}: not only can such a measurement never be part of a CHSH set-up that violates the (asymmetric) CHSH inequality, but Alice's output would be deterministic. By choosing $\{\ket{\psi_{=}},\ket{\psi_{\neq}}\}$ such that $|\inner{\psi_{=}}{\psi_{\neq}}|=0$, we have $\mathbb{H}(A|I=i, E)_{\rho} = \mathbb{H}(A|I=i, E)_{\sigma}=0$. Eq.~\eqref{Eq: SingleISimplification} is thus trivially satisfied. Similarly, without loss of generality, we only consider set-ups for which $g_{\score}$ is well-defined. As an example, for $\beta =1 $, a well-defined $g_{\score}$ corresponds to achieving a CHSH score of $\score \geq 2$. These other cases would not violate the desired (asymmetric) CHSH inequality. By instead choosing $\score_{\beta}$ such that $g_{\score}=0$, one achieves a higher score and we still bound any {\Renyi} entropy by $0$. 
\end{rem}
\subsection{Derivation of \texorpdfstring{$f_{\widetilde{H}^{\downarrow}_\alpha}(\score_{\beta})$}{Downarrow Sandwiched {\Renyi} Entropy}}

\begin{thm}
		Let $\alpha \in   (1,\infty)$, $\lvert \beta \rvert \geq 1$,  and $\score_\beta \in \left[2\lvert \beta \rvert, 2\sqrt{1+\beta^2}\right]$. Then
		\begin{align}		f_{\widetilde{H}^{\downarrow}_\alpha}(\score_{\beta}) = 1 + \frac{\alpha}{1-\alpha}\log \left[ \left(\frac{1-g_\score}{2}\right)^{\frac{1}{\alpha}} + \left(\frac{1+g_\score}{2} \right)^{\frac{1}{\alpha}}\right]  \; ,
		\end{align}
        where 
$g_\score = \sqrt{\tfrac{\score^2_\beta }{4} - \beta^2}$.
	\end{thm}
    \begin{proof}
    For any $\rho_{IAE}$ of the form given by Eq.~\eqref{Eq: AllISigma}, we first prove that Alice's measurement outcome after a key-generation measurement satisfies
    \begin{align}
        \widetilde{H}^{\downarrow}_\alpha(A|X=0,IE)_{\rho} \geq 1 + \frac{\alpha}{1-\alpha}\log \left[ \left(\frac{1-g_\score}{2}\right)^{\frac{1}{\alpha}} + \left(\frac{1+g_\score}{2} \right)^{\frac{1}{\alpha}}\right]\, .
    \end{align}
    We now consider Alice's and Eve's bipartite state for some $I=i$. For any such $\sigma_{AE}$ (we omit the index $i$ for now) as in Eq.\ \eqref{Eq: SingleISigma}, Eve's reduced density matrix is given by 
    \begin{equation}\label{Eq:reducedSigma}
        \sigma_{E}=\frac{1-g_\score}{2} \ketbra{v_1}{v_1} + \frac{1+g_\score}{2} \ketbra{v_2}{v_2},
    \end{equation}
    where
    \begin{align}
        \ket{v_1} &= - \sqrt{\frac{1-g_\score}{2}} \ket{0} + \sqrt{\frac{1+g_\score}{2}} \ket{1} \label{Eq: v1}\\
        \ket{v_2} &=\phantom{-} \sqrt{\frac{1+g_\score}{2}} \ket{0} + \sqrt{\frac{1-g_\score}{2}} \ket{1} \label{Eq: v2}\, .
    \end{align}
Plugging this directly into $\sigma_{E}^{\frac{1-\alpha}{2\alpha}}\sigma_{AE}\sigma_{E}^{\frac{1-\alpha}{2\alpha}}$ gives us that
\begin{align}
    & \sigma_{E}^{\frac{1-\alpha}{2\alpha}}\sigma_{AE}\sigma_{E}^{\frac{1-\alpha}{2\alpha}} \\
    %&\notag\\
    & = \frac{1}{2} \ketbra{0}{0}_{A} \otimes \left[ \left( \frac{1-g_\score}{2}\right)^{\frac{1}{\alpha}}\ketbra{v_1}{v_1}_{E} -  \left( \frac{1-g_\score^2}{4}\right)^{\frac{1}{2\alpha}}\left(\ketbra{v_1}{v_2}_{E}+\ketbra{v_2}{v_1}_{E} \right) + \left( \frac{1+g_\score}{2}\right)^{\frac{1}{\alpha}}\ketbra{v_2}{v_2}_{E}
    \right] \notag \\
    &+ \frac{1}{2} \ketbra{1}{1}_{A} \otimes \left[ \left( \frac{1-g_\score}{2}\right)^{\frac{1}{\alpha}}\ketbra{v_1}{v_1}_{E} +  \left( \frac{1-g_\score^2}{4}\right)^{\frac{1}{2\alpha}}\left(\ketbra{v_1}{v_2}_{E}+\ketbra{v_2}{v_1}_{E} \right) + \left( \frac{1+g_\score}{2}\right)^{\frac{1}{\alpha}}\ketbra{v_2}{v_2}_{E}
    \right]  \\
    & = \frac{1}{2} \ketbra{0}{0}_{A} \otimes \left[\left( \left( \frac{1-g_\score}{2}\right)^{\frac{1}{2\alpha}}\ket{v_1}_{E} - \left( \frac{1+g_\score}{2}\right)^{\frac{1}{2\alpha}}\ket{v_2}_{E} \right)\left( \left( \frac{1-g_\score}{2}\right)^{\frac{1}{2\alpha}}\bra{v_1}_{E} - \left( \frac{1+g_\score}{2}\right)^{\frac{1}{2\alpha}}\bra{v_2}_{E} \right)
    \right] \notag \\
    &+ \frac{1}{2} \ketbra{1}{1}_{A} \otimes \left[\left( \left( \frac{1-g_\score}{2}\right)^{\frac{1}{2\alpha}}\ket{v_1}_{E} + \left( \frac{1+g_\score}{2}\right)^{\frac{1}{2\alpha}}\ket{v_2}_{E} \right)\left( \left( \frac{1-g_\score}{2}\right)^{\frac{1}{2\alpha}}\bra{v_1}_{E} + \left( \frac{1+g_\score}{2}\right)^{\frac{1}{2\alpha}}\bra{v_2}_{E} \right)
    \right] \, .
\end{align}
This can alternatively be expressed as
\begin{multline}
    \sigma_{E}^{\frac{1-\alpha}{2\alpha}}\sigma_{AE}\sigma_{E}^{\frac{1-\alpha}{2\alpha}}= \frac{1}{2}\left(\left(\frac{1-g_\score}{2}\right)^{\frac{1}{\alpha}} + \left(\frac{1+g_\score}{2} \right)^{\frac{1}{\alpha}} \right)\ketbra{0}{0}_{A} \otimes \ketbra{w_1}{w_1}_{E}  \\ + \frac{1}{2}\left(\left(\frac{1-g_\score}{2}\right)^{\frac{1}{\alpha}} + \left(\frac{1+g_\score}{2} \right)^{\frac{1}{\alpha}} \right)\ketbra{1}{1}_{A} \otimes \ketbra{w_1^\prime}{w_1^\prime}_{E} \, , \label{Eq: SandwichedInterimExpression}
\end{multline}
where $\{\ket{w_1},\ket{w_1^\prime} \}$ are normalized vectors. It then directly follows from this that
\begin{align}
\widetilde{H}^{\downarrow}_\alpha(A|E)_{\sigma} &=  \frac{1}{1-\alpha}\log\left[ \operatorname{Tr} \left[\left(\sigma_{E}^{\frac{1-\alpha}{2\alpha}}\sigma_{AE}\sigma_{E}^{\frac{1-\alpha}{2\alpha}}\right)^\alpha \right] \right] \\
 & = \frac{1}{1-\alpha}\log \frac{2\left[\left(\frac{1-g_\score}{2}\right)^{\frac{1}{\alpha}} + \left(\frac{1+g_\score}{2} \right)^{\frac{1}{\alpha}} \right]^{\alpha}}{2^{\alpha} } \\
 & = \frac{1}{1-\alpha}\log \frac{\left[\left(\frac{1-g_\score}{2}\right)^{\frac{1}{\alpha}} + \left(\frac{1+g_\score}{2} \right)^{\frac{1}{\alpha}} \right]^{\alpha}}{2^{\alpha-1} } \\
   &= 1 + \frac{\alpha}{1-\alpha}\log \left[ \left(\frac{1-g_\score}{2}\right)^{\frac{1}{\alpha}} + \left(\frac{1+g_\score}{2} \right)^{\frac{1}{\alpha}}\right] \, .
\end{align}
This concludes the calculations for the individual qubit block related to some index $I=i$. We now explicitly write the index for rest of the calculation. Let $h(\score^i_{\beta})$ denote the function
\begin{align}
    h(\score^i_{\beta}) \coloneqq  \left[ \left(\frac{1-g_\score^i}{2}\right)^{\frac{1}{\alpha}} + \left(\frac{1+g_\score^i}{2} \right)^{\frac{1}{\alpha}}\right]^{\alpha} \, .
\end{align}
Using both Eq.~\eqref{Eq: SingleISimplification} and~\cite[Prop.~5.1]{Tomamichel2015QuantumIP}, it holds that
\begin{align}
\widetilde{H}^{\downarrow}_\alpha(A|X=0,IE)_{\rho} &\geq \frac{1}{1-\alpha}\log \left[\sum_{i} \operatorname{Pr} \left[ I=i\right] 2^{\left(1-\alpha\right)\widetilde{H}^{\downarrow}_\alpha(A|I=i,E)_{\sigma^i}} \right]\\
&= 1+ \frac{1}{1-\alpha}\log \left[\sum_{i} \operatorname{Pr} \left[ I=i\right] h(\score^i_{\beta})\right] \, .
\end{align}
We show in Appendix~\ref{app: ConcavityProperties} that $h(\score^i_{\beta})$ is concave for $\alpha > 1$. The desired lower bound then follows from this property, together with the monotonicity of the logarithm, i.e. 
\begin{align}
\widetilde{H}^{\downarrow}_\alpha(A|X=0,IE)_{\rho} &\geq 1+ \frac{1}{1-\alpha}\log \left[ h(\score_{\beta}) \right] \\
&= 1 + \frac{\alpha}{1-\alpha}\log \left[ \left(\frac{1-g_\score}{2}\right)^{\frac{1}{\alpha}} + \left(\frac{1+g_\score}{2} \right)^{\frac{1}{\alpha}}\right] \, ,
\end{align}
where $g_\score = \sqrt{\tfrac{\score^2_\beta }{4} - \beta^2}$ and $\score_{\beta} = \operatorname{Pr} \left[ I=i\right] \score^i_{\beta}$. This inequality is saturated if Alice and Bob share the state
\begin{align}  \sqrt{P_+}\ket{\phi^+}_{Q_AQ_B}\ket{0}_{E}+\sqrt{P_-}\ket{\phi^-}_{Q_AQ_B}\ket{1}_{E} \, ,
		\end{align}
		where $P_\pm = \tfrac{1}{2}\left(1 \pm g_{\score}\right)$,
 and measure the observables $A_0 = \sigma_{z}$, $A_1 = \sigma_{x}$, $B_0 = \frac{\beta\sigma_{z} + g_{\score}\sigma_{x}}{\sqrt{\beta^2+g_{\score}^2}}$, $B_1 = \frac{\beta\sigma_{z} - g_{\score}\sigma_{x}}{\sqrt{\beta^2+g_{\score}^2}}$ (see Appendix~\ref{Sec: TightnessBounds} for more details). The derived rate function is thus tight.

    \end{proof}
\subsection{Derivation of \texorpdfstring{$f_{\widetilde{H}^{\uparrow}_\alpha}(\score_{\beta})$}{Uparrow Sandwiched {\Renyi} Entropy}} 
\label{Uparrow Sandwiched {\Renyi} Entropy}

\begin{thm}
		Let $\alpha \in   (1,\infty)$, $\lvert \beta \rvert \geq 1$,  and $\score_\beta \in \left[2\lvert \beta \rvert, 2\sqrt{1+\beta^2} \right]$. Then
		\begin{align}
			f_{\widetilde{H}^{\uparrow}_\alpha}(\score_{\beta}) =f_{\widetilde{H}^{\downarrow}_{2-\frac{1}{\alpha}}}(\score_{\beta}) \; .
		\end{align}
	\end{thm}
\begin{proof}
  For any $\rho_{IAE}$ of the form given by Eq.~\eqref{Eq: AllISigma}, we first consider Alice's and Eve's bipartite state for some $I=i$.  
    Due to Eq.~\eqref{Eq:HUpSandwichDef}, for any such $\sigma_{AE}$ (we omit the index $i$ for now) as in Eq.~\eqref{Eq: SingleISigma}, it must hold that 
    \begin{align}
        \widetilde{H}^{\uparrow}_{\alpha}\left(A|E\right)_\sigma \geq -\widetilde{D}_\alpha(\sigma_{AE}||\mathds{1}_A\otimes\tau_E) \, ,
    \end{align}
    where we pick the following choice of state in the second argument:
    \begin{equation}
        \tau_E = \frac{\sigma_E^{\frac{\alpha}{2\alpha-1}}}{\mathrm{Tr}[\sigma_E^{\frac{\alpha}{2\alpha-1}}]}.
    \end{equation}
    For this state, however, one finds that
    \begin{align}       
        -\widetilde{D}_\alpha(\sigma_{AE}||\mathds{1}_A\otimes\tau_E) &= \frac{1}{1-\alpha}\log \left[ \operatorname{Tr} \left[\left(\tau_{E}^{\frac{1-\alpha}{2\alpha}}\sigma_{AE}\tau_{E}^{\frac{1-\alpha}{2\alpha}} \right)^\alpha \right]\right] \\
        &= \frac{1}{1-\alpha}\log \left[\frac{ \operatorname{Tr} \left[\left(\sigma_{E}^{\frac{1-\alpha}{2\left(2\alpha-1\right)}}\sigma_{AE}\sigma_{E}^{\frac{1-\alpha}{2\left(2\alpha-1\right)}} \right)^\alpha \right]}{\operatorname{Tr}\left[ \sigma_{E}^{\frac{\alpha}{2\alpha -1}}\right]^{1-\alpha} }\right] \\
        &= \frac{1}{1-\alpha}\log \left[\frac{ \operatorname{Tr} \left[\left(\sigma_{E}^{\frac{1-\alpha^\prime}{2\alpha^\prime}}\sigma_{AE}\sigma_{E}^{\frac{1-\alpha^\prime}{2\alpha^\prime}} \right)^\alpha \right]}{\operatorname{Tr}\left[ \sigma_{E}^{\frac{1}{\alpha^\prime }}\right]^{1-\alpha} }\right] \, ,
    \end{align}
    where $\alpha^\prime = 2-\frac{1}{\alpha}$. Moreover, it holds that
    \begin{align}
        \operatorname{Tr} \left[\left(\sigma_{E}^{\frac{1-\alpha^\prime}{2\alpha^\prime}}\sigma_{AE}\sigma_{E}^{\frac{1-\alpha^\prime}{2\alpha^\prime}} \right)^\alpha \right] &  = 2\left(\frac{\left(\frac{ 1-g_\score}{2}\right)^{\frac{1}{\alpha^\prime}} + \left( \frac{1+g_\score}{2}\right)^{\frac{1}{\alpha^\prime}} }{2}\right)^\alpha = \frac{1}{2^{\alpha-1}} \left( \left( \frac{1-g_\score}{2}\right)^{\frac{1}{\alpha^\prime}} + \left( \frac{1+g_\score}{2}\right)^{\frac{1}{\alpha^\prime}} \right)^\alpha \\
        \operatorname{Tr}\left[ \sigma_{E}^{\frac{1}{\alpha^\prime }}\right]^{1-\alpha}&= \left(\left( \frac{1-g_\score}{2}\right)^{\frac{1}{\alpha^\prime}} + \left( \frac{1+g_\score}{2}\right)^{\frac{1}{\alpha^\prime}} \right)^{1-\alpha} \, ,
    \end{align}
     where the first equation directly follows from Eq.~\eqref{Eq: SandwichedInterimExpression} and the second equation is due to the decomposition ${\sigma_{E}=\frac{1-g_\score}{2} \ketbra{v_1}{v_1} + \frac{1+g_\score}{2} \ketbra{v_2}{v_2}}$, where $\{\ket{v_1},\ket{v_2} \}$ are orthonormal vectors given  by Eqs.~\eqref{Eq: v1}--\eqref{Eq: v2}. This then gives us that
    \begin{align}
       \widetilde{H}^{\uparrow}_{\alpha}\left(A|E\right)_\sigma &\geq \frac{1}{1-\alpha}\log \left[\frac{ \operatorname{Tr} \left[\left(\sigma_{E}^{\frac{1-\alpha^\prime}{2\alpha^\prime}}\sigma_{AE}\sigma_{E}^{\frac{1-\alpha^\prime}{2\alpha^\prime}} \right)^\alpha \right]}{\operatorname{Tr}\left[ \sigma_{E}^{\frac{1}{\alpha^\prime }}\right]^{1-\alpha} }\right] \\
        &= \frac{1}{1-\alpha} \log \left[ 2^{1-\alpha} \left(\left( \frac{1-g_\score}{2}\right)^{\frac{1}{\alpha^\prime}} + \left( \frac{1+g_\score}{2}\right)^{\frac{1}{\alpha^\prime}} \right)^{2\alpha-1} \right] \\
        &= 1+ \frac{1}{1-\alpha} \log \left[  \left(\left( \frac{1-g_\score}{2}\right)^{\frac{1}{\alpha^\prime}} + \left( \frac{1+g_\score}{2}\right)^{\frac{1}{\alpha^\prime}} \right)^{2\alpha-1} \right] \\
        &= 1+ \frac{2\alpha-1}{1-\alpha} \log \left[  \left( \frac{1-g_\score}{2}\right)^{\frac{1}{\alpha^\prime}} + \left( \frac{1+g_\score}{2}\right)^{\frac{1}{\alpha^\prime}}  \right] \\
        &= 1+ \frac{\alpha^\prime}{1-\alpha^\prime} \log \left[  \left( \frac{1-g_\score}{2}\right)^{\frac{1}{\alpha^\prime}} + \left( \frac{1+g_\score}{2}\right)^{\frac{1}{\alpha^\prime}}  \right] \, .
    \end{align}
    This, however, is simply $f_{\widetilde{H}^{\downarrow}_{\alpha^\prime}}(\score_{\beta})$. This concludes the calculations for the individual qubit block related to some index $I=i$. We now explicitly write the index for rest of the calculation. Let $h(\score^i_{\beta})$ denote the function
\begin{align}
    h(\score^i_{\beta}) \coloneqq  \left[ \left(\frac{1-g_\score^i}{2}\right)^{\frac{1}{\alpha^\prime}} + \left(\frac{1+g_\score^i}{2} \right)^{\frac{1}{\alpha^\prime}}\right]^{\alpha^\prime} \, .
\end{align}
Using both Eq.~\eqref{Eq: SingleISimplification} and~\cite[Prop.~5.1]{Tomamichel2015QuantumIP}, it holds that
\begin{align}
\widetilde{H}^{\uparrow}_\alpha(A|X=0,IE)_{\rho} &\geq \frac{\alpha}{1-\alpha}\log \left[\sum_{i} \operatorname{Pr} \left[ I=i\right] 2^{\frac{\left(1-\alpha\right)}{\alpha}\widetilde{H}^{\uparrow}_\alpha(A|I=i,E)_{\sigma^i}} \right]\\
&= \frac{1}{1-\alpha^\prime}\log \left[\sum_{i} \operatorname{Pr} \left[ I=i\right] 2^{\left(1-\alpha^\prime\right)\widetilde{H}^{\uparrow}_\alpha(A|I=i,E)_{\sigma^i}} \right]\\
&\geq  1+ \frac{1}{1-\alpha^\prime}\log \left[\sum_{i} \operatorname{Pr} \left[ I=i\right] h(\score^i_{\beta}) \right]\, .
\end{align}
We show in Appendix~\ref{app: ConcavityProperties} that $h(\score^i_{\beta})$ is concave for $\alpha^\prime > 1$. The desired lower bound then follows from this property, together with the monotonicity of the logarithm, i.e. 
\begin{align}
\widetilde{H}^{\uparrow}_\alpha(A|X=0,IE)_{\rho} &\geq 1+ \frac{1}{1-\alpha^\prime}\log \left[ h(\score_{\beta})\right] \\
&= 1 + \frac{\alpha^\prime}{1-\alpha^\prime}\log \left[ \left(\frac{1-g_\score}{2}\right)^{\frac{1}{\alpha^\prime}} + \left(\frac{1+g_\score}{2} \right)^{\frac{1}{\alpha^\prime}}\right] \\
&=f_{\widetilde{H}^{\downarrow}_{\alpha^\prime}}(\score_{\beta}) \, ,
\end{align}
where $g_\score = \sqrt{\tfrac{\score^2_\beta }{4} - \beta^2}$ and $\score_{\beta} = \operatorname{Pr} \left[ I=i\right] \score^i_{\beta}$. This inequality is saturated if Alice and Bob share the state
\begin{align}  \ket{\psi} =\sqrt{P_+}\ket{\phi^+}_{Q_AQ_B}\ket{0}_{E}+\sqrt{P_-}\ket{\phi^-}_{Q_AQ_B}\ket{1}_{E} \, ,
		\end{align}
		where $P_\pm = \tfrac{1}{2}\left(1 \pm g_{\score}\right)$,
 and measure the observables $A_0 = \sigma_{z}$, $A_1 = \sigma_{x}$, $B_0 = \frac{\beta\sigma_{z} + g_{\score}\sigma_{x}}{\sqrt{\beta^2+g_{\score}^2}}$, $B_1 = \frac{\beta\sigma_{z} - g_{\score} \sigma_{x}}{\sqrt{\beta^2+g_{\score}^2}}$. To see this, we first note it can be readily verified that this achieves the desired score, $\score_{\beta}$. Moreover, after Alice applies measurement $A_0$, the classical output (stored in register $A$) satisfies
      \begin{align}        \widetilde{H}^{\uparrow}_{\alpha}\left(A|E\right)_\psi \leq \widetilde{H}^{\downarrow}_{2-\frac{1}{\alpha}} \left(A|E\right)_\psi \, ,
    \end{align}
    due to~\cite[Cor.~5.3]{Tomamichel2015QuantumIP}. The latter, however, is simply equal to $f_{\widetilde{H}^{\downarrow}_{2-\frac{1}{\alpha}}}(\score_{\beta})$ (see Appendix~\ref{Sec: TightnessBounds} for more details). The derived upper and lower bounds are thus equal and the rate function must therefore be tight.
\end{proof}

\subsection{Derivation of \texorpdfstring{$f_{\bar{H}^{\downarrow}_\alpha}(\score_{\beta})$}{Downarrow Petz-{\Renyi} Entropy}}
\begin{thm}
		Let $\alpha \in   (1,2)$, $\lvert \beta \rvert \geq 1$,  and $\score_\beta \in \left[2\lvert \beta \rvert, 2\sqrt{1+\beta^2}\right]$. Then
		\begin{align}		f_{\bar{H}^{\downarrow}_\alpha}(\score_\beta) = 1 + \frac{1}{1-\alpha}\log \left[ \left(\frac{1-g_\score}{2}\right)^{2-\alpha} + \left(\frac{1+g_\score}{2} \right)^{2-\alpha}\right]  \; ,
		\end{align}
        where 
$g_\score = \sqrt{\tfrac{\score^2_\beta }{4} - \beta^2}$.
	\end{thm}

\begin{proof}
For any $\rho_{IAE}$ of the form given by Eq.~\eqref{Eq: AllISigma}, we first prove that Alice's measurement outcome after a key-generation measurement satisfies
\begin{align}
\bar{H}^{\downarrow}_\alpha(A|X=0,IE)_{\rho} \geq 1 + \frac{1}{1-\alpha}\log \left[ \left(\frac{1-g_\score}{2}\right)^{2-\alpha} + \left(\frac{1+g_\score}{2} \right)^{2-\alpha}\right]  \, .
\end{align}
We now consider Alice's and Eve's bipartite state for some $I=i$. For any such $\sigma_{AE}$ (we omit the index $i$ for now) as in Eq.\ \eqref{Eq: SingleISigma}, one can readily verify that
 \begin{align}
     \sigma_{AE}^{\alpha} &= \frac{1}{2^\alpha} \ketbra{0}{0} \otimes \ketbra{\psi_{=}}{\psi_{=}} +  \frac{1}{2^\alpha} \ketbra{1}{1} \otimes \ketbra{\psi_{\neq}}{\psi_{\neq}} \\
     \sigma_{E}^{\frac{1-\alpha}{2}}&=\left(\frac{1-g_\score}{2}\right)^{\frac{1-\alpha}{2}} \ketbra{v_1}{v_1} + \left(\frac{1+g_\score}{2} \right)^{\frac{1-\alpha}{2}}\ketbra{v_2}{v_2} \, ,
 \end{align}
where $\{ \ket{\psi_{=}},\ket{\psi_{\neq}} \}$ are given by Eqs.~\eqref{Eq: psieq}--\eqref{Eq: psineq} and $\{ \ket{v_1},\ket{v_2} \}$  by Eqs.~\eqref{Eq: v1}--\eqref{Eq: v2}. Using the fact that $\sqrt{1-g_\score^2}\sqrt{\frac{1\pm g_\score}{2}} = \left( 1 \pm g_\score\right)\sqrt{\frac{1 \mp g_\score}{2}}$, we thus find that 
\begin{align}
    \sigma_{E}^{\frac{1-\alpha}{2}}\sigma_{AE}^{\alpha}\sigma_{E}^{\frac{1-\alpha}{2}} &= \frac{1}{2^\alpha}\left(\frac{1-g_\score}{2}\right)^{2-\alpha}\left( \ketbra{0}{0} +   \ketbra{1}{1}\right) \otimes \ketbra{v_1}{v_1} + \frac{1}{2^\alpha}\left(\frac{1+g_\score}{2}\right)^{2-\alpha}\left( \ketbra{0}{0} +   \ketbra{1}{1}\right) \otimes \ketbra{v_2}{v_2} \, .
\end{align}
We note that one could alternatively have calculated $\sigma_{AE}^{\alpha}\sigma_{E}^{1-\alpha}$. Using this, we can express the Petz-{\Renyi} entropy as
\begin{align}\label{Eq: PetzInterimExpression}
\bar{H}^{\downarrow}_\alpha(A|E)_{\sigma} &=  \frac{1}{1-\alpha}\log\left[ \operatorname{Tr} \left[\sigma_{AE}^{\alpha}\sigma_{E}^{1-\alpha} \right] \right] \\
&=  \frac{1}{1-\alpha}\log\left[ \operatorname{Tr}\left[ \sigma_{E}^{\frac{1-\alpha}{2}}\sigma_{AE}^{\alpha}\sigma_{E}^{\frac{1-\alpha}{2}}\right] \right] \\
&= \frac{1}{1-\alpha}\log \left[  \frac{1}{2^{\alpha-1}}\left(\frac{1-g_\score}{2}\right)^{2-\alpha} +  \frac{1}{2^{\alpha-1}}\left(\frac{1+g_\score}{2}\right)^{2-\alpha}\right] \\
&= 1 + \frac{1}{1-\alpha}\log \left[ \left(\frac{1-g_\score}{2}\right)^{2-\alpha} + \left(\frac{1+g_\score}{2} \right)^{2-\alpha}\right]
\end{align}
This concludes the calculations for the individual qubit block related to some index $I=i$. We now explicitly write the index for rest of the calculation. Let $h(\score^i_\beta)$ denote the function
\begin{align}
    h(\score^i_\beta) \coloneqq \left(\frac{1-g_\score^i}{2}\right)^{2-\alpha} + \left(\frac{1+g_\score^i}{2} \right)^{2-\alpha} \, .
\end{align}
Using both Eq.~\eqref{Eq: SingleISimplification} and~\cite[Prop.~5.1]{Tomamichel2015QuantumIP}, it holds that
\begin{align}
\bar{H}^{\downarrow}_\alpha(A|X=0,IE)_{\rho} &\geq \frac{1}{1-\alpha}\log \left[\sum_{i} \operatorname{Pr} \left[ I=i\right] 2^{\left(1-\alpha\right)\bar{H}^{\downarrow}_\alpha(A|I=i,E)_{\sigma^i}} \right]\\
&= 1+ \frac{1}{1-\alpha}\log \left[ \sum_{i} \operatorname{Pr} \left[ I=i\right] h(\score^i_\beta)\right] \, .
\end{align}
We show in Appendix~\ref{app: ConcavityProperties} that $h(\score^i_{\beta})$ is concave for $\alpha > 1$. The desired lower bound then follows from this property, together with the monotonicity of the logarithm, i.e. 
\begin{align}
\bar{H}^{\downarrow}_\alpha(A|X=0,IE)_{\rho} &\geq 1+ \frac{1}{1-\alpha}\log \left[ h(\score_{\beta})\right] \\
&= 1 + \frac{1}{1-\alpha}\log \left[ \left(\frac{1-g_\score}{2}\right)^{2-\alpha} + \left(\frac{1+g_\score}{2} \right)^{2-\alpha}\right] \, ,
\end{align}
where $g_\score = \sqrt{\tfrac{\score^2_\beta }{4} - \beta^2}$ and $\score_{\beta} = \operatorname{Pr} \left[ I=i\right] \score^i_{\beta}$. This inequality is saturated if Alice and Bob share the state
\begin{align}  \sqrt{P_+}\ket{\phi^+}_{Q_AQ_B}\ket{0}_{E}+\sqrt{P_-}\ket{\phi^-}_{Q_AQ_B}\ket{1}_{E} \, ,
		\end{align}
		where $P_\pm = \tfrac{1}{2}\left(1 \pm g_{\score}\right)$,
 and measure the observables $A_0 = \sigma_{z}$, $A_1 = \sigma_{x}$, $B_0 = \frac{\beta\sigma_{z} + g_{\score}\sigma_{x}}{\sqrt{\beta^2+g_{\score}^2}}$, $B_1 = \frac{\beta\sigma_{z} - g_{\score}\sigma_{x}}{\sqrt{\beta^2+g_{\score}^2}}$ (see Appendix~\ref{Sec: TightnessBounds} for more details).  The derived rate function is thus tight. 
\end{proof}

\subsection{Derivation of \texorpdfstring{$f_{\bar{H}^{\uparrow}_\alpha}(\score_{\beta})$}{Uparrow Petz-{\Renyi} Entropy}}

\begin{thm}
		Let $\alpha \in   (1,2)$, $\lvert \beta \rvert \geq 1$,  and $\score_\beta \in \left[2\lvert \beta \rvert, 2\sqrt{1+\beta^2} \right]$. Then
		\begin{align}
			f_{\bar{H}^{\uparrow}_\alpha}(\score_\beta) = 1 + \frac{\alpha}{1-\alpha}\log \left[ \left(\frac{1-g_\score}{2}\right)^{\frac{1}{\alpha}} + \left(\frac{1+g_\score}{2} \right)^{\frac{1}{\alpha}}\right]  \;,
		\end{align}
        where $g_\score = \sqrt{\tfrac{\score^2_\beta }{4} - \beta^2}$.
	\end{thm}
\begin{proof}
    From \cite[Lemma~5.1]{Tomamichel2015QuantumIP}, we know that
    \begin{align}
        \bar{H}^{\uparrow}_{\alpha}\left(A|E\right)_\sigma = \frac{\alpha}{1-\alpha} \log \left[\operatorname{Tr}\left[ \operatorname{Tr}_A\left(\sigma^\alpha_{AE}\right)^{\frac{1}{\alpha}}\right] \right] \, .
    \end{align}
   For any $\rho_{IAE}$ of the form given by Eq.~\eqref{Eq: AllISigma}, we first consider Alice's and Eve's bipartite state for some $I=i$. Recall that, for any $\sigma_{AE}$ (we omit the index $i$ for now) as in Eq.\ \eqref{Eq: SingleISigma}, one can readily verify that
    \begin{align}
         \sigma_{AE}^{\alpha} &= \frac{1}{2^\alpha} \ketbra{0}{0} \otimes \ketbra{\psi_{=}}{\psi_{=}} +  \frac{1}{2^\alpha} \ketbra{1}{1} \otimes \ketbra{\psi_{\neq}}{\psi_{\neq}} \\
         \operatorname{Tr}_A\left(\sigma^\alpha_{AE}\right)&= \frac{1}{2^\alpha}  \ketbra{\psi_{=}}{\psi_{=}} +  \frac{1}{2^\alpha}\ketbra{\psi_{\neq}}{\psi_{\neq}}= \frac{1}{2^{\alpha-1}}\left(\frac{1-g_\score}{2} \ketbra{v_1}{v_1} + \frac{1+g_\score}{2} \ketbra{v_2}{v_2}\right) \, ,
    \end{align}
    where $\{ \ket{\psi_{=}},\ket{\psi_{\neq}} \}$ are given by Eqs.~\eqref{Eq: psieq}--\eqref{Eq: psineq} and $\{ \ket{v_1},\ket{v_2} \}$  by Eqs.~\eqref{Eq: v1}--\eqref{Eq: v2}.
    \begin{align}
       \bar{H}^{\uparrow}_{\alpha}\left(A|E\right)_\sigma & = \frac{\alpha}{1-\alpha} \log \left[\operatorname{Tr}\left[ \operatorname{Tr}_A\left(\sigma^\alpha_{AE}\right)^{\frac{1}{\alpha}}\right] \right]\\
        &= \frac{\alpha}{1-\alpha} \log \left[ 2^{\frac{1-\alpha}{\alpha}} \left(\left( \frac{1-g_\score}{2}\right)^{\frac{1}{\alpha}} + \left( \frac{1+g_\score}{2}\right)^{\frac{1}{\alpha}} \right) \right] \\
        &= 1+ \frac{\alpha}{1-\alpha} \log \left[\left( \frac{1-g_\score}{2}\right)^{\frac{1}{\alpha}} + \left( \frac{1+g_\score}{2}\right)^{\frac{1}{\alpha}}\right] \, .
    \end{align}
    This concludes the calculations for the individual qubit block related to some index $I=i$. We now explicitly write the index for rest of the calculation. Let $h(\score^i_{\beta})$ denote the function
    \begin{align}
        h(\score^i_{\beta}) \coloneqq   \left(\frac{1-g_\score^i}{2}\right)^{\frac{1}{\alpha}} + \left(\frac{1+g_\score^i}{2} \right)^{\frac{1}{\alpha}}\, .
    \end{align}
    Using both Eq.~\eqref{Eq: SingleISimplification} and~\cite[Prop.~5.1]{Tomamichel2015QuantumIP}, it holds that
    \begin{align}
        \bar{H}^{\uparrow}_\alpha(A|X=0,IE)_{\rho} &\geq \frac{\alpha}{1-\alpha}\log \left[\sum_{i} \operatorname{Pr} \left[ I=i\right] 2^{\frac{\left(1-\alpha\right)}{\alpha}\bar{H}^{\uparrow}_\alpha(A|I=i,E)_{\sigma^i}} \right]\\
        &= 1+ \frac{\alpha}{1-\alpha}\log \left[\sum_{i} \operatorname{Pr} \left[ I=i\right] h(\score^i_{\beta}) \right]\, .
    \end{align}
    We show in Appendix~\ref{app: ConcavityProperties} that $h(\score^i_{\beta})$ is concave for $\alpha> 1$. The desired lower bound then follows from this property, together with the monotonicity of the logarithm, i.e. 
    \begin{align}
        \bar{H}^{\uparrow}_\alpha(A|X=0,IE)_{\rho} &\geq 1+ \frac{\alpha}{1-\alpha}\log \left[ h(\score_{\beta}) \right] \\
        &= 1 + \frac{\alpha}{1-\alpha}\log \left[ \left(\frac{1-g_\score}{2}\right)^{\frac{1}{\alpha}} + \left(\frac{1+g_\score}{2} \right)^{\frac{1}{\alpha}}\right] \, ,
    \end{align}
   where $g_\score = \sqrt{\tfrac{\score^2_\beta }{4} - \beta^2}$ and $\score_{\beta} = \operatorname{Pr} \left[ I=i\right] \score^i_{\beta}$. This inequality is saturated if Alice and Bob share the state
    \begin{align}  \ket{\psi} =\sqrt{P_+}\ket{\phi^+}_{Q_AQ_B}\ket{0}_{E}+\sqrt{P_-}\ket{\phi^-}_{Q_AQ_B}\ket{1}_{E} \, ,
    		\end{align}
    		where $P_\pm = \tfrac{1}{2}\left(1 \pm g_{\score}\right)$,
     and measure the observables $A_0 = \sigma_{z}$, $A_1 = \sigma_{x}$, $B_0 = \frac{\beta\sigma_{z} + g_{\score}\sigma_{x}}{\sqrt{\beta^2+g_{\score}^2}}$, $B_1 = \frac{\beta\sigma_{z} - g_{\score}\sigma_{x}}{\sqrt{\beta^2+g_{\score}^2}}$ (see Appendix~\ref{Sec: TightnessBounds} for more details).  The derived rate function is thus tight. 
\end{proof}
\subsection{Monotonicity and Concavity Properties} \label{app: ConcavityProperties}
In this section, we aim to prove the monotonicity and concavity of the following three functions for certain regions of $\alpha>1$ and $|\beta| \geq 1$:
\begin{align}
 h_{1}(\score_{\beta}) &= \left[ \left(\frac{1-g_\score}{2}\right)^{\frac{1}{\alpha}} + \left(\frac{1+g_\score}{2} \right)^{\frac{1}{\alpha}}\right]^{\alpha} \\
 h_{2}(\score_{\beta}) &=\left(\frac{1-g_\score}{2}\right)^{2-\alpha} + \left(\frac{1+g_\score}{2} \right)^{2-\alpha} \\
  h_{3}(\score_{\beta}) &=  \left(\frac{1-g_\score}{2}\right)^{\frac{1}{\alpha}} + \left(\frac{1+g_\score}{2} \right)^{\frac{1}{\alpha}}  \, ,
\end{align}
where $g_\score = \sqrt{\tfrac{\score^2_\beta }{4} - \beta^2}$. These functions can alternatively be expressed as 
\begin{align}
     h_{i}(\score_{\beta}) = f_{i}(\bar{g}(\score_{\beta})) \, ,
\end{align}
where 
\begin{align}
    f_{1}(x) &= \left[ \left(\frac{1-\sqrt{x}}{2}\right)^{\frac{1}{\alpha}} + \left(\frac{1+\sqrt{x}}{2} \right)^{\frac{1}{\alpha}}\right]^{\alpha} \\
 f_{2}(x) &=\left(\frac{1-\sqrt{x}}{2}\right)^{2-\alpha} + \left(\frac{1+\sqrt{x}}{2} \right)^{2-\alpha} \\
  f_{3}(x) &=  \left(\frac{1-\sqrt{x}}{2}\right)^{\frac{1}{\alpha}} + \left(\frac{1+\sqrt{x}}{2} \right)^{\frac{1}{\alpha}}  \, ,
\end{align}
and $\bar{g}(\score_{\beta}) =\tfrac{\score^2_\beta }{4} - \beta^2$. The first derivative of $h(\score_{\beta})$ is given by $h_{i}^\prime(\score_{\beta}) = f_{i}^\prime(\bar{g} (\score_{\beta})) \cdot \bar{g}^\prime (\score_{\beta})$. Since $\bar{g}^\prime (\score_{\beta}) > 0$ in the regime $\score_\beta \in \left[2\lvert \beta \rvert, 2\sqrt{1+\beta^2} \right]$, the monotonicity of $h_{i}(\score)$ is purely determined by $f_{i}(x)$. We now show that $f_{i}^\prime (x) \leq 0$ in the relevant regimes $\alpha >1$. In particular this ensures that the rate functions are monotonically increasing, a property which is relevant for the discussion in Appendix~\ref{Sec: Qubit Reduction}.
\begin{prop}
 For all $x \in [0,1]$, the functions $f_{1}(x)$ and $f_{3}(x)$ are monotonically decreasing for $\alpha \in (1,\infty)$. Moreover, for all $x \in [0,1]$, the function $f_{2}(x)$ is monotonically decreasing for $\alpha \in (1,2)$.
\end{prop}
\begin{proof}
    Note that because $f_1(x) = f_3(x)^{\alpha}$, 
    \begin{align}
        f_1^\prime(x) = \alpha f_3(x)^{\alpha-1} f_3^\prime(x) \, .
    \end{align}
However, because $f_3(x) \geq 0$, $f_1^\prime(x)$ must have the same sign as $f_3^\prime(x)$. We will now prove that $f_3^\prime(x) \leq 0$ for all $x \in (0,1)$. The continuity of $f_3(x)$ then ensures that $f_3^\prime(x)$ is monotonic over the entire regime $x \in [0,1]$. For $\alpha > 1$, it follows that
\begin{align}
    f_3^\prime(x) &= \frac{\left(1+\sqrt{x}\right)^{\frac{1}{\alpha}-1}-\left(1-\sqrt{x}\right)^{\frac{1}{\alpha}-1}}{2\alpha\sqrt{x}} \\
    &\leq 0 \, ,
\end{align}
thus ensuring monotonicity. For $f_2(x)$, one can either prove monotonicity by noting that it is equivalent to $f_3(x)$ after a suitable modification of the {\Renyi} parameter $\alpha$, or alternatively one can see it via 
\begin{align}
    f_2^\prime(x) &= \frac{\left(2-\alpha\right)\left(\left(1+\sqrt{x}\right)^{1-\alpha}-\left(1-\sqrt{x}\right)^{1-\alpha}\right)}{2\alpha\sqrt{x}} \\
    &\leq 0 \, ,
\end{align}  
which holds for all $x \in (0,1)$. Again one can extend the result to $x \in [0,1]$ via a continuity argument.
\end{proof}
As was argued in~\cite{Woodhead_2021}, to prove concavity of $h_{i}(\score_{\beta})$, it is sufficient to prove that $f_{i}(x)$ is concave and monotonically decreasing and that $\bar{g}(\score_{\beta})$ is convex. Since $\bar{g}(\score_{\beta})$ is clearly convex and we have shown that all $f_{i}(x)$ are decreasing, it simply remains to show that each $f_{i}(x)$ is concave.
\begin{prop} \label{Prop: ConcavF3}
    For all $x \in [0,1]$, the function
    \begin{align}
        f_{3}(x) = \left(1-\sqrt{x}\right)^{\frac{1}{\alpha}} + \left(1+\sqrt{x}\right)^{\frac{1}{\alpha}}
    \end{align}
    is concave for $\alpha \in (1,\infty)$.
\end{prop}
\begin{proof}
    Due to continuity arguments, it is sufficient to prove it for all $x \in (0,1)$. The second derivative of this function is given by
    \begin{align}
        f_{3}^{\prime \prime} (x) &= \frac{\left(1-\sqrt{x}\right)^{\frac{1}{\alpha}-1}}{4\alpha x^{\frac{3}{2}}} - \frac{\left(1+\sqrt{x}\right)^{\frac{1}{\alpha}-1}}{4\alpha x^{\frac{3}{2}}} - \frac{\left(\alpha-1\right)\left(\left(1+\sqrt{x}\right)^{\frac{1}{\alpha}-2}+\left(1-\sqrt{x}\right)^{\frac{1}{\alpha}-2}\right)}{4\alpha^2 x} \\
        &= \frac{1}{4\alpha x^{\frac{3}{2}}} \left( \left(1-\sqrt{x}\right)^{\frac{1}{\alpha}-2} \left( 1-\sqrt{x} - \frac{\alpha-1}{\alpha}\sqrt{x}\right)-\left(1+\sqrt{x}\right)^{\frac{1}{\alpha}-2} \left( 1+\sqrt{x} + \frac{\alpha-1}{\alpha}\sqrt{x}\right)\right) \\
        &=\frac{1}{4\alpha x^{\frac{3}{2}}} \left( \left(1-\sqrt{x}\right)^{\frac{1}{\alpha}-2} \left( 1 - \left(2-\frac{1}{\alpha}\right)\sqrt{x}\right)-\left(1+\sqrt{x}\right)^{\frac{1}{\alpha}-2} \left( 1+\left( 2-\frac{1}{\alpha}\right)\sqrt{x}\right)\right) \, .
    \end{align}
    Note that, for any $\alpha>1$ and $y \in \left[ 0,1\right)$, 
   \begin{align}
       (1-y)^{2-\frac{1}{\alpha}} &\geq 1 - \left( 2-\frac{1}{\alpha}\right)y + \frac{\left(\alpha-1\right)\left(2\alpha-1\right)}{2\alpha^2}y^2  \label{Eq:TaylorPetzup1}\\
        (1+y)^{2-\frac{1}{\alpha}} &\leq 1 - \left( 2+\frac{1}{\alpha}\right)y + \frac{\left(\alpha-1\right)\left(2\alpha-1\right)}{2\alpha^2}y^2  \, . \label{Eq:TaylorPetzup2}
   \end{align}
   The right-hand-side are simply Taylor approximations of the left-hand-side. These inequalities follow from the fact that the error arising from these Taylor approximations are proportional to the third derivatives of the left-hand-side at some value $\xi \in [0,y]$; see e.g.~\cite[Eq.~(1.3.2)]{Hildebrand56}. The third derivative, i.e.
   \begin{align}
       \frac{d^3}{dy^3} (1-y)^{2-\frac{1}{\alpha}} &= \frac{\left(2\alpha-1\right)\left(\alpha-1\right)}{\alpha^3}\left(1-y\right)^{-1-\frac{1}{\alpha}} \\
       \frac{d^3}{dy^3} (1+y)^{2-\frac{1}{\alpha}} &= - \frac{\left(2\alpha-1\right)\left(\alpha-1\right)}{\alpha^3}\left(1+y\right)^{-1-\frac{1}{\alpha}} \, ,
   \end{align}
 are positive and negative, respectively, for all $y \in [0,1)$, thus ensuring that the bounds hold. Using this and setting $y =\sqrt{x}$, we find that for all $\alpha>1$,
   \begin{align}
       f_{3}^{\prime \prime} (x) &\leq \frac{1}{4\alpha x^{\frac{3}{2}}}\left(1-\sqrt{x}\right)^{\frac{1}{\alpha}-2} \left( \left(1-\sqrt{x}\right)^{2-\frac{1}{\alpha}}-\frac{\left(\alpha-1\right)\left(2\alpha-1\right)}{2\alpha^2}x\right)\notag\\
       &\qquad\qquad\qquad\qquad-\frac{1}{4\alpha x^{\frac{3}{2}}}\left(1+\sqrt{x}\right)^{\frac{1}{\alpha}-2} \left( \left(1+\sqrt{x}\right)^{2-\frac{1}{\alpha}}- \frac{\left(\alpha-1\right)\left(2\alpha-1\right)}{2\alpha^2}x\right) \\
       &= \frac{1}{4\alpha x^{\frac{3}{2}}}  \left(\left( 1-\left(1-\sqrt{x}\right)^{\frac{1}{\alpha}-2}\frac{\left(\alpha-1\right)\left(2\alpha-1\right)}{2\alpha^2}x\right)-\left( 1-\left(1+\sqrt{x}\right)^{\frac{1}{\alpha}-2} \frac{\left(\alpha-1\right)\left(2\alpha-1\right)}{2\alpha^2}x\right) \right) \\
       &= \frac{1}{4\alpha x^{\frac{3}{2}}} \frac{\left(\alpha-1\right)\left(2\alpha-1\right)}{2\alpha^2}x \left( \left(1+\sqrt{x}\right)^{\frac{1}{\alpha}-2} - \left(1-\sqrt{x}\right)^{\frac{1}{\alpha}-2}\right) \\
       &\leq 0 \, .
   \end{align}
   This concludes the concavity proof for $\alpha>1$.

   %For $\frac{1}{2}\leq \alpha <1$, the inequalities in Eq.~\eqref{Eq:TaylorPetzup1} and Eq.~\eqref{Eq:TaylorPetzup2} are reversed. This can then be used to analogously prove $f^{\prime \prime} (x) \geq 0$.
   
   %The same proof technique, however, cannot be used to prove convexity for $0 \leq \alpha \leq \frac{1}{2}$. This can be seen by the fact that the inequalities in Eq.~\eqref{Eq:TaylorPetzup1} and Eq.~\eqref{Eq:TaylorPetzup2} again flip sign. Moreover, in this regime, one has that $\left(1+\sqrt{x}\right)^{\frac{1}{\alpha}-2} \geq \left(1-\sqrt{x}\right)^{\frac{1}{\alpha}-2}$.
\end{proof}
\begin{prop}
For all $x \in [0,1]$, the function
    \begin{align}
        f_{1}(x) = \left(\left(1-\sqrt{x}\right)^{\frac{1}{\alpha}} + \left(1+\sqrt{x}\right)^{\frac{1}{\alpha}}\right)^{\alpha}
    \end{align}
    is concave for $\alpha \in (1,\infty)$.
\end{prop}
\begin{proof}
     Due to continuity arguments, it is sufficient to prove it for all $x \in (0,1)$. The second derivative of this function is given by
\begin{multline}
    f_{1}^{\prime \prime}(x) =  \frac{k(x)}{(1-x)^2}\left( \left( 1-x\right)^{\frac{1}{\alpha}} \left(\left(4-2\alpha \right)\sqrt{x}-2\alpha x^{\frac{3}{2}}\right)-\alpha \left(1-\sqrt{x}\right)^{\frac{2}{\alpha}} \left( -1- \sqrt{x}+ x+ x^{\frac{3}{2}}\right)\right.\\
    \left.-\alpha \left(1+\sqrt{x}\right)^{\frac{2}{\alpha}} \left( 1- \sqrt{x}- x+ x^{\frac{3}{2}}\right)
    \right) \, ,
\end{multline}
where $k(x) = \frac{\left( \left(1-\sqrt{x}\right)^{\frac{1}{\alpha}}+\left(1+\sqrt{x}\right)^{\frac{1}{\alpha}}\right)^{\alpha-2}}{4\alpha x^{\frac{3}{2}}} \geq 0$. This can be simplified as follows.
\begin{align}
    f_{1}^{\prime \prime}(x) &=  \frac{k(x)}{(1-x)^2}\left( \left( 1-x\right)^{\frac{1}{\alpha}} \left(\left(4-2\alpha \right)\sqrt{x}-2\alpha x^{\frac{3}{2}}\right)+\alpha\left(1-\sqrt{x}\right)^{\frac{2}{\alpha}} \left( 1+\sqrt{x}\right) \left( 1 -x\right)-\alpha \left(1+\sqrt{x}\right)^{\frac{2}{\alpha}} \left(1-\sqrt{x}\right)\left( 1- x\right)
    \right)\notag \\
    &= \frac{k(x)}{(1-x)^2}\left( \left( 1-x\right)^{\frac{1}{\alpha}} \left(\left(4-2\alpha \right)\sqrt{x}-2\alpha x^{\frac{3}{2}}\right)+\alpha\left( 1 -x\right)^2 \left(\left(1-\sqrt{x}\right)^{\frac{2}{\alpha}-1} - \left(1+\sqrt{x}\right)^{\frac{2}{\alpha}-1} \right)
    \right) \\
    &= k(x)\left( \left( 1-x\right)^{\frac{1}{\alpha}-2} \left(\left(4-2\alpha \right)\sqrt{x}-2\alpha x^{\frac{3}{2}}\right)+\alpha \left(\left(1-\sqrt{x}\right)^{\frac{2}{\alpha}-1} - \left(1+\sqrt{x}\right)^{\frac{2}{\alpha}-1} \right)
    \right)\label{eqn:borrowing_2}
\end{align}
For the regime $\alpha \geq 2$, we find that for any $y \in \left[ 0,1\right)$
\begin{align}
    \alpha \left( 1-y^2\right)^{2-\frac{1}{\alpha}}\left(\left(1+y\right)^{\frac{2}{\alpha}-1} - \left(1-y\right)^{\frac{2}{\alpha}-1} \right) \geq \left(4-2\alpha \right)y \, . \label{Eq: Downarrowalphageq2taylor}
\end{align}
Note that equality holds when $y=0$. The inequality then follows from the fact that
\begin{align}
    &\frac{d}{dy} \left(\alpha \left( 1-y^2\right)^{2-\frac{1}{\alpha}}\left(\left(1+y\right)^{\frac{2}{\alpha}-1} - \left(1-y\right)^{\frac{2}{\alpha}-1} \right) -\left(4-2\alpha \right)y \right)\\
    &= \left( 1-y^2\right)^{1-\frac{1}{\alpha}} \left( \left( 3\alpha y -\alpha +2\right)\left( 1-y\right)^{\frac{2}{\alpha}-1} - \left( 3\alpha y +\alpha -2\right)\left( 1+y\right)^{\frac{2}{\alpha}-1}\right) - \left( 4-2\alpha\right) \\
    &\geq \left( 1-y^2\right)^{1-\frac{1}{\alpha}} \left( \left(  2-\alpha \right)\left( 1-y\right)^{\frac{2}{\alpha}-1} + \left( 2-\alpha \right)\left( 1+y\right)^{\frac{2}{\alpha}-1}\right) - \left( 4-2\alpha\right) \\
    &= \left( 2-\alpha\right)\left( 1-y^2\right)^{1-\frac{1}{\alpha}} \left( \left( 1-y\right)^{\frac{2}{\alpha}-1} +\left( 1+y\right)^{\frac{2}{\alpha}-1}\right) - \left( 4-2\alpha\right) \\
    &\geq \left( 2-\alpha\right)\left( 1-y^2\right)^{1-\frac{2}{\alpha}} \left( \left( 1-y\right)^{\frac{2}{\alpha}-1} +\left( 1+y\right)^{\frac{2}{\alpha}-1}\right) - \left( 4-2\alpha\right) \\
    &= \left( 2-\alpha\right) \left( \left( 1+y\right)^{1-\frac{2}{\alpha}} +\left( 1-y\right)^{1-\frac{2}{\alpha}}\right) - \left( 4-2\alpha\right) \\
    &\geq 2\left( 2-\alpha\right)  - \left( 4-2\alpha\right)\\
    &= 0 \, .
\end{align}
The second-to-last line holds with equality for $y=0$. For $y >0$, the inequality holds because ${\left( 1+y\right)^{1-\frac{2}{\alpha}} +\left( 1-y\right)^{1-\frac{2}{\alpha}} \geq 0}$ and 
\begin{align}
    \frac{d}{dy} \left(\left( 1+y\right)^{1-\frac{2}{\alpha}} +\left( 1-y\right)^{1-\frac{2}{\alpha}} \right) &= - \frac{\left(\alpha -2 \right) \left( \left( 1-y\right)^{-\frac{2}{\alpha}} - \left( 1+y\right)^{-\frac{2}{\alpha}} \right)}{\alpha} \\
    &\leq 0 \, .
\end{align}
Using Eq.~\eqref{Eq: Downarrowalphageq2taylor} and setting $y = \sqrt{x}$, it follows that
\begin{align}
    f_{1}^{\prime \prime}(x) &\leq -2\alpha x^{\frac{3}{2}}k(x) \left( 1-x\right)^{1/a-2} \\
    &\leq 0 \, .
\end{align}
For $1<\alpha \leq 2$, we instead use the inequality 
\begin{align}
    \alpha \left( 1-y^2\right)^{2-\frac{1}{\alpha}}\left(\left(1+y\right)^{\frac{2}{\alpha}-1} - \left(1-y\right)^{\frac{2}{\alpha}-1} \right) \geq \left(4-2\alpha \right)y + \frac{2\left(\alpha-2\right)\left(3\alpha^2+2\alpha-2\right)}{3\alpha^2}y^3 \, .
\end{align}
The right-hand-side should be viewed as a third-order Taylor approximation, and the bound holds because the error arising from this Taylor approximation is proportional to the fourth derivative at some $\xi \in [0,y]$, see~\cite[Eq.~(1.3.2)]{Hildebrand56}, and
\begin{align}
    &\frac{d^4}{dy^4} \left(\alpha \left( 1-y^2\right)^{2-\frac{1}{\alpha}}\left(\left(1+y\right)^{\frac{2}{\alpha}-1} - \left(1-y\right)^{\frac{2}{\alpha}-1} \right)\right) \\
    &= \frac{16\left( \alpha -1\right)\left( \alpha +1\right)\left( 2\alpha -1\right)\left( 1-y^2\right)^{-\frac{1}{\alpha}-2}\left( \left( 1+y\right)^{\frac{2}{\alpha}-1}-\left( 1-y\right)^{\frac{2}{\alpha}-1}\right)}{a^4} \\
    &\geq 0 \,.
\end{align}
It then holds that
\begin{align}
    f_{1}^{\prime \prime}(x) &\leq  x^{\frac{3}{2}}k(x) \left( 1-x\right)^{1/a-2}\left(\frac{2\left(2-\alpha\right)\left(3\alpha^2+2\alpha-2\right)}{3\alpha^2} -2\alpha \right) \leq 0 \, .
\end{align}
The last inequality follows from the fact that 
\begin{align}
2\left(2-\alpha\right)\left(3\alpha^2+2\alpha-2\right) = 6\alpha^3
\end{align}
has the three solutions $\alpha \in \{ -1, \frac{2}{3},1\}$. The fraction 
\begin{align}
    \left(\frac{2\left(2-\alpha\right)\left(3\alpha^2+2\alpha-2\right)}{3\alpha^2} -2\alpha \right)
\end{align}
must thus have the same sign for all $1<\alpha \leq 2$. Explicitly verifying the sign for one such $\alpha$ then suffices to prove the claim. This concludes the proof that $f_{1}(x)$ is concave for all $\alpha >1$.
\end{proof}

\begin{prop}
For all $x \in [0,1]$, the function
    \begin{align}
        f_{2}(x) = \left(1-\sqrt{x}\right)^{2-\alpha} + \left(1+\sqrt{x}\right)^{2-\alpha}
    \end{align}
    is concave for $\alpha \in (1,2)$.
\end{prop}
\begin{proof}
For this function, one can either prove concavity by noting that this statement is equivalent to Proposition~\ref{Prop: ConcavF3} after a suitable modification of the {\Renyi} parameter $\alpha$, or alternatively one can prove it as follows. Due to continuity arguments, it is sufficient to prove it for all $x \in (0,1)$.
The second derivative is given by
\begin{align}
    f_{2}^{\prime \prime}(x) =  \frac{2-\alpha}{4x^{\frac{3}{2}}} \left[\left(1-\alpha \sqrt{x} \right) \left( 1-\sqrt{x}\right)^{-\alpha}   - \left(1+\alpha \sqrt{x} \right) \left( 1+\sqrt{x}\right)^{-\alpha}\right] \, .
\end{align}   
Note that, for any $\alpha \in (1,2)$ and $y \in \left[ 0,1\right)$, 
   \begin{align}
          (1-y)^{\alpha} &\geq 1 - \alpha y + \frac{\alpha \left(\alpha +1 \right)}{2}y^2   \label{Eq: Petzdown2ineq} \\
       (1+y)^{\alpha} &\leq 1 + \alpha y + \frac{\alpha \left(\alpha +1 \right)}{2}y^2 \label{Eq: Petzdown1ineq} \, .
   \end{align}
   The right-hand-side are simply Taylor approximations of the left-hand-side,  and the bound holds because the errors arising from such Taylor approximations are proportional to the third derivative at some $\xi \in [0,y]$, see~\cite[Eq.~(1.3.2)]{Hildebrand56}, and
   \begin{align}
   \frac{d^3}{dy^3} (1-y)^{\alpha} &=  \alpha \left( \alpha-1 \right) \left( 2-\alpha\right)\left(1+y\right)^{\alpha -3} \\
       \frac{d^3}{dy^3} (1+y)^{\alpha} &= -\alpha \left( \alpha-1 \right) \left( 2-\alpha\right)\left(1-y\right)^{\alpha -3} 
       \, ,
   \end{align}
   which are positive and negative, respectively, for all $y \in [0,1)$. Setting $y = \sqrt{x}$ and using the inequalities from Eqs.~\eqref{Eq: Petzdown1ineq}--\eqref{Eq: Petzdown2ineq}, then yields 
   \begin{align}
        f_{2}^{\prime \prime}(x) &\leq \frac{\left(2-\alpha\right)\alpha \left(\alpha +1 \right)}{8\sqrt{x}\left(1-x \right)^{\alpha}}  \left[  \left( 1-\sqrt{x}\right)^{\alpha} - \left( 1+\sqrt{x}\right)^{\alpha}\right] \\
        &\leq 0
   \end{align}
This concludes the proof that $f_{2}(x)$ is concave for all $\alpha  \in \left(1,2\right)$.
\end{proof}
\subsection{Tightness of Rate Bounds} \label{Sec: TightnessBounds}

 It can be shown that all inequalities are tight by considering the following attack which saturates the bound. First, it is easy to verify that measuring that state
 \begin{align}  \sqrt{P_+}\ket{\phi^+}_{Q_AQ_B}\ket{0}_{E}+\sqrt{P_-}\ket{\phi^-}_{Q_AQ_B}\ket{1}_{E} \, ,  \label{Eq: OptAttackEve}
		\end{align}
 where $P_\pm = \tfrac{1}{2}\left(1 \pm g_{\score}\right)$ and $g_\score = \sqrt{\tfrac{\score^2_\beta }{4} - \beta^2}$, in the observables $A_0 = \sigma_{z}$, $A_1 = \sigma_{x}$, $B_0 = \frac{\beta\sigma_{z} + g_{\score}\sigma_{x}}{\sqrt{\beta^2+g_{\score}^2}}$, $B_1 = \frac{\beta \sigma_{z} - g_{\score}\sigma_{x}}{\sqrt{\beta^2+g_{\score}^2}}$ achieves a score of $\score_\beta$.
 Moreover, if Alice measures this state
 in the observable $A_0$, one finds that the post-measurement state is given by
\begin{align}\label{def:attack_tight}
    \rho_{AE} = \frac{1}{2} \ketbra{0}{0}_{A} \otimes \ketbra{\psi_0}{\psi_0}  + \frac{1}{2} \ketbra{1}{1}_{A} \otimes \ketbra{\psi_1}{\psi_1}  \, ,
\end{align}
 where
\begin{align}
    \ket{\psi_0} & = \sqrt{P_+} \ket{0} + \sqrt{P_-} \ket{1} \\
   \ket{\psi_1} & = \sqrt{P_+} \ket{0} - \sqrt{P_-} \ket{1} \, . 
\end{align}
In the basis
\begin{align}
    \ket{0^\prime} &\coloneqq  \ket{\psi_0} = \sqrt{P_+} \ket{0} + \sqrt{P_-} \ket{1}\\
    \ket{1^\prime} &\coloneqq  \sqrt{P_-} \ket{0} - \sqrt{P_+} \ket{1} \, ,
\end{align}
one finds that
\begin{align}
   \ket{\psi_0} &= \ket{0^\prime} \\
 \ket{\psi_1} &= |\braket{\psi_0 | \psi_1}|\ket{0^\prime} + \sqrt{1-|\braket{\psi_0 | \psi_1}|^2}\ket{1^\prime} \, ,
\end{align}

where
\begin{align}
    |\braket{\psi_0 | \psi_1}| &= P_+ -P_- \\
    &=g_\score \, .
\end{align}
This post-measurement state thus has the same form as those from Eq.~\eqref{Eq: SingleISigma}. However, it is precisely on these states that we prove explicit bounds on $\mathbb{H}\left(A|E \right)$. Whenever we provide an exact expression for $\mathbb{H}\left(A|E \right)$ for these states, the same proof must also hold for Eq.~\eqref{Eq: OptAttackEve}. The only key difference is that the score of Eq.~\eqref{Eq: OptAttackEve} is already $\score_{\beta}$. We note that although one can technically calculate $\widetilde{H}^{\uparrow}_\alpha (A|E)$ for such states, we use a slight variation of this argument to prove tightness in Appendix~\ref{Uparrow Sandwiched {\Renyi} Entropy}.

The proofs for deriving tight rate functions require that $\mathbb{Q}$ satisfies the properties mentioned in Lemma~\ref{Lem: QProperties} and that the bounds we derive for $\mathbb{Q}$ are monotonically decreasing in $g_\score$. These two conditions are necessary to reduce the analysis to states of the form of Eq.~\eqref{Eq: SingleISigma}. For $\alpha>1$, concavity is necessary to derive an expression which no longer depends on the index $I$. Whenever these properties hold and one derives an exact expression for $\mathbb{H}\left(A|E \right)$ for states of the form of Eq.~\eqref{Eq: SingleISigma}, then the bounds on the corresponding rate function must be tight. We additionally note that if the bound on $\mathbb{Q}$ were not concave for $\alpha >1$, but satisfied the properties of Lemma~\ref{Lem: QProperties} and monotonicity, then one can still achieve tight bounds; however they would depend on the bound's concave envelope. We discuss and explicitly use this property in Appendix~\ref{app:noisypreprocessing}, when incorporating noisy preprocessing into the analysis.

\subsection{Discontinuity Behavior for Edge Cases}\label{app:discont}
In this subsection, we justify the claim that the Theorem~\ref{thm:allcasesnoNPP} bounds can be extended to the right-endpoints of the $\alpha$ ranges by taking the respective limits from below. As briefly mentioned previously, some of the resulting formulas are discontinuous with respect to $\score_{\beta}$ (for $|\beta|\geq1$).
Specifically, $f_{\widetilde{H}^{\downarrow}_\infty}(\score_{\beta})$ and $f_{\bar{H}^{\downarrow}_2}(\score_{\beta})$ have a discontinuity at $\score_{\beta}=2\sqrt{1+\beta^2}$, taking the value $1$ at that point and $0$ elsewhere. We further note that these discontinuities are a genuine property of $\widetilde{H}^{\downarrow}_\infty(A\vert E)_{\rho}$ and $\bar{H}^{\downarrow}_2(A\vert E)_{\rho}$, respectively, and not a result of the methods used to obtain these bounds --- we shall show this by constructing a family of states saturating these discontinuous bounds. 

To prove the claim, we simply note that for all of the {\Renyi} entropies $\mathbb{H}_{\alpha}$ in Definition~\ref{def:condent}, for any $\alpha^\star \in (1,\infty]$ we have $\mathbb{H}_{\alpha^\star} = \lim_{\alpha \nearrow \alpha^\star} \mathbb{H}_{\alpha} = \inf_{\alpha \in (1,\alpha^\star)} \mathbb{H}_{\alpha}$ because they are monotone decreasing with respect to $\alpha$. Given this, we can freely interchange the infimum over $\alpha$ with the infimum over quantum strategies in the definition of the rate functions, from which we can conclude the desired claim $f_{\mathbb{H}_{\alpha^\star}} (\score_\beta) = \inf_{\alpha \in (1,\alpha^\star)} f_{\mathbb{H}_{\alpha}} (\score_\beta)  = \lim_{\alpha \nearrow \alpha^\star} f_{\mathbb{H}_{\alpha}} (\score_\beta) $ (the second equality holds because $f_{\mathbb{H}_{\alpha}}$ immediately inherits the monotonicity in $\alpha$ from $\mathbb{H}_{\alpha}$).%\footnote{By considering feasible points in the optimization over quantum strategies.} (An overly clever alternative is to note that the first equality $f_{\mathbb{H}_{\alpha^\star}} (\score_\beta) = \inf_{\alpha \in (1,\alpha^\star)} f_{\mathbb{H}_{\alpha}} (\score_\beta)$ is itself sufficient to imply it is decreasing in $\alpha^\star$, since the domain of the infimum is becoming strictly larger.)

In fact, for the cases $f_{\widetilde{H}^{\downarrow}_\infty}$ and $f_{\bar{H}^{\downarrow}_2}$, we can instead prove this via a more direct analysis of the optimal attack for each $\score_\beta$, which also shows that the discontinuities are a genuine feature of the bounds. Specifically, to calculate $\widetilde{H}^{\downarrow}_\alpha(A\vert E)_{\rho}$ for the optimal attack from Eq.~\eqref{def:attack_tight}, we need to first calculate $\rho_{E}^{\frac{1-\alpha}{2\alpha}}\rho_{AE}\rho_{E}^{\frac{1-\alpha}{2\alpha}}$. From Eq.~\eqref{Eq: SandwichedInterimExpression}, we see that as $\alpha\rightarrow\infty$ for $\score_{\beta}<2\sqrt{1+\beta^2}$,
\begin{equation}
    \rho_{E}^{\frac{1-\alpha}{2\alpha}}\rho_{AE}\rho_{E}^{\frac{1-\alpha}{2\alpha}}\rightarrow \ketbra{0}{0}_{A} \otimes \ketbra{w_1}{w_1}_{E} + \ketbra{1}{1}_{A} \otimes \ketbra{w_1^\prime}{w_1^\prime}_{E},
\end{equation}
where $\{\ket{w_1},\ket{w_1^\prime} \}$ are normalized vectors. Thus,
\begin{align}
\lim_{\alpha \to \infty}\widetilde{H}^{\downarrow}_\alpha(A|E)_{\rho} &=  \lim_{\alpha \to \infty}\frac{1}{1-\alpha} \cdot \lim_{\alpha \to \infty}\log\left[ \operatorname{Tr} \left[\left(\sigma_{E}^{\frac{1-\alpha}{2\alpha}}\sigma_{AE}\sigma_{E}^{\frac{1-\alpha}{2\alpha}}\right)^\alpha \right] \right] \\
 & = \lim_{\alpha \to \infty}\frac{1}{1-\alpha}\log 2 \\
   &=0 \, ,
\end{align}
where the first lines follows from the fact that both limits are finite and well-defined.
Hence, we get that $\widetilde{H}^{\downarrow}_\infty(A\vert E)_{\rho}=0$ is indeed an achievable bound. Conversely, for the maximal score, one can lower bound $\widetilde{H}^{\downarrow}_\infty(A\vert X=0, IE)_{\rho}$ as follows.  Recall that
\begin{align}
\bar{H}^{\downarrow}_\alpha(A|X=0,IE)_{\rho} &\geq \frac{1}{1-\alpha}\log \left[\sum_{i} \operatorname{Pr} \left[ I=i\right] 2^{\left(1-\alpha\right)\bar{H}^{\downarrow}_\alpha(A|I=i,E)_{\sigma^i}}\right]
\end{align}
were the states $\sigma_{AE}^{i}$ are defined in Eq.~\eqref{Eq: SingleISigma}. Each $\sigma_{AE}^{i}$ has to achieve the maximum score and therefore must be of the form
\begin{align} \label{Eq: MaxScoreSigma}
   \sigma_{AE}^{i} =  \frac{1}{2} \ketbra{0}{0}_{A} \otimes \ketbra{0}{0}_{E} +  \frac{1}{2} \ketbra{1}{1}_{A} \otimes \ketbra{0}{0}_{E} \, .
\end{align}
Using
\begin{equation}
   \left( \sigma_{E}^{i} \right)^{ \frac{1-\alpha}{2\alpha}}\sigma_{AE}^{i}\left( \sigma_{E}^{i} \right)^{\frac{1-\alpha}{2\alpha}}= \frac{1}{2}\ketbra{0}{0}_{A} \otimes \ketbra{0}{0}_{E} + \frac{1}{2}\ketbra{1}{1}_{A} \otimes \ketbra{0}{0}_{E} \, ,
\end{equation}
we get that 
\begin{align}
\lim_{\alpha \to \infty}\widetilde{H}^{\downarrow}_\alpha(A|E)_{\sigma^i} &=  \lim_{\alpha \to \infty}\frac{1}{1-\alpha}\log\left[ \operatorname{Tr} \left[\left(\left( \sigma_{E}^{i} \right)^{ \frac{1-\alpha}{2\alpha}}\sigma_{AE}^{i}\left( \sigma_{E}^{i} \right)^{\frac{1-\alpha}{2\alpha}}\right)^\alpha \right] \right] \\
 & = \lim_{\alpha \to \infty}\frac{1}{1-\alpha}\log 2^{1-\alpha} \\
 & = \lim_{\alpha \to \infty}1 \\
   &=1 \, .
\end{align}
It follows from this that $\widetilde{H}^{\downarrow}_\infty(A\vert X=0, IE)_{\rho}=1$ as well.

To calculate $\bar{H}^{\downarrow}_\alpha(A\vert E)_{\rho}$ for the optimal attack from Eq.~\eqref{def:attack_tight}, we need to calculate $\rho_{E}^{\frac{1-\alpha}{2}}\rho_{AE}^{\alpha}\rho_{E}^{\frac{1-\alpha}{2}}$. From Eq.~\eqref{Eq: PetzInterimExpression}, we see that as $\alpha \nearrow 2$ for $\score_{\beta}<2\sqrt{1+\beta^2}$,
\begin{equation}
    \rho_{E}^{\frac{1-\alpha}{2}}\rho_{AE}^{\alpha}\rho_{E}^{\frac{1-\alpha}{2}}\rightarrow \frac{1}{4}\left( \ketbra{0}{0} +   \ketbra{1}{1}\right) \otimes \ketbra{v_1}{v_1} + \frac{1}{4}\left( \ketbra{0}{0} +   \ketbra{1}{1}\right) \otimes \ketbra{v_2}{v_2} \, .
\end{equation}
where $\{\ket{v_1},\ket{w_2} \}$ are normalized vectors. Thus,
\begin{align}
\lim_{\alpha \nearrow 2}\bar{H}^{\downarrow}_\alpha(A|E)_{\rho} &=  \lim_{\alpha \nearrow 2}\frac{1}{1-\alpha} \cdot \lim_{\alpha \nearrow 2}\log\left[ \operatorname{Tr} \left[\rho_{E}^{\frac{1-\alpha}{2}}\rho_{AE}^{\alpha}\rho_{E}^{\frac{1-\alpha}{2}} \right] \right] \\
 & = \lim_{\alpha \nearrow 2}\frac{1}{1-\alpha}\log 1 \\
 & = 0 \, .
\end{align}
Conversely, if we witness a maximal score, then we again simply need to consider the state from Eq.~\eqref{Eq: MaxScoreSigma}. For this state, we find that
\begin{equation}
     \left( \sigma_{E}^{i} \right)^{\frac{1-\alpha}{2}}\left( \sigma_{AE}^{i} \right)^{\alpha}\left( \sigma_{E}^{i} \right)^{\frac{1-\alpha}{2}}=\frac{1}{2^{\alpha}}\ketbra{0}{0}_{A} \otimes \ketbra{0}{0}_{E} + \frac{1}{2^{\alpha}}\ketbra{1}{1}_{A} \otimes \ketbra{0}{0}_{E} \, ,
\end{equation}
and
\begin{align}
\lim_{\alpha \nearrow 2}\bar{H}^{\downarrow}_\alpha(A|E)_{\sigma^i} &=  \lim_{\alpha \nearrow 2}\frac{1}{1-\alpha}\log\left[ \operatorname{Tr} \left[\left( \sigma_{E}^{i} \right)^{\frac{1-\alpha}{2}}\left( \sigma_{AE}^{i} \right)^{\alpha}\left( \sigma_{E}^{i} \right)^{\frac{1-\alpha}{2}} \right] \right] \\
 & = \lim_{\alpha \nearrow 2}\frac{1}{1-\alpha}\log 2^{1-\alpha} \\
 & = 1
\end{align}
It again follows that $\bar{H}^{\downarrow}_\infty(A\vert X=0, IE)_{\rho}=1$. 

From the calculations presented above, it is evident that the discontinuities of  $f_{\widetilde{H}^{\downarrow}_\infty}(\score_{\beta})$ and $f_{\bar{H}^{\downarrow}_2}(\score_{\beta})$ at ${\score_{\beta}=2\sqrt{1+\beta^2}}$ 
arise due to the edge-case behaviors of $\widetilde{H}^{\downarrow}_\infty(A\vert E)_{\rho}$ and $\bar{H}^{\downarrow}_2(A\vert E)_{\rho}$. Hence, our bounds are tight even at these discontinuities. Such behavior can also be found in the Belavkin-Staszewski conditional entropy which is discontinuous on states that are not full-rank~\cite{bluhm2023continuity}. One can then show that this implies it has the same rate function as $\widetilde{H}^{\downarrow}_\infty$ and $\bar{H}^{\downarrow}_2$.
Finally, we note that the functions $f_{\widetilde{H}^{\uparrow}_\alpha}(\score_{\beta})$ and $f_{\bar{H}^{\uparrow}_\alpha}(\score_{\beta})$ have no discontinuities for $\alpha >1$ (within their respective domains of validity, i.e.~$\alpha\leq2$ for the latter).

\subsection{Alternative Concavity Proofs} \label{app:Concavity2}

\begin{prop}\label{thm:f_pr_up_int}
    For all $x \in [0,1]$, the function
    \begin{align}\label{def:f_pr_up_int}
       f_{3}(x) = \left(1-\sqrt{x}\right)^{\frac{1}{\alpha}}+\left(1+\sqrt{x}\right)^{\frac{1}{\alpha}}
    \end{align}
    is concave for $\alpha \in (1,2)$.
\end{prop}
\begin{proof}
    Due to continuity arguments, it is sufficient to consider $x \in (0,1)$. To prove that Eq.~\eqref{def:f_pr_up_int} is concave, we need to calculate its second derivative. The second derivative of Eq.~\eqref{def:f_pr_up_int} is as follows:
	\begin{equation}\label{eqn:f_pr_up_sec_der}
		\frac{1}{4\alpha x^{\frac{3}{2}}} \left(-\left(1+\sqrt{x}\right)^{\frac{1}{\alpha}-2} \left[1+\left( 2-\frac{1}{\alpha}\right)\sqrt{x}\right]+\left(1-\sqrt{x}\right)^{\frac{1}{\alpha}-2}\left[1-\left(2-\frac{1}{\alpha}\right)\sqrt{x}\right]\right) \, .
	\end{equation}
	To show that Eq.~\eqref{def:f_pr_up_int} is concave, we need to show that Eq.~\eqref{eqn:f_pr_up_sec_der} is non-positive. The first term in the above expression is negative. The second term is negative when $x>x'$, where $x'$ is the root of ${d(x) = 1-\left(2-1/\alpha\right)\sqrt{x}}$ between $x=0$ and $x=1$. We can see that the $x'$ exists by observing that $d(x)$ switches signs and is a decreasing between $x=0$ and $x=1$. Let us consider the case when $\alpha>1$ and $x<x'$. For concavity of Eq.~\eqref{def:f_pr_up_int},
	\begin{equation}
		\frac{1}{4\alpha x^{\frac{3}{2}}} \left(-\left(1+\sqrt{x}\right)^{\frac{1}{\alpha}-2} \left[1+\left( 2-\frac{1}{\alpha}\right)\sqrt{x}\right]+\left(1-\sqrt{x}\right)^{\frac{1}{\alpha}-2}\left[1-\left(2-\frac{1}{\alpha}\right)\sqrt{x}\right]\right)\leq 0
	\end{equation}
	\begin{equation}
		\implies\left(1+\sqrt{x}\right)^{\frac{1}{\alpha}-2} \left[1+\left( 2-\frac{1}{\alpha}\right)\sqrt{x}\right]\geq \left(1-\sqrt{x}\right)^{\frac{1}{\alpha}-2}\left[1-\left(2-\frac{1}{\alpha}\right)\sqrt{x}\right]
	\end{equation}
	\begin{equation}
		\implies \left(\frac{1+\left(2-1/\alpha\right)\sqrt{x}}{1-\left(2-1/\alpha\right)\sqrt{x}}\right)\left(\frac{1-\sqrt{x}}{1+\sqrt{x}}\right)^{2-\frac{1}{\alpha}}\geq 1 \, .
	\end{equation}
	Note that as $x\rightarrow 0$, the left-hand-side tends to $1$. Also when $x\rightarrow x'$, the left-hand-side tends to $\infty$. The above inequality holds for $x<x'$ if, additionally, 
\begin{equation}\label{eqn:need_2b_monotone_1}
			\left(\frac{1+\left(2-1/\alpha\right)\sqrt{x}}{1-\left(2-1/\alpha\right)\sqrt{x}}\right)\left(\frac{1-\sqrt{x}}{1+\sqrt{x}}\right)^{2-\frac{1}{\alpha}}
		\end{equation}
		is a monotone function. To show this, let $t \coloneq \sqrt{x}$, where $t\in[0,1)$. Then 
		\begin{equation}
			\left(\frac{1+\left(2-1/\alpha\right)\sqrt{x}}{1-\left(2-1/\alpha\right)\sqrt{x}}\right)\left(\frac{1-\sqrt{x}}{1+\sqrt{x}}\right)^{2-\frac{1}{\alpha}}=\left(\frac{1+\left(2-1/\alpha\right)t}{1-\left(2-1/\alpha\right)t}\right)\left(\frac{1-t}{1+t}\right)^{2-\frac{1}{\alpha}} \, ,
		\end{equation}
		and the first derivative of Eq.~\eqref{eqn:need_2b_monotone_1} is as follows:
		\begin{align}
			&\frac{d}{dt}\left[\left(\frac{1+\left(2-1/\alpha\right)t}{1-\left(2-1/\alpha\right)t}\right)\left(\frac{1-t}{1+t}\right)^{2-\frac{1}{\alpha}}\right]\\
			&=\frac{[\left(2-1/\alpha\right)(1-t)^{2-\frac{1}{\alpha}}-(1+\left(2-1/\alpha\right)t)(2-1/\alpha)(1-t)^{1-\frac{1}{\alpha}}](1-\left(2-1/\alpha\right)t)(1+t)^{2-\frac{1}{\alpha}})}{(1-\left(2-1/\alpha\right)t)^2(1+t)^{4-\frac{2}{\alpha}}} \notag \\
			&\qquad+\frac{[\left(2-1/\alpha\right)(1+t)^{2-\frac{1}{\alpha}}-(1-\left(2-1/\alpha\right)t)(2-1/\alpha)(1+t)^{1-\frac{1}{\alpha}}](1+\left(2-1/\alpha\right)t)(1-t)^{2-\frac{1}{\alpha}})}{(1-\left(2-1/\alpha\right)t)^2(1+t)^{4-\frac{2}{\alpha}}}\\
            &=\frac{2\left(2-1/\alpha\right)(1-t^2)^{2-\frac{1}{\alpha}}-2(1-\left(2-1/\alpha\right)^2t^2)(2-1/\alpha)(1-t^2)^{1-\frac{1}{\alpha}}}{(1-\left(2-1/\alpha\right)t)^2(1+t)^{4-\frac{2}{\alpha}}}\\
            &=\left(4-\frac{2}{\alpha}\right)(1-t^2)^{1-\frac{1}{\alpha}}\frac{((2-1/\alpha)^2-1)t^2}{(1-\left(2-1/\alpha\right)t)^2(1+t)^{4-\frac{2}{\alpha}}}\geq 0 \, ,
		\end{align}
        when $t\in[0,\sqrt{x'})$. Hence, Eq.~\eqref{eqn:need_2b_monotone_1} is monotone when $x \in [0,x')$. Therefore, we have shown that Eq.~\eqref{eqn:f_pr_up_sec_der} is non-positive and thus Eq.~\eqref{def:f_pr_up_int} is concave.
\end{proof}

\begin{prop}
    For all $x \in [0,1]$, the function
    \begin{equation}\label{def:f_pr_down_int}
		f_{2} (x)=\left(1-\sqrt{x}\right)^{2-\alpha} + \left(1+\sqrt{x}\right)^{2-\alpha}
	\end{equation}
    is concave for $\alpha \in (1,2)$.
\end{prop}
\begin{proof}
    To prove that Eq.~\eqref{def:f_pr_down_int} is concave, we need to calculate its second derivative. Due to continuity arguments, it is sufficient to consider $x \in [0,1)$.
    The second derivative is given by
    \begin{align}\label{eqn:f_pr_down_sec_der}
        \frac{2-\alpha}{4x^{\frac{3}{2}}} \left[- \left(1+\alpha \sqrt{x} \right) \left( 1+\sqrt{x}\right)^{-\alpha}+\left(1-\alpha \sqrt{x} \right) \left( 1-\sqrt{x}\right)^{-\alpha}\right] \, .
    \end{align}   
	To show that Eq.~\eqref{def:f_pr_down_int} is concave, we need to show that Eq.~\eqref{eqn:f_pr_down_sec_der} is non-positive. The first term in the above expression is negative. The second term is negative when $x>x'$, where $x'$ is the root of $h(x) = 1-\alpha \sqrt{x} $ between $x=0$ and $x=1$. We can see that the $x'$ exists by observing that $h(x)$ switches signs and is a decreasing between $x=0$ and $x=1$. Let us consider the case when, $\alpha>1$ and $x<x'$
    \begin{align}
        &\frac{2-\alpha}{4x^{\frac{3}{2}}} \left[- \left(1+\alpha \sqrt{x} \right) \left( 1+\sqrt{x}\right)^{-\alpha}+\left(1-\alpha \sqrt{x} \right) \left( 1-\sqrt{x}\right)^{-\alpha}\right] \leq 0\\
        &\implies \left(1+\alpha \sqrt{x} \right) \left( 1+\sqrt{x}\right)^{-\alpha}\geq \left(1-\alpha \sqrt{x} \right) \left( 1-\sqrt{x}\right)^{-\alpha}\\
        &\implies \frac{\left(1+\alpha \sqrt{x} \right) \left(1-\sqrt{x}\right)^{\alpha}}{\left(1-\alpha \sqrt{x} \right) \left(1+\sqrt{x}\right)^{\alpha}}\geq 1 \, .
    \end{align}
	Note that as $x\rightarrow 0$, the left-hand-side tends to $1$. Also when $x\rightarrow x'$, the left-hand-side tends to $\infty$. The above inequality holds for $x<x'$ if, additionally, 
		\begin{equation}\label{eqn:need_2b_monotone_4}
			\frac{\left(1+\alpha \sqrt{x} \right) \left(1-\sqrt{x}\right)^{\alpha}}{\left(1-\alpha \sqrt{x} \right) \left(1+\sqrt{x}\right)^{\alpha}}
		\end{equation}
		is a monotone function. To show this, let $t \coloneq \sqrt{x}$, where $t\in[0,1)$. Then 
		\begin{equation}
			\frac{\left(1+\alpha \sqrt{x} \right) \left(1-\sqrt{x}\right)^{\alpha}}{\left(1-\alpha \sqrt{x} \right) \left(1+\sqrt{x}\right)^{\alpha}}=\frac{\left(1+\alpha t \right) \left(1-t\right)^{\alpha}}{\left(1-\alpha t \right) \left(1+t\right)^{\alpha}} \, ,
		\end{equation}
		and the first derivative of Eq.~\eqref{eqn:need_2b_monotone_4} is as follows:
		\begin{align}
			&\frac{d}{dt}\left[\frac{\left(1+\alpha t \right) \left(1-t\right)^{\alpha}}{\left(1-\alpha t \right) \left(1+t\right)^{\alpha}}\right]\\
			&=\frac{[\alpha(1-t)^{\alpha}-\alpha(1+\alpha t)(1-t)^{\alpha-1}]((1-\alpha t )(1+t)^{\alpha})}{(1-\alpha t )^2(1+t)^{2\alpha}}+\frac{[\alpha(1+t)^{\alpha}-\alpha(1-\alpha t )(1+t)^{\alpha-1}]((1+\alpha t )(1-t)^{\alpha})}{(1-\alpha t )^2(1+t)^{2\alpha}}\notag\\
			&=\frac{2\alpha(1-t^2)^{\alpha}+2\alpha(1-t^2)^{\alpha-1}(\alpha^2 t^2 -1)}{(1-\alpha t )^2(1+t)^{2\alpha}}= 2\alpha(1-t^2)^{\alpha-1}\frac{(\alpha^2-1)t^2}{(1-\alpha t )^2(1+t)^{2\alpha}}\geq 0 \, .
		\end{align}
    when $t\in[0,\sqrt{x'})$. Hence, Eq.~\eqref{eqn:need_2b_monotone_4} is monotone when $x \in [0,x')$. Therefore, we have shown that Eq.~\eqref{eqn:f_pr_down_sec_der} is non-positive and thus Eq.~\eqref{def:f_pr_down_int} is concave.
\end{proof}
\section{Incorporating Noisy Preprocessing} \label{app:noisypreprocessing}
Noisy preprocessing is a technique to boost key rates by injecting randomness into Alice's raw string during generation rounds~\cite{RGK05}. The intuition is that in certain settings this can decrease Eve's information about Alice's raw string more than it decreases Bob's information, overall leading to an increase in the key rate. It has already been shown for the CHSH-based DIQKD protocol that this can lower the minimal detection efficiencies~\cite{HSTRBS20}. To implement noisy preprocessing, Alice simply flips her output bit with some probability $q$ in generation rounds, using private local randomness. After incorporating such a step, the resulting classical-quantum state $\rho_{IAE}$ is no longer given by Eq.~\eqref{Eq: AllISigma}; rather it is of the form
\begin{align} \label{Eq: rhopreprocessing}
    \rho_{IAE} = \sum_{i} \operatorname{Pr} \left[ I=i\right] \ketbra{i}{i}_{I} \otimes \sum_{a} \left[ \ketbra{a}{a}_{A} \otimes \left( (1-q)\rho_E^{ia}  + q\rho_E^{i(a\oplus 1)} \right) \right] \; .
\end{align}
Following the same arguments as those used to prove Lemma~\ref{Lem: QProperties}, to lower bound $\mathbb{H}(A|I, E)_{\rho}$ it suffices to consider states of the form
\begin{align}
    \sigma_{IAE} &= \sum_{i} \operatorname{Pr} \left[ I=i\right] \ketbra{i}{i}_{I} \otimes \sigma_{AE}^{i} \label{Eq: AllISigmaNPP}\\
    \sigma_{AE}^{i} &= \frac{1}{2} \ketbra{0}{0} \otimes \left( (1-q)\ketbra{\psi_{=}}{\psi_{=}}  + q\ketbra{\psi_{\neq}}{\psi_{\neq}} \right) +  \frac{1}{2} \ketbra{1}{1} \otimes \left( q\ketbra{\psi_{=}}{\psi_{=}}  + (1-q)\ketbra{\psi_{\neq}}{\psi_{\neq}} \right)\, , \label{Eq: SingleISigmaNPP}
\end{align}
such that $ |\inner{\psi_{=}}{\psi_{\neq}}| \geq g_{\score}$. In the following theorem, we will use the notation that for $g_\score = \sqrt{\tfrac{\score^2_\beta }{4} - \beta^2}$, we define the following quantities:
\begin{align}
    N_{1}&=\sqrt{\left(g_{\score}-\sqrt{g^2_{\score}+\left(1-2q\right)^2(1-g^2_{\score})}\right)^2 + \left(1-2q\right)^2(1-g^2_{\score})} \\
N_{2}&=\sqrt{\left(g_{\score}+\sqrt{g^2_{\score}+\left(1-2q\right)^2(1-g^2_{\score})}\right)^2 + \left(1-2q\right)^2(1-g^2_{\score})} \, .
\end{align}
 \begin{thm} \label{Thm: NPPBounds}
		Let $\lvert \beta \rvert \geq 1$, $\score_\beta \in \left[2\lvert \beta \rvert, 2\sqrt{1+\beta^2}\right]$, and $g_\score = \sqrt{\tfrac{\score^2_\beta }{4} - \beta^2}$. Then, for any $q \in [0,1]$, we have (writing $\invbreve{h}$ to denote the concave envelope of an arbitrary function $h: \left[2\lvert \beta \rvert, 2\sqrt{1+\beta^2}\right] \to \mathbb{R}$):
		\begin{align}
		      f_{\widetilde{H}^{\downarrow}_\alpha}(\score_{\beta}) &= 1 + \frac{1}{1-\alpha}\log \left[ \invbreve{h}_{\widetilde{H}^{\downarrow}_\alpha}(\score_{\beta})\right]\label{def:down_arrow_sandwNPP} 
		\end{align}
        for all $\alpha \in  (1,\infty)$, where
\begin{multline}
 h_{\widetilde{H}^{\downarrow}_\alpha}(\score_{\beta}) =  \frac{1}{2^{\alpha+1}}\left[\left(\left( 1-g_\score\right)^{\frac{1}{\alpha}} + \left( 1+g_\score\right)^{\frac{1}{\alpha}} + \sqrt{\left(\left( 1-g_\score\right)^{\frac{1}{\alpha}} + \left( 1+g_\score\right)^{\frac{1}{\alpha}} \right)^2 - 16\left( 1-g_\score^2\right)^{\frac{1}{\alpha}} \left( q-q^2\right) } \right)^{\alpha} \right. \\
 \left. + \left(\left( 1-g_\score\right)^{\frac{1}{\alpha}} + \left( 1+g_\score\right)^{\frac{1}{\alpha}} - \sqrt{\left(\left( 1-g_\score\right)^{\frac{1}{\alpha}} + \left( 1+g_\score\right)^{\frac{1}{\alpha}} \right)^2 - 16\left( 1-g_\score^2\right)^{\frac{1}{\alpha}} \left( q-q^2\right) } \right)^{\alpha} \right]\, . \label{eq:down_arrow_sandwNPP} 
\end{multline}
        Similarly, 
        \begin{align}  
            f_{\Bar{H}^{\downarrow}_\alpha}(\score_{\beta}) &= 1 + \frac{1}{1-\alpha}\log \left[ \invbreve{h}_{\Bar{H}^{\downarrow}_\alpha}(\score_{\beta})\right] \label{def:down_arrow_petzNPP}\\
    	    f_{\Bar{H}^{\uparrow}_\alpha}(\score_{\beta}) &= 1 + \frac{\alpha}{1-\alpha}\log \left[ \invbreve{h}_{\Bar{H}^{\uparrow}_\alpha}(\score_{\beta})\right]\label{def:up_arrow_petzNPP} 
    	\end{align}
        for all $\alpha \in  (1,2]$, where
 \begin{multline}
 h_{\Bar{H}^{\downarrow}_\alpha}(\score_{\beta}) = \sum_{j=1}^{2}\frac{1}{2}\left(1+(-1)^{j+1}\sqrt{g^2_{\score}+(2q-1)^2(1-g^2_{\score})}\right)^{\alpha}\left(1-g_{\score}\right)^{1-\alpha}\left(\frac{g_{\score}+(-1)^{j}\sqrt{g^2_{\score}+(2q-1)^2(1-g^2_{\score})}}{N_{j}}\right)^2\\
 +  \sum_{j=1}^{2}\frac{1}{2}\left(1+(-1)^{j+1}\sqrt{g^2_{\score}+(2q-1)^2(1-g^2_{\score})}\right)^{\alpha}\left(1+g_{\score}\right)^{1-\alpha}\left(\frac{(2q-1)\sqrt{1-g^2_{\score}}}{N_{j}}\right)^2
\end{multline}
        and
        \begin{multline}
 h_{\Bar{H}^{\uparrow}_\alpha}(\score_{\beta}) =  \left[ \sum_{j=1}^{2}\left(\frac{1+\left(-1\right)^{j+1}\sqrt{g^2_{\score}+(2q-1)^2(1-g^2_{\score})}}{2}\right)^{\alpha}\left(1-\left(\frac{(2q-1)\sqrt{1-g^2_{\score}}}{N_{j}}\right)^2\right) \right]^{\frac{1}{\alpha}}\\
 + \left[ \sum_{j=1}^{2}\left(\frac{1+\left(-1\right)^{j+1}\sqrt{g^2_{\score}+(2q-1)^2(1-g^2_{\score})}}{2}\right)^{\alpha}\left(\frac{(2q-1)\sqrt{1-g^2_{\score}}}{N_{j}}\right)^2 \right]^{\frac{1}{\alpha}}\, .
\end{multline}
    \end{thm}
A proof of Theorem~\ref{Thm: NPPBounds} can be found in Appendix~\ref{Sec: NPPBound}. 
Note that the final bounds $f_{\mathbb{H}}$ are presented in terms of the concave envelopes $\invbreve{h}_{\mathbb{H}}$ of the functions $h_{\mathbb{H}}$ --- in order to avoid having to take this concave envelope,
it would suffice to show that all $h_{\mathbb{H}}$ are concave. 
We note that, up to numerical precision, these functions indeed appear to be concave.
We leave a rigorous proof of concavity for future work, and highlight that a consequence of this concavity would be that the equality $f_{\widetilde{H}^{\downarrow}_\alpha}=f_{\Bar{H}^{\uparrow}_\alpha}$ does not generally hold for every $q \in [0,1]$ and $\alpha \in  (1,2]$.

\subsection{Proof of Theorem~\ref{Thm: NPPBounds}} \label{Sec: NPPBound}
\begin{proof}
For any $\sigma_{IAE}$ of the form given by Eq.~\eqref{Eq: AllISigmaNPP}, we first consider Alice's and Eve's bipartite state for some $I=i$. For any $\sigma_{AE}$ as in Eq.~\eqref{Eq: SingleISigmaNPP}, Eve's reduced density matrix is given by 
\begin{align} \label{Eq: SigmaENPP}
    \sigma_{E}=\frac{1-g_\scalar}{2} \ketbra{v_1}{v_1} + \frac{1+g_\scalar}{2} \ketbra{v_2}{v_2} \; ,
\end{align}
where $g_\scalar = |\inner{\psi_{=}}{\psi_{\neq}}|$ and\footnote{Note that we work with $g_\scalar$, rather than $g_\score$, as we do not explicitly prove the fact that $h_{\mathbb{H}}(\score_{\beta})$ is monotonically decreasing.
Rather, at an appropriate point, we use the fact that its concave envelope is monotonically decreasing.}
      \begin{align}
        \ket{v_1} &= - \sqrt{\frac{1-g_\scalar}{2}} \ket{0} + \sqrt{\frac{1+g_\scalar}{2}} \ket{1} \\
        \ket{v_2} &=\phantom{-} \sqrt{\frac{1+g_\scalar}{2}} \ket{0} + \sqrt{\frac{1-g_\scalar}{2}} \ket{1} \, .
      \end{align}
We first derive an exact expression for $\widetilde{H}^{\downarrow}_\alpha(A|E)_{\sigma}$. Plugging the above expressions directly into $\sigma_{E}^{\frac{1-\alpha}{2\alpha}}\sigma_{AE}\sigma_{E}^{\frac{1-\alpha}{2\alpha}}$ gives us that
\begin{align}
    & \sigma_{E}^{\frac{1-\alpha}{2\alpha}}\sigma_{AE}\sigma_{E}^{\frac{1-\alpha}{2\alpha}} \\
    & = \frac{1}{2} \ketbra{0}{0}_{A} \otimes \left[ \left( \frac{1-g_\scalar}{2}\right)^{\frac{1}{\alpha}}\ketbra{v_1}{v_1}_{E} +  \left( \frac{1-g_\scalar^2}{4}\right)^{\frac{1}{2\alpha}}\left(q-\bar{q}\right)\left(\ketbra{v_1}{v_2}_{E}+\ketbra{v_2}{v_1}_{E} \right) + \left( \frac{1+g_\scalar}{2}\right)^{\frac{1}{\alpha}}\ketbra{v_2}{v_2}_{E}
    \right] \notag \\
    &+ \frac{1}{2} \ketbra{1}{1}_{A} \otimes \left[ \left( \frac{1-g_\scalar}{2}\right)^{\frac{1}{\alpha}}\ketbra{v_1}{v_1}_{E} +  \left( \frac{1-g_\scalar^2}{4}\right)^{\frac{1}{2\alpha}}\left(\bar{q}-q\right)\left(\ketbra{v_1}{v_2}_{E}+\ketbra{v_2}{v_1}_{E} \right) + \left( \frac{1+g_\scalar}{2}\right)^{\frac{1}{\alpha}}\ketbra{v_2}{v_2}_{E}
    \right] \; ,
\end{align}
where $\bar{q}=1-q$. This expression is diagonalizable, and it is easily verifiable that
\begin{align}
    \sigma_{E}^{\frac{1-\alpha}{2\alpha}}\sigma_{AE}\sigma_{E}^{\frac{1-\alpha}{2\alpha}}= 2^{-1-\frac{1}{\alpha}} \left(\ketbra{0}{0}_{A} \otimes \left[ \lambda_{1}\ketbra{w_1}{w_1}_{E} + \lambda_{2}\ketbra{w_2}{w_2}_{E}
    \right]   + \ketbra{1}{1}_{A} \otimes \left[ \lambda_{1}^\prime\ketbra{w_1^\prime}{w_1^\prime}_{E} + \lambda_{2}^\prime\ketbra{w_2^\prime}{w_2^\prime}_{E}
    \right]\right) \; ,
\end{align}
where $\{\ket{w_1},\ket{w_2} \} \ \text{and}\  \{ \ket{w_1^\prime}, \ket{w_2^\prime} \}$ are pairwise orthonormal vectors, and 
\begin{align}
    \lambda_{1,2} &= \frac{\left( 1-g_\scalar\right)^{\frac{1}{\alpha}} + \left( 1+g_\scalar\right)^{\frac{1}{\alpha}} }{2} \pm \sqrt{\left(\frac{\left( 1-g_\scalar\right)^{\frac{1}{\alpha}} + \left( 1+g_\scalar\right)^{\frac{1}{\alpha}} }{2}\right)^2 - \left( 1-g_\scalar^2\right)^{\frac{1}{\alpha}} \left( 4q-4q^2\right) }  \\
\lambda_{1,2}^{\prime} &= \frac{\left( 1-g_\scalar\right)^{\frac{1}{\alpha}} + \left( 1+g_\scalar\right)^{\frac{1}{\alpha}} }{2} \pm \sqrt{\left(\frac{\left( 1-g_\scalar\right)^{\frac{1}{\alpha}} + \left( 1+g_\scalar\right)^{\frac{1}{\alpha}} }{2}\right)^2 - \left( 1-g_\scalar^2\right)^{\frac{1}{\alpha}} \left( 4\bar{q}-4\bar{q}^2\right) } \, . 
\end{align}
Note that because $q-q^2 = \bar{q} - \bar{q}^2$, it must also hold that $\lambda_{1,2}= \lambda_{1,2}^\prime$. It then directly follows that
\begin{align}
\widetilde{H}^{\downarrow}_\alpha(A|E)_{\sigma} &=  \frac{1}{1-\alpha}\log\left[ \operatorname{Tr} \left[\left(\sigma_{E}^{\frac{1-\alpha}{2\alpha}}\sigma_{AE}\sigma_{E}^{\frac{1-\alpha}{2\alpha}}\right)^\alpha \right] \right] \\
 & = \frac{1}{1-\alpha}\log \frac{\left[ \lambda_{1}^{\alpha} + \lambda_{2}^{\alpha} +  \lambda_{1}^{\prime \alpha} + \lambda_{2}^{\prime \alpha}\right] }{2^{\alpha+1}} \\
 & = \frac{1}{1-\alpha}\log \frac{\left[ 2\lambda_{1}^{\alpha} + 2\lambda_{2}^{\alpha} \right] }{2^{\alpha+1}} \\
   & = \frac{1}{1-\alpha}\log \frac{\left[ \lambda_{1}^{\alpha} + \lambda_{2}^{\alpha} \right] }{2^{\alpha}} \, .
\end{align}
This thus yields
\begin{multline}
 \widetilde{H}^{\downarrow}_\alpha(A|E)_{\sigma} =  1 + \frac{1}{1-\alpha}\log \left[\frac{1}{2^{\alpha+1}}\left[\left(\left( 1-g_\scalar\right)^{\frac{1}{\alpha}} + \left( 1+g_\scalar\right)^{\frac{1}{\alpha}} + \sqrt{\left(\left( 1-g_\scalar\right)^{\frac{1}{\alpha}} + \left( 1+g_\scalar\right)^{\frac{1}{\alpha}} \right)^2 - 16\left( 1-g_\scalar^2\right)^{\frac{1}{\alpha}} \left( q-q^2\right) } \right)^{\alpha} \right. \right.\\
 \left. \left.+ \left(\left( 1-g_\scalar\right)^{\frac{1}{\alpha}} + \left( 1+g_\scalar\right)^{\frac{1}{\alpha}} - \sqrt{\left(\left( 1-g_\scalar\right)^{\frac{1}{\alpha}} + \left( 1+g_\scalar\right)^{\frac{1}{\alpha}} \right)^2 - 16\left( 1-g_\scalar^2\right)^{\frac{1}{\alpha}} \left( q-q^2\right) } \right)^{\alpha} \right]\right]\, .
\end{multline}
Up to the use of the concave envelope and the difference between $g_{\scalar}$ and $g_{\score}$, this is the expression from Eq.~\eqref{def:down_arrow_sandwNPP}. Next, we consider $\Bar{H}^{\downarrow}_\alpha(A|E)_{\sigma}$. For any $\sigma_{AE}$ as in Eq.~\eqref{Eq: SingleISigmaNPP}, it holds that 
\begin{multline}
    \sigma_{AE}^{\alpha} =  \ketbra{0}{0} \otimes\left[\left(\frac{1+\sqrt{g^2_{\scalar}+(q-\bar{q})^2(1-g^2_{\scalar})}}{4}\right)^{\alpha}\ketbra{u_1}{u_1} + \left(\frac{1-\sqrt{g^2_{\scalar}+(q-\bar{q})^2(1-g^2_{\scalar})}}{4}\right)^{\alpha}\ketbra{u_2}{u_2}\right]\\
        +\ketbra{1}{1} \otimes\left[\left(\frac{1+\sqrt{g^2_{\scalar}+(q-\bar{q})^2(1-g^2_{\scalar})}}{4}\right)^{\alpha}\ketbra{u_1^\prime}{u_1^\prime} + \left(\frac{1-\sqrt{g^2_{\scalar}+(q-\bar{q})^2(1-g^2_{\scalar})}}{4}\right)^{\alpha}\ketbra{u_2^\prime}{u_2^\prime}\right] \; ,
\end{multline}
where
\begin{align}
    \ket{u_1} &= \frac{g_{\scalar}-\sqrt{g^2_{\scalar}+(q-\bar{q})^2(1-g^2_{\scalar})}}{N_{1}}\ket{v_1} -\frac{(q-\bar{q})\sqrt{1-g^2_{\scalar}}}{N_{1}}\ket{v_2}\\
   \ket{u_2} &= \frac{g_{\scalar}+\sqrt{g^2_{\scalar}+(q-\bar{q})^2(1-g^2_{\scalar})}}{N_{2}}\ket{v_1} -\frac{(q-\bar{q})\sqrt{1-g^2_{\scalar}}}{N_{2}}\ket{v_2}\\
    \ket{u_1^\prime} &= \frac{g_{\scalar}-\sqrt{g^2_{\scalar}+(q-\bar{q})^2(1-g^2_{\scalar})}}{N_{1}}\ket{v_1} +\frac{(q-\bar{q})\sqrt{1-g^2_{\scalar}}}{N_{1}}\ket{v_2}\\
    \ket{u_2^\prime} &= \frac{g_{\scalar}+\sqrt{g^2_{\scalar}+(q-\bar{q})^2(1-g^2_{\scalar})}}{N_{2}}\ket{v_1} +\frac{(q-\bar{q})\sqrt{1-g^2_{\scalar}}}{N_{2}}\ket{v_2} \; .
\end{align}
It then holds that
\begin{align}
    \Bar{H}^{\downarrow}_\alpha(A|E)_{\sigma} &=  \frac{1}{1-\alpha}\log\left[ \operatorname{Tr} \left[  \sigma_{AE}^{\alpha} \sigma_{E}^{1-\alpha}\right] \right] \\
    &=  \frac{1}{1-\alpha}\log \left[\sum_{j=1}^{2}2\left(\frac{1+(-1)^{j+1}\sqrt{g^2_{\scalar}+(q-\bar{q})^2(1-g^2_{\scalar})}}{4}\right)^{\alpha}\left(\frac{1-g_{\scalar}}{2}\right)^{1-\alpha}\left(\frac{g_{\scalar}+(-1)^{j}\sqrt{g^2_{\scalar}+(q-\bar{q})^2(1-g^2_{\scalar})}}{N_{j}}\right)^2
    \right. \notag
    \\
    &+ \left. \sum_{j=1}^{2}2\left(\frac{1+(-1)^{j+1}\sqrt{g^2_{\scalar}+(q-\bar{q})^2(1-g^2_{\scalar})}}{4}\right)^{\alpha}\left(\frac{1+g_{\scalar}}{2}\right)^{1-\alpha}\left(\frac{(q-\bar{q})\sqrt{1-g^2_{\scalar}}}{N_{j}}\right)^2
    \right] \, .
\end{align}
Slightly rewriting this equality, we then attain that this is equal to
 \begin{multline}
 1 + \frac{1}{1-\alpha}\log \left[\sum_{j=1}^{2}\frac{1}{2}\left(1+(-1)^{j+1}\sqrt{g^2_{\scalar}+(2q-1)^2(1-g^2_{\scalar})}\right)^{\alpha}\left(1-g_{\scalar}\right)^{1-\alpha}\left(\frac{g_{\scalar}+(-1)^{j}\sqrt{g^2_{\scalar}+(2q-1)^2(1-g^2_{\scalar})}}{N_{j}}\right)^2 \right.\\
 + \left. \sum_{j=1}^{2}\frac{1}{2}\left(1+(-1)^{j+1}\sqrt{g^2_{\scalar}+(2q-1)^2(1-g^2_{\scalar})}\right)^{\alpha}\left(1+g_{\scalar}\right)^{1-\alpha}\left(\frac{(2q-1)\sqrt{1-g^2_{\scalar}}}{N_{j}}\right)^2 \right]\, .
\end{multline}
To calculate $ \Bar{H}^{\uparrow}_{\alpha}\left(A|E\right)_\sigma$, we again use~\cite[Lemma~5.1]{Tomamichel2015QuantumIP}, i.e.
 \begin{align}
\Bar{H}^{\uparrow}_{\alpha}\left(A|E\right)_\sigma = \frac{\alpha}{1-\alpha} \log \left[\operatorname{Tr}\left[ \operatorname{Tr}_A\left(\sigma^\alpha_{AE}\right)^{\frac{1}{\alpha}}\right] \right] \, .
    \end{align}
A consequence of the above calculations is that
 \begin{align}  \operatorname{Tr}_A[\sigma^{\alpha}_{AE}]
        &=\sum_{j=1}^{2} \left(\frac{1+\left(-1\right)^{j+1}\sqrt{g^2_{\scalar}+(q-\bar{q})^2(1-g^2_{\scalar})}}{4}\right)^{\alpha}\left(\ketbra{u_j}{u_j}+ \ketbra{u_j^\prime}{u_j^\prime}\right)
        \\
        &=\sum_{j=1}^{2} 2\left(\frac{1+\left(-1\right)^{j+1}\sqrt{g^2_{\scalar}+(q-\bar{q})^2(1-g^2_{\scalar})}}{4}\right)^{\alpha}\notag \\
    &\ \ \ \cdot \left[\left(\frac{g_{\scalar}+\left(-1\right)^{j}\sqrt{g^2_{\scalar}+(q-\bar{q})^2(1-g^2_{\scalar})}}{N_{j}}\right)^2\ketbra{v_1}{v_1} + \left(\frac{(q-\bar{q})\sqrt{1-g^2_{\scalar}}}{N_{j}}\right)^2\ketbra{v_2}{v_2}\right]\\
&=\sum_{j=1}^{2} 2\left(\frac{1+\left(-1\right)^{j+1}\sqrt{g^2_{\scalar}+(q-\bar{q})^2(1-g^2_{\scalar})}}{4}\right)^{\alpha}\notag \\
    &\ \ \ \cdot \left[\left(1-\left(\frac{(q-\bar{q})\sqrt{1-g^2_{\scalar}}}{N_{j}}\right)^2\right)\ketbra{v_1}{v_1} + 
    \left(\frac{(q-\bar{q})\sqrt{1-g^2_{\scalar}}}{N_{j}}\right)^2\ketbra{v_2}{v_2}\right] \, .
    \end{align}
    The eigenvalues of $\operatorname{Tr}_A[\sigma^{\alpha}_{AE}]$ are given by
    \begin{align}
        \mu_1 =& 2 \sum_{j=1}^{2}\left(\frac{1+\left(-1\right)^{j+1}\sqrt{g^2_{\scalar}+(q-\bar{q})^2(1-g^2_{\scalar})}}{4}\right)^{\alpha}\left(1-\left(\frac{(q-\bar{q})\sqrt{1-g^2_{\scalar}}}{N_{j}}\right)^2\right)\\
        \mu_2 =& 2 \sum_{j=1}^{2}\left(\frac{1+\left(-1\right)^{j+1}\sqrt{g^2_{\scalar}+(q-\bar{q})^2(1-g^2_{\scalar})}}{4}\right)^{\alpha}\left(\frac{(q-\bar{q})\sqrt{1-g^2_{\scalar}}}{N_{j}}\right)^2 \, ,
    \end{align}
and therefore
\begin{multline}
\Bar{H}^{\uparrow}_{\alpha}\left(A|E\right)_\sigma =  1 + \frac{\alpha}{1-\alpha}\log \left[\left[ \sum_{j=1}^{2}\left(\frac{1+\left(-1\right)^{j+1}\sqrt{g^2_{\scalar}+(2q-1)^2(1-g^2_{\scalar})}}{2}\right)^{\alpha}\left(1-\left(\frac{(2q-1)\sqrt{1-g^2_{\scalar}}}{N_{j}}\right)^2\right) \right]^{\frac{1}{\alpha}} \right.\\
 + \left. \left[ \sum_{j=1}^{2}\left(\frac{1+\left(-1\right)^{j+1}\sqrt{g^2_{\scalar}+(2q-1)^2(1-g^2_{\scalar})}}{2}\right)^{\alpha}\left(\frac{(2q-1)\sqrt{1-g^2_{\scalar}}}{N_{j}}\right)^2 \right]^{\frac{1}{\alpha}} \right] \, .
\end{multline}
Again, up to the use of the concave envelope and the difference between $g_{\scalar}$ and $g_{\score}$, these expressions for the Petz-{\Renyi} entropies are equal to Eqs.~${\eqref{def:down_arrow_petzNPP}-\eqref{def:up_arrow_petzNPP}}$. Let us now reintroduce the index $I=i$. We prove in Appendix~\ref{Sec: MonotonicityPropNPP} that all three functions $\invbreve{h}_{\mathbb{H}}$ are monotonically decreasing. 
By upper bounding $h_{\mathbb{H}}$  with $\invbreve{h}_{\mathbb{H}}$ and then using the fact that, due to the monotonicity of $\invbreve{h}_{\mathbb{H}}$, one can provide a subsequent upper  bound by replacing $g_{\scalar}$ with $g_{\score}$, the following must be true. For any $\rho_{IAE}$ of the form given by Eq.~\eqref{Eq: rhopreprocessing}, it must hold that
\begin{align}    \widetilde{H}^{\downarrow}_\alpha(A|X=0,IE)_{\rho} &\geq \frac{1}{1-\alpha}\log \left[\sum_{i} \operatorname{Pr} \left[ I=i\right] 2^{\left(1-\alpha\right)\widetilde{H}^{\downarrow}_\alpha(A|I=i,E)_{\sigma^i}} \right] \\
&\geq 1+ \frac{1}{1-\alpha}\log \left[ \sum_{i} \operatorname{Pr} \left[ I=i\right] \invbreve{h}_{\widetilde{H}^{\downarrow}_\alpha}(\score^i_{\beta}) \right]\\
&\geq 1 + \frac{1}{1-\alpha}\log \left[ \invbreve{h}_{\widetilde{H}^{\downarrow}_\alpha}(\score_{\beta})\right] \, ,
\end{align}
where we use~\cite[Prop. 5.1]{Tomamichel2015QuantumIP} in the first inequality. To prove that this bound is tight, first note that, by construction, for any $\score_{\beta} \in \left[2\lvert \beta \rvert, 2\sqrt{1+\beta^2}\right]$, there must exist a set of $\{ \score_{\beta}^{i} \}_{i}$ and a distribution defined by the elements $\{p_i \}_{i}$ such that
\begin{align}
    \sum_{i} p_i \score_{\beta}^{i} &= \score_{\beta} \\
    \sum_{i} p_i h_{\widetilde{H}_{\alpha}^{\downarrow}} \left(\score_{\beta}^{i}\right) &= \invbreve{h}_{\widetilde{H}^{\downarrow}_\alpha}\left(\score_{\beta}\right) \, .
\end{align}
For any such $\score_{\beta}^{i}$, the explicit attack from Appendix~\ref{Sec: TightnessBounds} produces a post-measurement state, $\sigma^{i}_{AE}$, of the form as in Eq.~\eqref{Eq: SingleISigmaNPP}. However, for such states, we have explicitly calculated the {\Renyi} entropy, and
        \begin{align}
            \widetilde{H}_{\alpha}^{\downarrow}(A|X=0,E)_{\sigma^i} &= 1 + \frac{1}{1-\alpha} \log\left(h_{\widetilde{H}_{\alpha}^{\downarrow}}\left(\score_{\beta}^{i}\right)\right) \, .
        \end{align}
If, with probability $p_i$, Eve constructs the attack from Appendix~\ref{Sec: TightnessBounds} that achieves a score of $\score_{\beta}^{i}$ and stores the index $I=i$ on a classical register, then she can generate the post-measurement state $\sigma_{IAE} = \sum_{i} p_i \ketbra{i}{i}_{I} \otimes \sigma^{i}_{AE}$. For this state, 
\begin{align}
    \widetilde{H}_{\alpha}^{\downarrow}(A|X=0,IE)_{\sigma} &= 1 + \frac{1}{1-\alpha} \log\left(\sum_{i} p_i h_{\widetilde{H}_{\alpha}^{\downarrow}} (\score_{\beta}^{i})\right) \\
    &= 1 + \frac{1}{1-\alpha} \log\left(\invbreve{h}_{\widetilde{H}^{\downarrow}_\alpha}\left(\score_{\beta}\right)\right) \, ,
\end{align}
 due to~\cite[Prop. 5.1]{Tomamichel2015QuantumIP} (or else see the classical linearity property in Lemma~\ref{Lem: QProperties}). 
This attack thus saturates our bounds.
Using the same arguments, it must also hold that
\begin{align}  
            \Bar{H}^{\downarrow}_\alpha(A|X=0,IE)_{\rho} &= 1 + \frac{1}{1-\alpha}\log \left[ \invbreve{h}_{\Bar{H}^{\downarrow}_\alpha}(\score_{\beta})\right] \\
         \Bar{H}^{\uparrow}_\alpha(A|X=0,IE)_{\rho} &= 1 + \frac{\alpha}{1-\alpha}\log \left[ \invbreve{h}_{\Bar{H}^{\uparrow}_\alpha}(\score_{\beta})\right] \, .
    	\end{align}
\end{proof}
\subsection{Monotonicity Properties} \label{Sec: MonotonicityPropNPP}
In this section, we show that all three functions $\invbreve{h}_{\mathbb{H}}(\score_{\beta})$ described in Theorem~\ref{Thm: NPPBounds} are monotonically decreasing in the score, $\score_{\beta}$ (for $\alpha$, $\beta$ in the appropriate ranges).
\begin{prop}
 For all $\alpha \in (1,\infty)$ and $q \in \left[0,1\right]$, the function $\invbreve{h}_{\widetilde{H}^{\downarrow}_\alpha}(\score_{\beta})$ is monotonically decreasing in the interval $\score_\beta \in \left[2\lvert \beta \rvert, 2\sqrt{1+\beta^2}\right]$.
 Similarly, for all $\alpha \in (1,2)$ and $q \in \left[0,1\right]$, the functions $\invbreve{h}_{\Bar{H}^{\downarrow}_\alpha}(\score_{\beta})$ and $\invbreve{h}_{\Bar{H}^{\uparrow}_\alpha}(\score_{\beta})$ are monotonically decreasing in the interval $\score_\beta \in \left[2\lvert \beta \rvert, 2\sqrt{1+\beta^2}\right]$.
\end{prop}
\begin{proof}
    Intuitively, the monotonicity will be a consequence of the fact that $\invbreve{h}_{\mathbb{H}}(\score_{\beta})$ is concave and non-increasing at the point $\score_{\beta} =2\lvert \beta \rvert$. More concretely, for any $h_{\mathbb{H}}(\score_{\beta})$,  we are evaluating (up to a multiplicative constant) the corresponding quantity $\mathbb{Q}(A|E)$ for a classical-quantum state of the form in Eq.~\eqref{Eq: SingleISigmaNPP} such that $ |\inner{\psi_{=}}{\psi_{\neq}}| = g_{\score}$.  For these classical-quantum states, we shall show that 
    \begin{align}   h_{\mathbb{H}}(\score_{\beta}) &\leq h_{\mathbb{H}}(2\lvert \beta \rvert) \, .
    \end{align}
    In other words, at $g_{\score}=0$ (i.e. $\score_{\beta} = 2\lvert \beta \rvert$), Eve only needs to guess the noisy preprocessing bit; for all other $g_{\score}>0$, the conditional entropy, $\mathbb{H}(A| E)$, cannot be lower than that value. To prove this, let $A^\prime$ contain Alice's output bit before applying noisy preprocessing, let $Q$ contain the noisy preprocessing bit that is added to $A^\prime$, and let $A$ contain the bit Alice stores after incorporating noisy preprocessing. It must hold that 
    \begin{align}
      \mathbb{H}(Q) =  \mathbb{H}(Q|A^\prime E) = \mathbb{H}(AQ|A^\prime E) = \mathbb{H}(A|A^\prime E) \leq \mathbb{H}(A| E) \, .
    \end{align}
The first inequality holds, as $Q$ is sampled independently and not related to the registers $A^\prime$ and $E$. The second equality holds, because $A$ can be constructed deterministically from $A^\prime$ and $Q$, in the sense of~\cite[Lemma~B.7]{DFR20}. The third  equality holds for the same reason, i.e. $Q$ can be constructed deterministically from $A^\prime$ and $A$. The last inequality holds due to data processing~\cite[Corollary~5.1]{Tomamichel2015QuantumIP}. 

For any $\mathbb{H}$, the constant function $l(\score_{\beta})=h_{\mathbb{H}}(2\lvert \beta \rvert) $
is an upper bound on $h_{\mathbb{H}}(\score_{\beta})$ that is (trivially) concave, and tight at the endpoint $\score_{\beta}=2\lvert \beta \rvert$; this implies that the concave envelope of $h_{\mathbb{H}}(\score_{\beta})$ also satisfies  
    \begin{align}   \invbreve{h}_{\mathbb{H}}(2\lvert \beta \rvert) &= h_{\mathbb{H}}(2\lvert \beta \rvert) \, .
    \end{align}
    We are now ready to prove monotonicity. For any $a,b \in \left[2\lvert \beta \rvert, 2\sqrt{1+\beta^2}\right]$ such that $a \leq b$, it must hold due to concavity that 
     \begin{align}   \invbreve{h}_{\mathbb{H}}(a) &\geq \frac{b-a}{b-2\lvert \beta \rvert}\invbreve{h}_{\mathbb{H}}(2\lvert \beta \rvert) + \frac{a-2\lvert \beta \rvert}{b-2\lvert \beta \rvert}\invbreve{h}_{\mathbb{H}}(b) = \frac{b-a}{b-2\lvert \beta \rvert}l(b) + \frac{a-2\lvert \beta \rvert}{b-2\lvert \beta \rvert}\invbreve{h}_{\mathbb{H}}(b) \, .
    \end{align}
    Moreover, as $l(\score_{\beta})$ is concave, it must also be an upper bound on $\invbreve{h}_{\mathbb{H}}(\score_{\beta})$. It then follows that
    \begin{align}   \invbreve{h}_{\mathbb{H}}(a) & \geq \frac{b-a}{b-2\lvert \beta \rvert}\invbreve{h}_{\mathbb{H}}(b) + \frac{a-2\lvert \beta \rvert}{b-2\lvert \beta \rvert}\invbreve{h}_{\mathbb{H}}(b) = \invbreve{h}_{\mathbb{H}}(b) \, .
    \end{align}
    This proves that all three functions are monotonically decreasing.
\end{proof}
\section{Finite-size key rates}
\label{app:finitesize}

\newcommand{\HEV}{H_{\mathrm{EV}}}
\newcommand{\ecom}{\eps_{\mathrm{com}}}
\newcommand{\ecomAT}{\eps_{\mathrm{com}}^{\mathrm{AT}}}
\newcommand{\ecomEV}{\eps_{\mathrm{com}}^{\mathrm{EV}}}
\newcommand{\esound}{\eps_{\mathrm{sound}}}
\newcommand{\ecorr}{\eps_{\mathrm{corr}}}
\newcommand{\esecret}{\eps_{\mathrm{secret}}}
\newcommand{\ECstring}{\mathbf{L}_\mathrm{EC}}
\newcommand{\EVstring}{\mathbf{L}_\mathrm{EV}}
\newcommand{\lenEC}{\ell_{\mathrm{EC}}}
\newcommand{\lenEV}{\ell_{\mathrm{EV}}}
\newcommand{\lkey}{\ell_{\mathrm{key}}}
\newcommand{\dupp}{\delta^\mathrm{upp}}
\newcommand{\dlow}{\delta^\mathrm{low}}
\newcommand{\wCHSHhon}{\omega_\mathrm{hon}}
\newcommand{\qberhon}{Q^\mathrm{err}_\mathrm{hon}}

We first elaborate on the classical postprocessing used in the last steps. For a DIQKD protocol, this consists of the following procedures:
\begin{enumerate}
\item Error correction and error verification:
% \footnote{In~\cite{NDN+22}, these two procedures were described as a single error correction step. To stay in line with more current terminology, we choose to use separate terms for the two procedures, though this is merely a change of terminology with no change in the actual physical procedures implemented.} 
In error correction, Alice sends a string $\ECstring$ (of some fixed length $\lenEC$) to Bob, who uses it to produce a guess for Alice's string $A_1^n$. Then in error verification, Alice draws some choice of hash function $\HEV$ from a $\delta$-almost-universal hash family~\cite{DK07} (with fixed output length $\lenEV$), then applies it to $A_1^n$ and sends the resulting value $\EVstring$ to Bob, along with the choice of hash function $\HEV$.
%\footnote{Many QKD protocols are described using the more stringent requirement of 2-universal hashing in this step; however, in this work we accommodate the broader possibility of $\delta$-almost-universal hashing, as this was used in~\cite{NDN+22} and still allows us to prove the desired properties.} 
Bob then computes the corresponding hash of his guess and aborts if it does not match.
\item Privacy amplification: Alice applies a privacy amplification procedure to $A_1^n$ to produce a final key of length $\lkey$, and Bob does the same to his guess for $A_1^n$.
\end{enumerate}

When designing a protocol for the finite-size regime, there are two critical ``overall'' parameters that should be considered. We briefly outline them here, deferring the details to e.g.~\cite{NDN+22,PR22}.

\begin{enumerate}
\item The \term{completeness} parameter $\ecom$: this is an upper bound on the probability that the honest behavior aborts. Since this protocol might abort during either the acceptance test or the error verification step, it is convenient to construct upper bounds $\ecomAT$ and $\ecomEV$ on the abort probabilities in each of those two steps respectively, after which one can validly take $\ecom = \ecomAT+\ecomEV$ due to the union bound. Here we design the protocol such that $\ecomAT = 10^{-3}$ and $\ecomEV \leq 10^{-3}$ in an essentially similar fashion to~\cite{NDN+22}\footnote{While the final completeness parameter reported in that work was $\ecom=10^{-2}$, that was a somewhat conservative estimate.} (we elaborate on this in Sec.~\ref{subsec:completeness} below). We emphasize that this parameter does not affect the security properties of the protocol in any way, which are instead quantified by the next parameter.
\item The \term{soundness} parameter $\esound$: informally, this quantifies the ``security'' of the final key; refer to~\cite{NDN+22,PR22} for a rigorous definition. As discussed in those works, to analyze this parameter it suffices to separately consider a \term{correctness} parameter $\ecorr$ and a \term{secrecy} parameter $\esecret$, then set $\esound = \ecorr + \esecret$. Basically, $\ecorr$ is an upper bound on the probability that the final keys do not match and the protocol accepts, while roughly speaking $\esecret$ quantifies how well Alice's final key is decoupled from Eve; again, see~\cite{NDN+22,PR22} for details. 
Following~\cite{NDN+22}, we design the protocol such that $\esound = 10^{-10}$, by setting $\ecorr = 2^{-61}$ and $\esecret = \esound - \ecorr$. 
\end{enumerate}
We now discuss the details of our protocol in terms of the above parameters. 
For the testing probability $\gamma$, in Fig.~\ref{subfig:smalln} we followed the value $\gamma=13/256$ used in~\cite{NDN+22} for all data points, whereas in Fig.~\ref{subfig:largen} we optimized over $\gamma$ in units of $1/256$ (as was done in~\cite{NDN+22} so that the test/generation decision could be straightforwardly determined by drawing $8$ uniformly random bits). We did so because we found that for the range of $n$ values in the former, optimizing the choice of $\gamma$ only improved the finite-size key rates by less than $0.002$. In contrast, for the larger $n$ values considered in the latter, we found that it was important to optimize over the choice of $\gamma$ to obtain better finite-size key rates --- this is due to some subtle limitations we discuss in Appendix~\ref{app:protmods} later.

We emphasize that apart from the above point regarding $\gamma$, our protocol only differs from the protocol in~\cite{NDN+22} in terms of using a slightly different accept condition (see Remark~\ref{remark:AT}), a technical point in privacy amplification (see Remark~\ref{remark:PA}), and having Bob directly announce the values $\bar{B}_1^n$ for Alice to compute $\bar{C}_1^n$ (which slightly simplifies the analysis without sacrificing key rate; see Remark~\ref{remark:diffs}).

\subsection{Completeness}
\label{subsec:completeness}

To discuss completeness, we need to specify some honest behavior for the devices. We suppose that the honest behavior is IID, and each round produces some distribution $\vect{q}_\mathrm{hon}$ on the register $\bar{C}_j$ for that round. For our protocol, this distribution would have the form
\begin{align}
q_\mathrm{hon}(0) = \gamma (1-\wCHSHhon), \quad q_\mathrm{hon}(1) = \gamma \wCHSHhon, \quad q_\mathrm{hon}(\perp) = 1-\gamma,
\end{align}
where $\wCHSHhon$ is the expected CHSH winning probability of the honest behavior in test rounds.
Furthermore, let $\qberhon$ denote the probability of Alice and Bob getting different outcomes in generation rounds.
Following~\cite{NDN+22}, we set 
\begin{align}
\wCHSHhon = 0.83, \quad \qberhon = 0.018,
\end{align}
where the $\wCHSHhon$ value corresponds to the expected CHSH ``correlator'' score of $S = 2.64$ used in that work (as a somewhat conservative estimate of the device performance in that experiment).

For a given $\ecomAT$, we need to choose the set $S_\mathrm{acc}$ in the accept condition such that the probability of the honest behavior yielding a frequency distribution outside $S_\mathrm{acc}$ is at most $\ecomAT$. 
We shall focus on $S_\mathrm{acc}$ of the following form: for each value $\bar{c} \in \{0,1,\perp\}$ we take some values $\dlow_{\bar{c}},\dupp_{\bar{c}}>0$, and set $S_\mathrm{acc}$ to be the set of distributions $\vect{q}$ satisfying
\begin{align}\label{eq:acceptbox}
\forall \bar{c} \in \{0,1,\perp\}, \quad 
q_\mathrm{hon}(\bar{c}) - \dlow_{\bar{c}} \leq q(\bar{c}) \leq q_\mathrm{hon}(\bar{c}) + \dupp_{\bar{c}}
.
\end{align}
For $S_\mathrm{acc}$ of this form, to achieve some desired $\ecomAT$, it suffices to choose the values $\dlow_{\bar{c}},\dupp_{\bar{c}}$ such that for the honest behavior we have
\begin{align}\label{eq:probscom}
\forall \bar{c} \in \{0,1,\perp\}, \quad \Pr[\freq_{\bar{C}_1^n}(\bar{c}) < q_\mathrm{hon}(\bar{c}) - \dlow_{\bar{c}}] \leq \frac{\ecomAT}{6} 
\quad\text{and}\quad
\Pr[\freq_{\bar{C}_1^n}(\bar{c}) > q_\mathrm{hon}(\bar{c}) + \dupp_{\bar{c}}] \leq \frac{\ecomAT}{6} 
,
\end{align}
since by the union bound, the probability of violating one or more of the inequalities is upper bounded by the sum of the individual probabilities of violating each one.
(It would of course be possible to ``distribute'' $\ecomAT$ in some other fashion across the terms; however, we found heuristically that this appears to give better performance as compared to e.g.~distributing it such that all the values $\dlow_{\bar{c}},\dupp_{\bar{c}}$ are equal.)
Since the honest behavior is IID, the probabilities in~\eqref{eq:probscom} can be written in terms of the CDF of a binomial distribution, which allows us to use the inverse CDF (available in most computational software) to solve for $\dlow_{\bar{c}},\dupp_{\bar{c}}$ in terms of $\ecomAT$. 

\begin{rem}\label{remark:AT}
This choice of accept condition differs slightly from~\cite{NDN+22}, which used an accept condition with only a one-sided bound on $q(1)$. We chose to use the form presented here because in some cases it seems to improve the key rates from the REAT (albeit usually only by a small amount), and also because a lower bound on $q(\perp)$ in the accept condition is needed to apply an improved chain rule we use later (in the first line of Eq.~\eqref{eq:DIQKDchain} below).

Furthermore, the accept condition in~\cite{NDN+22} was based on a $3$-standard-deviation ``tolerance'', rather than exactly computing the CDF of a binomial distribution to achieve a desired $\ecomAT$. For this work we choose to conservatively match this by setting $\ecomAT = 10^{-3}$, since a (one-sided) $3$-standard-deviation fluctuation in a normal distribution occurs with probability $1.35 \times 10^{-3} > 10^{-3}$.
\end{rem}

As for $\ecomEV$, we first observe that error verification can only abort if Bob's guess for $A_1^n$ is wrong, thus any upper bound on the probability of the latter (under the honest behavior) is a valid choice of $\ecomEV$. We then note that a specialized error correction procedure was developed in~\cite{NDN+22} with the following properties: for an IID honest behavior of the described form, an error-correction string of length
\begin{align}\label{eq:lengthEC}
\lenEC = n\left((1-\gamma) \binh\!\left(\qberhon\right) + \gamma \binh\!\left(1-\wCHSHhon\right) \right) + 50\sqrt{n}
\end{align}
suffices to ensure that Bob's guess is correct with probability over $99.9\%$, as estimated by simulations. (While this value is a somewhat heuristic estimate, recall that it only affects the probability that the honest behavior aborts, not any of the security properties of the protocol.) Hence performing error correction according to this procedure suffices to heuristically ensure $\ecomEV \leq 10^{-3}$. 

\subsection{Correctness}
\label{subsec:correctness}

In~\cite{NDN+22}, error verification was performed using a $\delta$-almost-universal hash with $\delta=2^{-61}$ and $\lenEV=64$ (under the condition that the message length in bits is at most $2^{64} \approx 10^{19}$, which is indeed the case here). We leave this aspect entirely unchanged, which suffices to ensure a correctness parameter of $\ecorr=2^{-61}$ as proven in~\cite{NDN+22}. 

\subsection{Secrecy}

This is the part of our analysis that differs the most from~\cite{NDN+22}, in that apart from improving the entropy accumulation bound, we simplify or improve a number of other steps in the analysis. We shall show that to achieve a desired secrecy parameter $\esecret$, it suffices to take the length of the final key to be 
\begin{align}\label{eq:lkey}
\lkey = n h_{\alpha} - n\left(\gamma + \dlow_{\perp} \right) - \lenEC - \lenEV - \frac{\alpha}{\alpha-1}\log\frac{1}{\esecret} + 2, 
\end{align}
where $h_{\alpha}$ is computed in terms of the $S_\mathrm{acc}$ choice defined in Sec.~\ref{subsec:completeness}, while $\lenEC$ and $\lenEV$ are as described in~\eqref{eq:lengthEC} and Sec.~\ref{subsec:correctness} respectively. Note that to evaluate $h_{\alpha}$, we used generic heuristic numerical methods rather than a convex solver that returns explicit dual bounds, because our bounds on the {\Renyi} entropy do not fall within the standard disciplined-convex-programming ruleset for such solvers. However, as the optimization is convex with respect to each of the individual variables, every local minimum is a global minimum and hence we believe it is unlikely that the resulting value we obtain for $h_{\alpha}$ is a significant overestimate of its true value.

\begin{rem} \label{remark:PA}
In order for the following analysis to hold, we currently require an implementation difference between the protocol described here and the protocol in~\cite{NDN+22}, in that privacy amplification would have to be performed using 2-universal hashing~\cite{Dup23} rather than Trevisan's extractor~\cite{DPV+12,MPS12,NDN+22} as used in that work. This is because a {\Renyi} privacy amplification theorem has currently only been proven for the former, not the latter. However, it seems likely that it should be possible to obtain such a result for the latter, and it would be a useful question to investigate in future work.
\end{rem}

Let $\OAT$ denote the event that the protocol accepts during the acceptance test, and let $\OEV$ denote the event that it accepts during error verification (so the event of the protocol accepting overall is $\OAT \land \OEV$. By applying the REAT (specifically~\cite[Lemma~5.1 with Lemma~6.1]{arx_AHT24}), the state conditioned on $\OAT$ satisfies
\begin{align}  
\widetilde{H}_{\alpha}^{\uparrow}\left({A}_{1}^{n}\bar{{C}}_{1}^{n}|
{X}_{1}^{n}{Y}_{1}^{n}{T}_{1}^{n}\mathbf{E}\right)_{\rho_{|\OAT}} \geq n h_{\alpha} 
- \frac{\alpha}{\alpha-1} \log\frac{1}{\Pr[\OAT]} \, .
\end{align}
However, in order to apply the relevant privacy amplification theorem from~\cite{Dup23}, we would need to account for various other publicly announced registers, such as $\bar{{B}}_{1}^{n}$, $\ECstring$, $\EVstring$ and $\HEV$, as well as further conditioning on $\OEV$. (In fact the REAT by itself technically allows us to directly condition on $\OAT \land \OEV$ instead; however, doing so in this proof would obstruct the last line in~\eqref{eq:removeECEV} later where we need to ``factor off'' $\HEV$.)
% \footnote{Technically, Alice also announces the hash function drawn from the $\delta$-almost-universal family, but this can be accounted for by noting that it is drawn independently of the other registers and thus does not reduce the entropy.}
Furthermore, since Alice performs privacy amplification only on ${A}_{1}^{n}$, the relevant {\Renyi} entropy should only have ${A}_{1}^{n}$ (not $\bar{{C}}_{1}^{n}$) on the left side of the conditioning. 

To address these points, we first handle the  conditioning on $\OEV$, and remove the error-correction and error-verification registers from the conditioning:
\begin{align}  
&\widetilde{H}_{\alpha}^{\uparrow}\left({A}_{1}^{n}|\bar{{B}}_{1}^{n} \ECstring \EVstring \HEV
{X}_{1}^{n}{Y}_{1}^{n}{T}_{1}^{n}\mathbf{E}\right)_{\rho_{|\OAT\land\OEV}} \nonumber\\ 
\geq& 
\widetilde{H}_{\alpha}^{\uparrow}\left({A}_{1}^{n}|\bar{{B}}_{1}^{n} \ECstring \EVstring \HEV
{X}_{1}^{n}{Y}_{1}^{n}{T}_{1}^{n}\mathbf{E}\right)_{\rho_{|\OAT}} 
-\frac{\alpha}{\alpha-1}\log\frac{1}{\Pr[\OEV|\OAT]} \nonumber\\
\geq& 
\widetilde{H}_{\alpha}^{\uparrow}\left({A}_{1}^{n}|\bar{{B}}_{1}^{n} \HEV
{X}_{1}^{n}{Y}_{1}^{n}{T}_{1}^{n}\mathbf{E}\right)_{\rho_{|\OAT}} - \lenEC - \lenEV
-\frac{\alpha}{\alpha-1}\log\frac{1}{\Pr[\OEV|\OAT]} \nonumber\\
=& \widetilde{H}_{\alpha}^{\uparrow}\left({A}_{1}^{n}|\bar{{B}}_{1}^{n}
{X}_{1}^{n}{Y}_{1}^{n}{T}_{1}^{n}\mathbf{E}\right)_{\rho_{|\OAT}} - \lenEC - \lenEV
-\frac{\alpha}{\alpha-1}\log\frac{1}{\Pr[\OEV|\OAT]} 
\, .\label{eq:removeECEV}
\end{align}
where the second line is~\cite[Lemma~B.5]{DFR20}, the third line is a standard chain rule\footnote{Specifically, \cite[Prop.~8]{MDS+13} together with the fact that classical registers have non-negative contributions to entropy~\cite[Lemma~5.3]{Tomamichel2015QuantumIP}.} for classical conditioning registers, and the fourth line holds because ${\rho_{{A}_{1}^{n}\bar{{B}}_{1}^{n}\HEV
{X}_{1}^{n}{Y}_{1}^{n}{T}_{1}^{n}\mathbf{E}}}$ can be viewed as the state immediately after the choice of hash function $\HEV$ was drawn, in which case $\HEV$ is independent of all other registers (even conditioned on $\OAT$) due to how it was generated.

Next, we relate this to $\widetilde{H}_{\alpha}^{\uparrow}\left({A}_{1}^{n}\bar{{C}}_{1}^{n}|
{X}_{1}^{n}{Y}_{1}^{n}{T}_{1}^{n}\mathbf{E}\right)_{\rho_{|\OAT}}$ following the approach in~\cite{arx_AHT24}: observe that 
\begin{align}
\widetilde{H}_{\alpha}^{\uparrow}\left({A}_{1}^{n}|\bar{{B}}_{1}^{n}{X}_{1}^{n}{Y}_{1}^{n}{T}_{1}^{n}\mathbf{E}\right)_{\rho_{|\OAT}} 
&\geq \widetilde{H}_{\alpha}^{\uparrow}\left({A}_{1}^{n}\bar{{B}}_{1}^{n}|
{X}_{1}^{n}{Y}_{1}^{n}{T}_{1}^{n}\mathbf{E}\right)_{\rho_{|\OAT}} - n (\gamma+\dlow_{\perp}) \log\dim(\bar{B}_i) \nonumber\\
&= \widetilde{H}_{\alpha}^{\uparrow}\left({A}_{1}^{n}\bar{{B}}_{1}^{n}\bar{C}_{1}^{n}|
{X}_{1}^{n}{Y}_{1}^{n}{T}_{1}^{n}\mathbf{E}\right)_{\rho_{|\OAT}} - n (\gamma+\dlow_{\perp}) \log\dim(\bar{B}_i) \nonumber\\
&\geq \widetilde{H}_{\alpha}^{\uparrow}\left({A}_{1}^{n}\bar{C}_{1}^{n}|
{X}_{1}^{n}{Y}_{1}^{n}{T}_{1}^{n}\mathbf{E}\right)_{\rho_{|\OAT}} - n (\gamma+\dlow_{\perp}) \log\dim(\bar{B}_i) \, ,
\label{eq:DIQKDchain}
\end{align}
where the first line is proven in~\cite[Remark~8.1]{arx_AHT24} (noting that the number of test rounds conditioned on $\OAT$ is at most $\gamma+\dlow_{\perp}$),
the second line holds because $\bar{{C}}_{1}^{n}$ can be ``projectively reconstructed'' from ${A}_{1}^{n}{B}_{1}^{n}
{X}_{1}^{n}{Y}_{1}^{n}{T}_{1}^{n}$ in the sense described in~\cite[Lemma~B.7]{DFR20}, and the last line holds because classical registers have non-negative contributions to entropy~\cite[Lemma~5.3]{Tomamichel2015QuantumIP} (this last step is not necessary in general, but we employ it here since our analytic bounds are only for the entropy of Alice's output, not Bob's). Putting all the above bounds together, we conclude (since $\dim(\bar{B}_i)=2$ and $\Pr[\OAT]\Pr[\OEV|\OAT] = \Pr[\OEV \land \OAT]$):
\begin{align}  
\widetilde{H}_{\alpha}^{\uparrow}\left({A}_{1}^{n}|\bar{{B}}_{1}^{n} \ECstring \EVstring
{X}_{1}^{n}{Y}_{1}^{n}{T}_{1}^{n}\mathbf{E}\right)_{\rho_{|\OAT\land\OEV}} \geq 
n h_{\alpha} - n(\gamma+\dlow_{\perp}) - \lenEC - \lenEV
- \frac{\alpha}{\alpha-1} \log\frac{1}{\Pr[\OEV \land \OAT]}
\, . \label{eq:finalentbnd}
\end{align}

With this, we note that if we let $K_A$ denote Alice's final key and $\mathbf{E}_\mathrm{fin}$ denote Eve's final side-information after privacy amplification, then we have 
\begin{align}  
&\Pr[\OEV \land \OAT] \, d\left(\rho_{K_A\mathbf{E}_\mathrm{fin}|\OEV \land \OAT} , \frac{\id_{K_A}}{|K_A|} \otimes \rho_{\mathbf{E}_\mathrm{fin}|\OEV \land \OAT}\right) 
\nonumber\\
\leq& \Pr[\OEV \land \OAT]  2^{\frac{2}{\alpha}-2} 2^{\frac{\alpha-1}{\alpha}\left(\lkey - \widetilde{H}_{\alpha}^{\uparrow}\left({A}_{1}^{n}|\bar{{B}}_{1}^{n} \ECstring \EVstring \HEV
{X}_{1}^{n}{Y}_{1}^{n}{T}_{1}^{n}\mathbf{E}\right)_{\rho_{|\OEV \land \OAT}} \right)} \nonumber\\
=& \Pr[\OEV \land \OAT]  2^{\frac{\alpha-1}{\alpha}\left(\lkey - \widetilde{H}_{\alpha}^{\uparrow}\left({A}_{1}^{n}|\bar{{B}}_{1}^{n} \ECstring \EVstring \HEV
{X}_{1}^{n}{Y}_{1}^{n}{T}_{1}^{n}\mathbf{E}\right)_{\rho_{|\OEV \land \OAT}} - 2 \right)} \nonumber\\
\leq& \Pr[\OEV \land \OAT]  2^{\log\frac{1}{\Pr[\OEV \land \OAT] } - \log\frac{1}{\esecret}} \nonumber\\
=& \esecret
\, ,
\end{align}
where the second line is the {\Renyi} privacy amplification theorem in~\cite[Theorem~9 with Lemma~7]{Dup23} (noting that the $1$-norm distance and trace distance differ by a factor of $1/2$), the third line simply regroups the exponents, and the fourth line follows from combining~\eqref{eq:lkey} with \eqref{eq:finalentbnd}.
This fulfills the definition of $\esecret$-secrecy as described in e.g.~\cite{Arnon-Friedman:2018aa,TSB+22}.

\begin{rem} \label{remark:diffs}
Comparing the above analysis to that in~\cite{NDN+22}, apart from the main change of the former being based on {\Renyi} entropies\footnote{Refer to~\cite[Fig.~1]{hahn2024boundspetzrenyidivergencesapplications} for an analysis of the effect of making only this change, without any of the other improvements we employ here.} (which also simplified some points regarding event conditioning), the other notable difference is the chain rule used above to obtain the first line of~\eqref{eq:DIQKDchain}. We believe it is likely to yield better results than the chain rules used in~\cite{NDN+22} to handle the $\bar{B}_1^n$ registers. Moreover, this chain rule also places $\bar{B}_1^n$ directly in the conditioning registers, which allowed us to modify the protocol such that Bob publicly announces the $\bar{B}_1^n$ registers for Alice to compute $\bar{C}_1^n$ --- this simplifies the analysis compared to~\cite{NDN+22}, where instead Bob computes $\bar{C}_1^n$ using his guess for Alice's $A_1^n$ registers, and extra steps had to be taken in that proof to accommodate the possibility that his guess could be wrong. 
\end{rem}

\subsection{Possible modifications}
\label{app:protmods}

Finally, we make some informal comments regarding some potential for slightly sharpening the above analysis. Namely, the way we computed the lower bound on $h_{\alpha}$ is slightly suboptimal, in that we were effectively ``sacrificing'' entropy contributions from test rounds (in that comparing to the REAT statement in~\cite[Lemma~5.1]{arx_AHT24}, we have handled the entropy contributions from terms with $\bar{c}\neq\perp$ by trivially lower bounding them with zero). 
In contrast, the earlier security proofs in e.g.~\cite{Arnon-Friedman:2018aa,NDN+22} using previous EAT versions (based on von Neumann entropy in single rounds) were able to incorporate the entropy contributions from those rounds, due to certain properties of von Neumann entropy that did not carry over to the REAT. Due to this, we found that at the larger $n$ values studied in Fig.~\ref{subfig:largen} (where the rates are closer to the asymptotic values), the REAT-based approach here gives worse rates than those in~\cite{NDN+22} when using the same value of $\gamma$ --- to obtain the improved rates in that figure, we instead had to optimize the choice of $\gamma$, using a smaller value that ``sacrifices'' less of the key rate to the test-round component.

From a theoretical point of view, one way to overcome this drawback would be to use~\cite[Theorem~5.1]{arx_AHT24} rather than~\cite[Lemma~5.1]{arx_AHT24}, in that it ``retains'' entropy contributions from test rounds. However, the former involves some slightly elaborate {\Renyi} divergence terms that do not seem straightforward to analyze using the methods in this work. It would be an interesting task for future work to generalize these methods to handle those {\Renyi} divergence terms. Another prospect could be to note that~\cite[Lemma~5.1]{arx_AHT24} itself technically involves {\Renyi} entropy terms conditioned on each value of $\bar{c}$ that could be individually analyzed (here we have basically only retained the $\bar{c}=\perp$ term); however, that approach seems less promising for future use, because those terms would all be zero if the protocol is one where $\bar{c}$ contains the full input-output values in test rounds.

Alternatively, from a practical perspective it might seem expedient to address this by having Alice perform privacy amplification on only the generation-round data (since the test-round entropy is anyway ``sacrificed''), in which case she would not need to include test-round data in the error correction information (as implicitly accounted for in~\eqref{eq:lengthEC}) --- this would reduce the value of $\lenEC$ accordingly and also make that step practically easier to implement. 
However, due to the structure of the above proof (mainly the use of the ``projective reconstruction'' property from~\cite[Lemma~B.7]{DFR20} in the second line of~\eqref{eq:DIQKDchain}, which requires Alice's test-round outputs to appear at some point in the entropy terms), it does not seem entirely straightforward to ``directly'' get a bound for the final entropy of Alice's generation-round registers only. 

We observe that technically, one way to obtain a bound would be as follows: define registers $\bar{A}_i$ that are equal to $A_i$ in test rounds and set to $0$ in generation rounds, and define registers $\hat{A}_i$ that are equal to $A_i$ in generation rounds and set to $0$ in test rounds. Then
\begin{align}
\widetilde{H}_{\alpha}^{\uparrow}\left(\hat{A}_{1}^{n}|\bar{{B}}_{1}^{n}{X}_{1}^{n}{Y}_{1}^{n}{T}_{1}^{n}\mathbf{E}\right)_{\rho_{|\OAT}} 
&\geq \widetilde{H}_{\alpha}^{\uparrow}\left(\hat{A}_{1}^{n}|\bar{A}_{1}^{n}\bar{{B}}_{1}^{n}{X}_{1}^{n}{Y}_{1}^{n}{T}_{1}^{n}\mathbf{E}\right)_{\rho_{|\OAT}} \nonumber\\
&\geq \widetilde{H}_{\alpha}^{\uparrow}\left(\hat{A}_{1}^{n}\bar{A}_{1}^{n}\bar{{B}}_{1}^{n}|
{X}_{1}^{n}{Y}_{1}^{n}{T}_{1}^{n}\mathbf{E}\right)_{\rho_{|\OAT}} - n (\gamma+\dlow_{\perp}) \log\dim(\bar{A}_i\bar{B}_i) \nonumber\\
&= \widetilde{H}_{\alpha}^{\uparrow}\left({A}_{1}^{n}\bar{{B}}_{1}^{n}|
{X}_{1}^{n}{Y}_{1}^{n}{T}_{1}^{n}\mathbf{E}\right)_{\rho_{|\OAT}} - n (\gamma+\dlow_{\perp}) \log\dim(\bar{A}_i\bar{B}_i) 
\, ,
\end{align}
where the second line is again via~\cite[Remark~8.1]{arx_AHT24}, and the third line is again via~\cite[Lemma~B.7]{DFR20}, observing that ${A}_{1}^{n}$ can be ``projectively reconstructed'' from $\hat{A}_{1}^{n}\bar{A}_{1}^{n}{T}_{1}^{n}$ (and vice versa, $\hat{A}_{1}^{n}\bar{A}_{1}^{n}$ can also be ``projectively reconstructed'' from ${A}_{1}^{n}{T}_{1}^{n}$).

With this we can continue on exactly the same way as in the second and third lines of~\eqref{eq:DIQKDchain}. However, observe that this results in $2n\left(\gamma + \dlow_{\perp} \right)$ in place of $n\left(\gamma + \dlow_{\perp} \right)$ in the final key length formula. In order to obtain an overall benefit from this approach, we would need a more detailed analysis of how much the $O(\sqrt{n})$ term in $\lenEC$ (Eq.~\eqref{eq:lengthEC}) can be improved by not having to include the test-round data, which is a somewhat more specialized coding-theory question that we shall not consider within this work. (Furthermore, we informally note that this proposed approach still ends up reducing the analysis to $\widetilde{H}_{\alpha}^{\uparrow}\left({A}_{1}^{n}\bar{{B}}_{1}^{n}|
{X}_{1}^{n}{Y}_{1}^{n}{T}_{1}^{n}\mathbf{E}\right)_{\rho_{|\OAT}}$, which again includes Alice's test-round data and hence does not really ``exploit'' the fact that our single-round analysis excludes the entropy contributions from those terms.)

\end{document}